\documentclass[11pt,twoside]{article}
\pdfoutput=1

\usepackage{amsfonts}
\usepackage{amsmath}
\usepackage{amsthm}
\usepackage{amscd}
\usepackage{eucal}
\usepackage{dsfont}
\usepackage{bm}
\usepackage{graphicx}

\usepackage{a4}
\usepackage{times}
\usepackage{booktabs}
\usepackage{cite}
\usepackage{microtype}
\usepackage[title]{appendix}
\def\beq{\begin{equation}}
\def\eeq{\end{equation}}
\def\bea{\begin{eqnarray}}
\def\eea{\end{eqnarray}}

\def\beann{\begin{eqnarray*}}
\def\eeann{\end{eqnarray*}}

\let\a=\alpha  \let\g=\gamma \let\de=\delta
\let\e=\varepsilon \let\z=\zeta  \let\th=\theta

\let\dh=\vartheta \let\k=\kappa \let\la=\lambda \let\m=\mu
\let\n=\nu \let\x=\xi \let\p=\pi \let\r=\rho \let\s=\sigma
\let\om=\omega \let\ps=\psi
\let\ph=\varphi \let\Ph=\phi \let\PH=\Phi \let\Ps=\Psi
\let\ups=\Upsilon
\let\Om=\Omega \let\Si=\Sigma 
\let\La=\Lambda \let\G=\Gamma \let\D=\Delta

\let\qd=\quad  

\def\epp{\, .}
\def\epc{\, ,}

\def\tst#1{{\textstyle #1}}

\theoremstyle{plain}
\newtheorem{theorem}{Theorem}
\newtheorem*{theorem*}{Theorem}
\newtheorem{lemma}{Lemma}
\newtheorem*{lemma*}{Lemma}
\newtheorem{proposition}{Proposition}
\newtheorem{corollary}{Corollary}
\newtheorem*{corollary*}{Corollary}
\newtheorem{conjecture}{Conjecture}
\newtheorem*{conjecture*}{Conjecture}

\theoremstyle{definition}

\newtheorem*{remark}{Remark}
\newtheorem*{question*}{Question}
\newtheorem{Remark}{Remark}

\def\2{\frac{1}{2}} \def\4{\frac{1}{4}}

\def\6{\partial}

\def\+{\dagger}

\def\<{\langle} \def\>{\rangle}

\let\auf=\uparrow \let\ab=\downarrow

 \def\CL{{\cal L}}
\def\CO{{\cal O}}

\let\nodoti\i
\def\i{{\rm i}}

\def\rd{{\rm d}}

\def\tg{\, {\rm tg}\,}
\def\ctg{\, {\rm ctg}\,}

\DeclareMathOperator{\re}{e}
\DeclareMathOperator{\Ln}{Ln}
\DeclareMathOperator{\sh}{sh}
\DeclareMathOperator{\ch}{ch}

\DeclareMathOperator{\tr}{tr}

\DeclareMathOperator{\one}{\mathds{1}}
\DeclareMathOperator{\Int}{Int}
\DeclareMathOperator{\Ext}{Ext}

\DeclareMathOperator{\sign}{sign}
\DeclareMathOperator{\End}{End}
\DeclareMathOperator{\id}{id}

\DeclareMathOperator{\card}{card}

\def\Re{{\rm Re\,}} \def\Im{{\rm Im\,}}

\def\ev{\mathbf{e}}

\def\uv{\mathbf{u}}
\def\vv{\mathbf{v}}

\def\xv{\mathbf{x}}
\def\yv{\mathbf{y}}

\def\Uv{\mathbf{U}}
\def\Vv{\mathbf{V}}

\def\fa{\mathfrak{a}}

\def\fe{\mathfrak{e}}

\DeclareMathOperator{\fz}{\mathfrak{z}}

\def\ks{h_R}

\renewcommand{\ks}{\k}
\newcommand{\nex}{\ell}

\allowdisplaybreaks

\begin{document}

\thispagestyle{empty}

\begin{center}

{\Large \bf
A thermal form factor series for the longitudinal two-point function of
the Heisenberg-Ising chain in the antiferromagnetic massive regime}

\vspace{10mm}

{\large
Constantin Babenko,$^\dagger$
Frank G\"{o}hmann,$^\dagger$
Karol K. Kozlowski$^\ast$ and Junji Suzuki$^\ddagger$}\\[3.5ex]
$^\dagger$Fakult\"at f\"ur Mathematik und Naturwissenschaften,\\
Bergische Universit\"at Wuppertal,
42097 Wuppertal, Germany\\[1.0ex]
$^\ast$Univ Lyon, ENS de Lyon, Univ Claude Bernard,\\ CNRS,
Laboratoire de Physique, F-69342 Lyon, France\\[1.0ex]
$^\ddagger$Department of Physics, Faculty of Science, Shizuoka University,\\
Ohya 836, Suruga, Shizuoka, Japan

%\vspace{40mm}

\vspace{17mm}
%Resubmitted March 8, 2021 %\today
{\it Dedicated to Professor Barry M. McCoy on the occasion of his 80th birthday}
\vspace{19mm}

{\large {\bf Abstract}}

\end{center}

\begin{list}{}{\addtolength{\rightmargin}{9mm}
               \addtolength{\topsep}{-5mm}}
\item
We consider the longitudinal dynamical two-point function of the
XXZ quantum spin chain in the antiferromagnetic massive
regime. It has a series representation based on the form factors
of the quantum transfer matrix of the model. The $n$th summand of 
the series is a multiple integral accounting for all $n$-particle
$n$-hole excitations of the quantum transfer matrix. In previous
works the expressions for the form factor amplitudes appearing
under the integrals were either again represented as multiple
integrals or in terms of Fredholm determinants. Here we obtain
a representation which reduces, in the zero-temperature limit,
essentially to a product of two determinants of finite matrices
whose entries are known special functions. This will facilitate
the further analysis of the correlation function.
%\\[2ex]
%{\it PACS: 05.30.-d, 75.10.Pq}
\end{list}

\clearpage

\section{Introduction}
In this work we renew our attempts to obtain simple and manageable
expressions for the dynamical correlation functions of the XXZ
chain. The XXZ chain is an anisotropic deformation of the Heisenberg
chain which is the fundamental model of 1d magnetism. The XXZ
Hamiltonian for a chain of length $L$ acts on the tensor product
${\cal H}_L = \bigotimes_{j=1}^L V_j$, $V_j = {\mathbb C}^2$, in which
every factor is identified with a site in a 1d lattice. It is defined by
\begin{equation} \label{hxxz}
     H_L = J \sum_{j = 1}^L \Bigl\{ \s_{j-1}^x \s_j^x + \s_{j-1}^y \s_j^y
                 + \D \bigl( \s_{j-1}^z \s_j^z - 1 \bigr) \Bigr\}
		 - \frac{h}{2} \sum_{j=1}^L \s_j^z \epc
\end{equation}
where $\s^\a \in \End {\mathbb C}^2$, $\a = x, y, z$, are the
Pauli matrices. The three real parameters of the Hamiltonian are
the anisotropy $\D$, the exchange interaction $J > 0$, and the
strength $h > 0$ of an external magnetic field in the direction of
the magnetic symmetry axis.

If the magnet is in contact with a heat bath of temperature $T$, it
is in a `mixed state' with canonical density matrix
\begin{equation}
     \r_L (T)  = \frac{\re^{- H_L/T}}{\tr \{\re^{- H_L/T}\}} \epp
\end{equation}
The Heisenberg time evolution of a local operator $x_j \in
\End V_j$ is defined by
\begin{equation}
     x_j \mapsto x_j (t) = \re^{\i H_L t} x_j \re^{- \i H_L t} \epp
\end{equation}
A typical quantity measured in experiments on quasi-1d magnets is the
dynamical two-point correlation function
\begin{equation} \label{twopointfun}
     \bigl\<x_j y_k (t)\bigr\>_T
        = \lim_{L \rightarrow \infty}
	  \tr \bigl\{\r_L (T) x_j y_k (t)\bigr\} \epp
\end{equation}
of two local operators $x_j$, $y_k$.

The analysis of finite-temperature dynamical correlation functions
such as (\ref{twopointfun}) is rather involved. Even for integrable
lattice models like the XXZ chain very little is known in the
general finite temperature case. Notable exceptions are so-called
free-fermion models as, for instance, the XX chain (Hamiltonian (\ref{hxxz})
with $\D = 0$) \cite{IIKS93b,CIKT92,Jie98,GKSS19,GKS20a,GKS20b} or
the transverse field Ising model \cite{GFE20,Iorgov11,IoLi11}, where
much of the long-time large-distance asymptotics was worked out and
numerically efficient representations of the correlation functions
are available.

In the low-temperature limit, $T \rightarrow 0$, when $\r_L (T)$
becomes proportional to the projector onto the groundstate subspace,
the dynamical correlation functions of the XXZ chain are better
understood. They were mainly analysed by means of form factor
series expansions based on matrix elements of local operators between
the ground state and the excited states of the Hamiltonian. The
ground state phase diagram is depicted in Figure~\ref{fig:phasediagram}.
\begin{figure}
\begin{center}
\includegraphics[width=.80\textwidth,angle=0,clip=true]{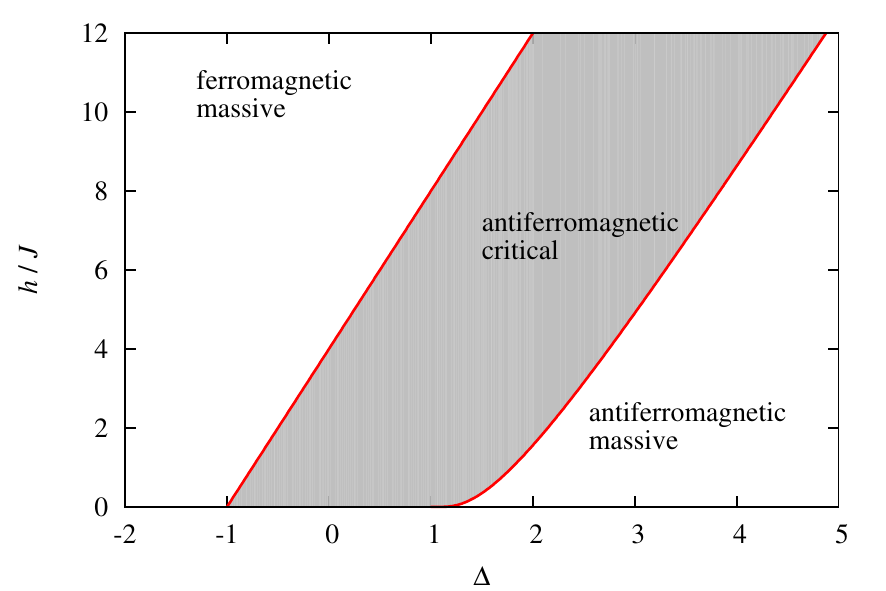}
\caption{\label{fig:phasediagram} The ground state phase diagram of
the XXZ chain in the $\D$-$h/J$ plane. The ferromagnetic massive regime
and the antiferromagnetic critical regime are separated by the upper
critical field $h_u = 4J(1 + \D)$ (left red line). The lower critical
field $h_\ell$ (right red line, defined in equation (\ref{hlow}) below)
determines the border between the antiferromagnetic critical and the
antiferromagnetic massive regime.
}
\end{center}
\end{figure}%
The analysis focused on the antiferromagnetic regions of the phase
diagram. Form factors for the antiferromagnetic massive regime
were first obtained within the $q$-vertex operator approach
by Jimbo and Miwa \cite{JiMi95}. In this approach the form factors
are obtained in the form of multiple integrals which are interpreted
in terms of even-numbered multiple-spinon excitations and are called
the $2n$-spinon form factors. Since only the two-spinon form
factors are known explicitly in terms of special functions, most
of the subsequent analysis focused on the two-spinon contribution
to the dynamical correlation functions. Moreover, the emphasis
was on the calculation of the spatial and temporal Fourier transforms of
the dynamical two-point functions, the so-called `dynamical structure
factors,' since these are the functions that are measured more or
less directly in neutron scattering experiments on certain quasi
one-dimensional magnets. The two-spinon dynamical structure factor
of the XXZ chain in the antiferromagnetic massive regime was studied
in \cite{BCK96,BKM98,CMP08,Castillo20pp}. The important limiting case
of the XXX chain (Hamiltonian (\ref{hxxz}) with $\D = 1$) was considered
in \cite{KMBFM97}. The four-spinon contribution for the XXX model was
calculated in \cite{CaHa06}. Two- and four spinon contributions
together seem to give a rather quantitative description of the available
experimental data \cite{MEKCSR13}. The space-time asymptotics of the
longitudinal dynamical two-point functions of the XXZ chain in the
antiferromagnetic massive regime was analysed in \cite{DGKS16a}.

The $q$-vertex operator approach is designed to work directly in
the thermodynamic limit. Unfortunately, so far it has been successful
only for massive integrable models and only for the case of vanishing
magnetic field. The two-spinon contribution to the dynamical structure
factor of the XXZ chain in the antiferromagnetic massless regime at
vanishing magnetic field was obtained in \cite{CKSW12} using expressions
for the two-spinon form factors of the XYZ chain \cite{Lashkevich02}
and performing a massless limit. For the exploration of higher-spinon
contributions or of the antiferromagnetic massless phase at non-zero
magnetic field the available results on dynamical correlation functions
rely on Bethe Ansatz techniques.

An efficient Bethe Ansatz approach to the calculation of form factors
of the finite XXZ chain was developed in \cite{KMT99a}. Such an
approach has the advantage that is applies equally well to all points
in the ground state phase diagram. However, the calculation of
the thermodynamic limit of the form factors in the general
situation is rather sophisticated \cite{Slavnov90,IKMT99,KKMST09b,%
KKMST11a,DGKS15a,Kozlowski17,KiKu19}. The summation of the form
factor series imposed additional difficulties that were overcome in
\cite{KKMST11b,KKMST12,DGKS15a,Kozlowski18}.

In \cite{KMT99a} the form factors pertaining to the finite XXZ chain
of even length $L$ were obtained in a form proportional to determinants
of matrices of size $N \times N$ with $N$ growing in $L$. These were
analysed quite successfully by solving the Bethe Ansatz equations numerically
\cite{BKM02a} and resulted in predictions for the dynamical structure
factors in finite magnetic fields \cite{SST04,CaMa05} that could be
compared with experiments. In the thermodynamic limit, $L \rightarrow \infty$,
the $N \times N$ determinants, generically, turn into Fredholm determinants
of integral operators \cite{KKMST09b,KKMST11a}. For a long time the
Fredholm determinant corresponding to Baxter's staggered polarisation
\cite{Baxter73a,Baxter76a}, i.e.\ the Fredholm determinant arising
in the calculation of the matrix element of $\s^z$ between the two
asymptotically degenerate groundstates of the XXZ chain in the
antiferromagnetic massive regime of the XXZ chain, was the only
known example that could be evaluated more explicitly in terms of
known special functions \cite{IKMT99}. Only recently the well-known
explicit result for the two-spinon form factors of the XXX chain
\cite{JiMi95} was reproduced from an algebraic Bethe Ansatz perspective
\cite{KiKu19}.

Altogether the thermodynamic limit of the form factors generally
results in rather technical expressions which can be seen as
multiple residues related to the zeros of the counting functions
of the massive excitations in the respective ground state regimes
\cite{DGKS15a,Kozlowski18}. In spite of the technically complicated
nature of the expressions for the form factors, the series derived
in \cite{Kozlowski18} could be used for a detailed analysis of
the space-time asymptotics of the dynamical two-point functions
of the XXZ chain in the antiferromagnetic massless regime at finite
magnetic field \cite{Kozlowski19} and of the threshold singularities
of their Fourier transforms \cite{Kozlowski18pp}. This analysis
confirmed and refined, from a microscopic point of view, the
non-linear Luttinger liquid phenomenology (for a review see
\cite{ISG12}).

An alternative route to the zero temperature correlation functions
of the XXZ chain is through the thermal form factor expansion 
\cite{DGK13a,DGKS16b} developed initially for static correlation
functions and later extended to the dynamical case in \cite{GKKKS17}.
The thermal form factor expansion is an expansion involving the
form factors and eigenvalues of the quantum transfer matrix. The
latter is an auxiliary object originally introduced
\cite{Suzuki85,SAW90,Kluemper93} in order to study the thermodynamics
of quantum spin chains. The quantum transfer matrix is generically
non-Hermitian. Its eigenvalues are not necessarily all real.
The eigenvalue of largest modulus, however, is always real and
non-degenerate (which was rigorously established for temperature
high enough in \cite{GGKS20}). We call it the dominant eigenvalue
and the corresponding eigenvector the dominant eigenvector.
The dominant eigenvector determines the reduced density matrices
of all finite sub-segments of the considered quantum spin chain
in the thermodynamic limit \cite{GKS04a,GKS05}. Dynamical
correlation functions can be expressed as form factor series
involving the matrix elements of the operators of the algebraic
Bethe Ansatz between the dominant state and excited states of
the quantum transfer matrix \cite{GKKKS17}.

The spectrum and the Bethe roots of the quantum transfer matrix of
an integrable quantum spin chain are rather different from the spectrum
and Bethe roots of the ordinary transfer matrix associated with the
Hamiltonian. Both depend parametrically on the temperature
and, in case of the XXZ chain, on the external magnetic field.
The Bethe roots form patterns in the complex plane which `evolve'
with increasing temperature. These patterns are currently not
sufficiently well understood to allow for an analysis of the thermal
form factor series of the dynamical two-point functions of the XXZ
chain at arbitrary temperatures. For this reason we shall focus
in this work on the case of the low-temperature limit of the
XXZ chain in the antiferromagnetic massive regime, which is the
only parameter regime so far, where we have obtained a complete
picture of all excitations and explicit expressions for all
eigenvalues \cite{DGKS15b}. Amazingly, the corresponding Bethe
root patterns are much simpler than those belonging to the
excited states of the Hamiltonian in the same regime of the
ground state phase diagram. They are characterized by the absence
of so-called strings. Accordingly, they also come with a different
interpretation. While the excitations of the Hamiltonian are
interpreted in terms of spinons (see e.g.\ \cite{DGKS15a}), the
excitations of the quantum transfer matrix are parameterised
by Bethe root patterns of particle-hole type. In the zero temperature
limit every excitation is determined by two finite sets of
complex parameters, called particle- and hole roots, and by a
`topological index' $k \in {\mathbb Z}\slash 2 {\mathbb Z}$.
The particle- and hole roots are confined on two simple finite
curves in the upper and lower half plane, respectively, which
are symmetric with respect to reflection about the real and
imaginary axis (see Figure~\ref{fig:regimes}).

In \cite{DGK13a} we found that the amplitudes in the form
factor expansions factorize in three parts which we called the
universal part, the factorizing part and the determinant part.
In the so-called Trotter limit, which is a necessary step of the formalism,
the universal and factorising parts were taking a very explicit form
given in terms of exponents of one and two dimensional integrals.
In the low-T limit, we showed that all these quantities can be
evaluated in closed form in terms of products of $q$-gamma and $q$-Barnes
functions. In its turn, the determinant part was represented as a ratio
of two Fredholm determinants in the numerator over two Fredholm
determinants in the denominator. While in the low-T limit the Fredholm
determinants in the denominator turned out to be explicitly computable
in terms of $q$-products the Fredholm determinants in the numerator
had too complicated kernels so as to go beyond the Fredholm determinant
representation. This was not much of a problem for the numerical
evaluation of the correlation functions through its thermal form
factor expansions. However, this structural intricacy of the
integrands in the form factor expansion raised doubts about the
usefulness of the representation for an analytic study of other
properties of the correlation functions, such as their
long-distance and large-time asymptotic behaviour.

By relying on a form of determinant factorisation proposed in
\cite{KiKu19}, we obtain in this work a representation for the
amplitudes arising in the thermal form factor expansion of the
generating function of the dynamical longitudinal correlation
functions. The latter is structured in such a way that, in the
low-$T$ limit, the full amplitudes become explicitly computable
in terms of products of $q$-gamma and $q$-Barnes functions and
\textit{finite size} determinants built up of $q$-gamma and theta
functions. These expressions for the form factors are then used
to express the dynamical two-point function $\<\s_1^z \s_{m+1}^z (t)\>_{T=0}$
as a novel series of multiple integrals with very explicit and
simple integrands. More precisely we shall derive the following
\begin{theorem*}
The longitudinal dynamical two-point function of the XXZ chain in
the antiferromagnetic massive regime and in the zero-temperature
limit can be expressed in terms of a generating function ${\cal G}^{(0)}$,
\begin{equation}
     \bigl\<\s_1^z \s_{m+1}^z (t)\bigr\> =
        \tst{\2} \6_{\g \a}^2 D_m^2 {\cal G}^{(0)} (m+1, t, \a)\bigr|_{\a = 0} \epp
\end{equation}
Here $\g$ parameterizes the anisotropy $\D = \ch (\g)$ and $D_m$
denotes the difference operator, formally defined by $D_m f_m =
f_m - f_{m-1}$. The generating function has a series representation
of the form
\begin{multline} % \label{genfunfinalseries}
     {\cal G}^{(0)} (m, t, \a) =
        \frac{\dh_2^2 (\i \g \a/2)}{\dh_2^2}
	+ \frac{\dh_1^2 (\i \g \a/2)}{\dh_2^2}
	  (-1)^m \\[1ex]
        + \sum_{\substack{\ell \in {\mathbb N}\\k = 0, 1}}
	  \frac{(-1)^{km}}{(\ell !)^2}
	     \int_{{\cal C}_h^\ell} \frac{\rd^\ell u}{(2\p \i)^\ell}
	     \int_{{\cal C}_p^\ell} \frac{\rd^\ell v}{(2\p \i)^\ell} \:
             {\cal A} ({\cal U},{\cal V}|k)
	     \re^{- \i \sum_{\la \in {\cal U} \ominus {\cal V}}
	     (m p(\la) + t \e (\la|h))} \epc
\end{multline}
where $\dh_1$ and $\dh_2$ are Jacobi theta functions, where $p$ and
$\e$ are the dressed momentum and dressed energy of the excitations,
and where the form factor amplitudes are given by
\begin{multline}
     {\cal A} ({\cal U},{\cal V}|k) =
        \frac{\dh_2^2 (\Si_0 ({\cal U},{\cal V}))}{\dh_2^2}
	\biggl[\prod_{\la, \m \in {\cal U}\ominus {\cal V}}
	       \mspace{-18.mu} \Ps (\la - \m)\biggr] \\ \times
        \det_{\ell} \{\Om (u_j,v_k|{\cal U}, {\cal V})\}
	\det_{\ell} \{\overline{\Om} (v_j,u_k|{\cal U}, {\cal V})\} \epp
\end{multline}
The sets ${\cal U} = \{u_j\}_{j=1}^\nex$, ${\cal V} = \{v_j\}_{j=1}^\nex$
gather all integration variables. The contours ${\cal C}_h$ and ${\cal C}_p$
are intervals $[-\p/2,\p/2]$ shifted down or up in the complex plane
(cf.\ (\ref{defchcp})). The notation $\sum_{\la \in {\cal U} \ominus {\cal V}}$
means that summands indexed by elements of ${\cal U}$ come with a
plus sign, while summands indexed by elements of ${\cal V}$ get
a minus sign (for a formal definition see (\ref{defsetsumprod}) below).
The remaining functions $\Si_0$, $\Ps$ and $\Om$ are explicitly defined
in terms of theta, $q$-gamma and $q$-Barnes functions. Their precise
definition will be given in equations (\ref{defsizero}), (\ref{defpsi})
and (\ref{defombarom}) in the main body of the text. $\Si_0$ and the
functions $\Om$ and $\overline \Om$ depend parametrically on the
`twist parameter' $\a$.
%\footnote{In the present context this parameter
%has no physical meaning. It is just the parameter with respect to which
%we have to differentiate the generating function in order to obtain the
%correlation function. The way it will be introduced into the formalism
%is that we first consider thermal form factors in which the dominant
%state of the quantum transfer matrix and an excited state pertain to
%different magnetic fields, $h$ and $h'$, say. Then for finite temperature
%$\a = (h - h')/(2 \g T)$ (cf.\ (\ref{defalpha})), which becomes
%an independent parameter in the limit $T \rightarrow 0$.}
\end{theorem*}

This work is organised as follows. In Section~2 we recall the
results of \cite{GKKKS17} on the thermal form-factor series
representation of dynamical correlation functions. Section~3
contains a review of the low-$T$ analysis of the spectrum
and Bethe root patterns of the quantum transfer matrix of the
XXZ chain in the antiferromagnetic massive regime obtained in
\cite{DGKS15b}. We also provide appropriate forms of non-linear
integral equations that will be needed in the sequel. Section~4
contains the main technical part of this work, a low-$T$ analysis
of the amplitudes in the thermal form factor series. The
new form of the amplitudes is then used in Section~5 in order
to obtain our novel form factor series expansion of the longitudinal 
two-point function. We also discuss the explicit evaluation of the
remaining finite determinants in terms of basic hypergeometric
series. Section~6 is devoted to the isotropic limit, and Section~7
to some conclusions of this work. A few supplementary technical
details are deferred to two appendices.

\section{The thermal form factor series}
In previous work \cite{GKKKS17} we have developed a thermal form
factor approach to dynamical correlation functions of Yang-Baxter
integrable models. It is based on a `vertex-model representation'
of the canonical density matrix and the time evolution operator
\cite{Sakai07}. In the following we shall apply this approach to
the longitudinal two-point function of the XXZ chain in the
antiferromagnetic massive regime in the low-temperature limit.
\subsection{The dynamical quantum transfer matrix}
The basic object in this approach is a `dynamical quantum
transfer matrix' which can be associated with any fundamental
solution of the Yang-Baxter equation. For the XXZ chain in the
antiferromagnetic massive regime we use the parametrization 
\begin{equation} \label{rmatrix}
     \begin{array}{cc}
     R(\la,\m) = \begin{pmatrix}
                  1 & 0 & 0 & 0 \\
		  0 & b(\la,\m) & c(\la,\m) & 0 \\
		  0 & c(\la,\m) & b(\la,\m) & 0 \\
		  0 & 0 & 0 & 1
		 \end{pmatrix} \epc &
     \begin{array}{c}
     b(\la, \m) = \frac{\sin(\m - \la)}{\sin(\m - \la + \i \g)} \\[2ex]
     c(\la, \m) = \frac{\sin(\i \g)}{\sin(\m - \la + \i \g)}
    \end{array}
    \end{array} \epp
\end{equation}
of the $R$-matrix. Here $\g > 0$, and we set $q = \re^{- \g}$ for
later convenience.
%\enlargethispage{2ex}

We define a staggered, twisted and inhomogeneous monodromy matrix acting
on `vertical spaces' with `site indices' $1, \dots, 2N + 2$, and on
a `horizontal auxiliary space' indexed $a$,
\begin{multline} \label{stagmon}
     T_a (\la|h) = \re^{h \s_a^z/2T} R_{2N+2, a}^{t_1} (\x_{2N+2},\la + \i \g/2)
                   R_{a, 2N+1} (\la + \i \g/2, \x_{2N+1}) \\[1ex] \times \dots \times
                   R_{2, a}^{t_1} (\x_2, \la) R_{a, 1} (\la + \i \g/2, \x_1) \epp
\end{multline}
The number $N$ will be called the Trotter number. The superscript
$t_1$ denotes transposition with respect to the first space $R$
is acting on, and $\x_1, \dots, \x_{2N+2}$ are $2N+2$ complex
`inhomogeneity parameters'. Following \cite{GKKKS17,Sakai07} we fix
these parameters to the values
\begin{equation} \label{trotterdecomp}
    \x_{2k-1} = - \x_{2k} =
                  \begin{cases}
		     - \frac{\i t}{\ks N} & k = 1, \dots, \frac{N}{2} \\[.5ex]
		     \e & k = \frac{N}{2} + 1 \\[1ex]
		     \frac{\i t + 1/T}{\ks N} & k = \frac{N}{2} + 2, \dots, N+1 \epc
                  \end{cases}
\end{equation}
where
\begin{equation} \label{tsuz}
    1/\ks = - 2 \i J \sh(\g) \epp
\end{equation}
The corresponding transfer matrix
\begin{equation} \label{defqtm}
     t (\la|h) = \tr_a \{T_a (\la|h)\}
\end{equation}
is then what we call a dynamical quantum transfer matrix of the
XXZ chain. As we have demonstrated in \cite{GKKKS17} the dynamical
two-point functions of the XXZ chain can be written as a series
involving the eigenvalues and form factors of the dynamical quantum
transfer matrix. This series will be the starting point of our
actual considerations. In order to define it we have to briefly review
the algebraic Bethe Ansatz construction of the eigenvectors and
eigenvalues of $t(\la|h)$.

\subsection{Eigenvectors and eigenvalues}
The monodromy matrix (\ref{stagmon}) fulfills the two prerequisites
for an algebraic Bethe Ansatz. It satisfies the Yang-Baxter algebra
relations
\begin{equation}
     R_{ab} (\la,\m) T_a (\la|h) T_b (\m|h) = T_b (\m|h) T_a (\la|h) R_{ab} (\la,\m)
\end{equation}
by construction, and it has a pseudo vacuum state
\begin{equation}
     |0\> = \bigl[ \tst{\binom{1}{0} \otimes \binom{0}{1}} \bigr]^{\otimes (N+1)}
\end{equation}
on which its action is triangular. More precisely, setting
\begin{equation}
     T_a (\la|h) = \begin{pmatrix}
                     A(\la|h) & B(\la|h) \\ C(\la|h) & D(\la|h)
                  \end{pmatrix}_a \epc
\end{equation}
it is not difficult to see that
\begin{equation}
     C(\la|h) |0\> = 0 \epc \qd 
     A(\la|h) |0\> = a(\la|h) |0\> \epc \qd 
     D(\la|h) |0\> = d(\la|h) |0\> \epc
\end{equation}
where $a(\la|h)$ and $d(\la|h)$ are the pseudo vacuum expectation values
\begin{align}
     a(\la|h) & = \re^\frac{h}{2T}
	       \frac{\sin\bigl(\la + \frac{\i \g}{2} + \e \bigr)}
	                   {\sin\bigl(\la + \frac{3 \i \g}{2} + \e \bigr)}
               \Biggl[
	       \frac{\sin\bigl(\la + \frac{\i \g}{2} - \frac{\i t }{\ks N} \bigr)}
	                   {\sin\bigl(\la + \frac{3 \i \g}{2} - \frac{\i t }{\ks N} \bigr)}
	       \frac{\sin\bigl(\la + \frac{\i \g}{2} + \frac{\i t + 1/T}{\ks N} \bigr)}
	                   {\sin\bigl(\la + \frac{3 \i \g}{2} + \frac{\i t + 1/T}{\ks N} \bigr)}
			    \Biggr]^\frac{N}{2}, \notag \\[1ex]
     d(\la|h) & = \re^{- \frac{h}{2T}}
	       \frac{\sin\bigl(\la + \frac{\i \g}{2} - \e \bigr)}
	                   {\sin\bigl(\la - \frac{\i \g}{2} - \e \bigr)}
               \Biggl[
	       \frac{\sin\bigl(\la + \frac{\i \g}{2} + \frac{\i t }{\ks N} \bigr)}
	                   {\sin\bigl(\la - \frac{\i \g}{2} + \frac{\i t }{\ks N} \bigr)}
	       \frac{\sin\bigl(\la + \frac{\i \g}{2} - \frac{\i t + 1/T}{\ks N} \bigr)}
	                   {\sin\bigl(\la - \frac{\i \g}{2} - \frac{\i t + 1/T}{\ks N} \bigr)}
			    \Biggr]^\frac{N}{2}.
\end{align}

The eigenvalues and eigenstates of the quantum transfer matrix can
be parameterized by sets of Bethe roots. These are defined with the
aid of an auxiliary function
\begin{equation} \label{baaux}
     \fa \bigl(\la \big| \{\la_k\}_{k=1}^M , h\bigr) =
	\frac{d(\la|h)}{a(\la|h)}
	\prod_{k=1}^M \frac{\sin(\la - \la_k - \i \g)}{\sin(\la - \la_k + \i \g)} 
\end{equation}
as the solutions of the `Bethe Ansatz equations'
\begin{equation} \label{baes}
     \fa \bigl(\la_j  \big| \{\la_k  \}_{k=1}^M, h\bigr) = - 1 \epc
        \qd j = 1, \dots, M \epp
\end{equation}
Note that the Bethe roots depend parametrically on $t$, $T$, $\g$ and $h$.
In order to simplify our notation we number the solutions consecutively
as $\{\la_j^{(n)}  \}_{j=1}^{M_n}$ and denote the corresponding auxiliary
functions by
\begin{equation}
     \fa_n (\la|h) = \fa \bigl(\la \big| \{\la_k^{(n)}  \}_{k=1}^{M_n}, h\bigr) \epp
\end{equation}

Left and right eigenvectors of the quantum transfer matrix can be
constructed from left and right `off-shell Bethe vectors' which, for
a given subset $\{\n\} = \{\n_j\}_{j=1}^M \subset {\mathbb C}$, are
defined as
\begin{equation} \label{offshellbv}
     \bigl\<\{\n\}, h \bigr| = \<0| C(\n_1 |h) \dots C(\n_M |h) \epc \qd
     \bigl|\{\n\}, h \bigr\> = B(\n_M |h) \dots B(\n_1 |h) |0\> \epp
\end{equation}
If $\{\la_j^{(n)} \}_{j=1}^{M_n}$ is a solution of the Bethe Ansatz equations
(\ref{baes}), then
\begin{equation}
    \<n, h| = \bigl\<\{\la_j^{(n)}  \}_{j=1}^{M_n}, h \bigr| \epc \qd 
    |n, h\> = \bigl|\{\la_j^{(n)}  \}_{j=1}^{M_n}, h \bigr\>
\end{equation}
are left and right eigenvectors of the quantum transfer matrix
satisfying
\begin{equation}
     \<n, h| t(\la|h) = \<n, h| \La_n (\la|h) \epc \qd
     t(\la|h) |n, h\> = \La_n (\la|h) |n, h\> \epc
\end{equation}
where
\begin{equation}
     \La_n (\la | h) =
        \bigl(1 + \fa_n (\la|h)\bigr) a(\la|h)
        \prod_{j=1}^{M_n} \frac{\sin\bigl(\la - \la_j^{(n)} + \i \g\bigr)}
	                       {\sin\bigl(\la - \la_j^{(n)}\bigr)}
\end{equation}
is the eigenvalue function.

There is a unique so-called dominant eigenvalue $\La_0$, say, which
is real and maximal in the sense that $\La_0 (- \i \g/2|h) >
|\La_n (- \i \g/2|h)|$ for all $n > 0$. We know that $M_0 = N + 1$.
To further simplify the notation we drop the index `$0$' for the
dominant state, such that $\La = \La_0$, $M = M_0$, $|h\> = |0, h\>$,
$\fa (\cdot|h) = \fa_0(\cdot|h)$, and $\{\la_j\}_{j=1}^M =
\{\la_j^{(0)} \}_{j=1}^{M_0}$.
%We choose a fixed excited state $n > 0$ with $M_n = M$ and denote its
%Bethe roots by $\{y_j\}_{j=1}^M = \{x_j^{(n)}\}_{j=1}^{M_n}$.

\subsection{The thermal form factor series}
Let
\begin{equation}
     \Si^z (\la|h) = \tr_a \{\s_a^z T_a (\la|h)\} = 2T \6_h t(\la|h)
\end{equation}
and
%\footnote{%
%As has turned to be useful in many other instances before we
%consider the dominant state and the excited states at slightly
%different magnetic fields. This regularizes many of the expressions
%we will have to deal with at intermediate stages of the calculation.
%It will also allow us to introduce a generating function of the
%longitudinal two-point function below, which will turn out to
%depend on the difference $h - h'$ of the two fields. In the end,
%when we come to physical quantities, we will have to set $h' = h$.}
\begin{equation}
     \r_n(\la|h, h') = \La_n (\la|h')/\La (\la|h) \epp
\end{equation}
Here and hereafter we keep the magnetic fields $h$ pertaining to the
dominant state and $h'$ pertaining to the $n$th excited state as independent
parameters. This is a common trick that will allow us to introduce a
generating function of the longitudinal two-point function. In the
end of our calculation in Section~\ref{sec:ffseries} we shall see
that this generating function depends only on a `twist parameter'
\begin{equation} \label{defalpha}
     \a = (h - h')/(2 \g T)
\end{equation}
which is introduced here for later convenience. It becomes an
independent variable in the zero-temperature limit.

Then, according to Theorem~1 of \cite{GKKKS17},
\begin{multline} \label{deflongffseries}
     \bigl\<\s_1^z \s_{m+1}^z (t)\bigr\>_T =
        \lim_{N \rightarrow \infty} \lim_{\e \rightarrow 0}
	\sum_n \frac{\<h|\Si^z (-\i \g/2|h)|n, h\>\<n, h|\Si^z (-\i \g/2|h)|h\>}
	            {\La_n (- \i \g/2|h) \<h|h\> \La (- \i \g/2|h) \<n, h|n, h\>} \\[-1ex]
        \times \r_n (- \i \g/2|h, h)^m
               \biggl(\frac{\r_n (- \i \g/2 + \i t/\ks N|h,h)}
	                   {\r_n (- \i \g/2 - \i t/\ks N|h,h)}\biggr)^\frac{N}{2} \epp
\end{multline}
Due to pseudo spin conservation the sum over the excited states can 
be restricted to all $n$ with $M_n = M$. The limit $N \rightarrow
\infty$ is called the Trotter limit. It is only in this limit that
the original quantum problem is recovered. The right hand side of
(\ref{deflongffseries}) is a thermal form factor series representation
of the correlation function. Rather than working directly with this
series we shall introduce a generating function, which will simplify
the subsequent calculations to a certain extent.

We call the functions
\begin{equation} \label{defazz}
     A_n^{zz} (h) = \frac{\<h|\Si^z (-\i \g/2|h)|n, h\>\<n, h|\Si^z (-\i \g/2|h)|h\>}
	             {\La_n (- \i \g/2|h) \<h|h\> \La (- \i \g/2|h) \<n, h|n, h\>}
\end{equation}
the form factor amplitudes of the longitudinal correlation function.
The right hand side of (\ref{defazz}) can be recast as
\begin{equation} \label{azzagen}
     A_n^{zz} (h) = 2 T^2 \6_{h'}^2 \,
                \frac{\<h|n, h'\>\<n, h'|h\>}{\<h|h\>\<n, h'|n, h'\>}
                \biggl[
		   \frac{\La_n (- \i \g/2|h')}{\La (- \i \g/2|h)} - 2 +
		   \frac{\La (- \i \g/2|h)}{\La_n (- \i \g/2|h')} \biggr]_{h' = h}
\end{equation}
(see Appendix~\ref{app:genfun} for a derivation). This form of
the amplitudes allows us to rewrite the form factor series
(\ref{deflongffseries}) in terms of a generating function.
We introduce the amplitudes of its form factor series,
\begin{equation} \label{ampgenfun}
     A_n (h,h') = \frac{\<h|n,h'\> \<n,h'|h\>}{\<h|h\> \<n,h'|n,h'\>} \epc
\end{equation}
and the formal difference operator $D_m$ acting as $D_m f_m
= f_m - f_{m-1}$ on sequences. Then
\begin{equation} \label{szszgenerated}
     \bigl\<\s_1^z \s_{m+1}^z (t)\bigr\>_T =
        2 T^2 \6_{h'}^2 D_m^2 {\cal G} (m+1, t, T, h, h')\bigr|_{h' = h} \epc
\end{equation}
where
\begin{multline} \label{defgenfun}
     {\cal G} (m, t, T, h, h') =
        \lim_{N \rightarrow \infty} \lim_{\e \rightarrow 0}
	\sum_n A_n (h, h') \r_n (- \i \g/2|h, h')^m \\[-2ex] \times
               \biggl(\frac{\r_n (- \i \g/2 + \i t/\ks N|h,h')}
	                   {\r_n (- \i \g/2 - \i t/\ks N|h,h')}\biggr)^\frac{N}{2} \epp
\end{multline}

\subsection{Amplitudes at finite Trotter number}
The two functions
\begin{align}
     \re (x) & = \ctg(x) - \ctg(x - \i \g) \epc \\[1ex] \label{defk}
     K(x) & = \ctg(x - \i \g) - \ctg(x + \i \g) = - \re (x) - \re (-x)
\end{align}
will be used throughout these notes. We will call them the bare
energy and the kernel function.

The main objective of this work is to obtain a manageable expression
for the Trotter limit, $N \rightarrow \infty$, of the amplitudes
(\ref{ampgenfun}) appearing in the form factor series of the
generating function (\ref{defgenfun}). To begin with, we will express
the amplitudes for a given excited state $|n, h\>$ at finite Trotter
number $N$ in terms of the two sets of Bethe roots of the excited
state and the dominant state. For the XXZ chain this is possible
due to the well-known scalar product formula of Nikita Slavnov
\cite{Slavnov89} (for a recent constructive proof see \cite{BeSl19}).

\begin{lemma}
{\bf Slavnov formula.}\\ Let $\<n, h|$ be a left eigenstate of the
dynamical quantum transfer matrix (\ref{defqtm}) with $M_n$
Bethe rapidities and $|\{\n\}, h\>$ be an off-shell Bethe vector
(\ref{offshellbv}) with $\card \{\n\} = M_n$. The scalar product of
these two states can be represented as 
\begin{equation} \label{slavnov}
     \bigl\<n, h\big|\{\n\},h \bigr\> = 
        \Biggl[ \prod_{j=1}^{M_n} d\bigl(\la_j^{(n)}\big|h\bigr) \La_n (\n_j |h) \Biggr]
	\frac{\det_{M_n} \Bigl\{\frac{\re(\la_j^{(n)} - \n_k)}{1 + \fa_n (\n_k|h)}
	                    - \frac{\re(\n_k - \la_j^{(n)})}{1 + 1/\fa_n (\n_k|h)}\Bigr\}}
             {\det_{M_n} \Bigl\{\frac{1}{\sin(\la_j^{(n)} - \n_k)}\Bigr\}} \epp
\end{equation}
\end{lemma}

\begin{corollary}
Let us now fix an excited state $|n, h'\>$ with $M_n = M$ and
denote the corresponding set of Bethe roots $\{\m_j\}_{j=1}^M
= \{\la_j^{(n)} \}_{j=1}^{M_n}$ for brevity. Then
\begin{multline} \label{start}
     A_n (h,h') = \Biggl[\prod_{j=1}^M \frac{\r_n (\la_j|h, h')}{\r_n (\m_j|h, h')}\Biggr]
                   \frac{\det_M \Bigl\{\frac{\re(\la_j - \m_k)}{1 + \fa (\m_k|h)}
		                       - \frac{\re(\m_k - \la_j)}{1 + 1/\fa (\m_k|h)}\Bigr\}}
                        {\det_M \Bigl\{\de^j_k + \frac{K (\la_j - \la_k)}{\fa' (\la_k|h)}\Bigr\}
			 \det_M \Bigl\{\frac{1}{\sin(\la_j - \m_k)}\Bigr\}} \\[1ex] \times
                   \frac{\det_M \Bigl\{\frac{\re(\m_j - \la_k)}{1 + \fa_n (\la_k|h')}
		                       - \frac{\re(\la_k - \m_j)}{1 + 1/\fa_n (\la_k|h')}\Bigr\}}
                        {\det_M \Bigl\{\de^j_k + \frac{K (\m_j - \m_k)}{\fa_n' (\m_k|h')}\Bigr\}
			 \det_M \Bigl\{\frac{1}{\sin(\m_j - \la_k)}\Bigr\}} \epp
\end{multline}
\end{corollary}
\begin{proof}
For the proof note that
\begin{equation} \label{scaproams}
     \frac{\<h|\{\n\}, h'\>\<n, h'|\{\tilde \n\}, h\>}
          {\<h|\{\tilde \n\}, h\>\<n, h'|\{\n\}, h'\>} =
     \frac{\<h|\{\n\}, h\>\<n, h'|\{\tilde \n\}, h'\>}
          {\<h|\{\tilde \n\}, h\>\<n, h'|\{\n\}, h'\>}
\end{equation}
for any two off-shell Bethe states $|\{\n\}, h'\>$, $|\{\tilde \n\}, h\>$
with $\card\{\n\} = \card\{\tilde \n\} = M$. Then use (\ref{slavnov})
in (\ref{scaproams}) and send $\n_k \rightarrow \m_k$, $\tilde \n_k
\rightarrow \la_k$. In order to obtain the limits in the denominator
write
\begin{multline}
	\frac{\det_M \Bigl\{\frac{\re(\m_j - \n_k)}{1 + \fa_n (\n_k|h)}
	                    - \frac{\re(\n_k - \m_j)}{1 + 1/\fa_n (\n_k|h)}\Bigr\}}
             {\det_M \Bigl\{\frac{1}{\sin(\m_j - \n_k)}\Bigr\}} = \\
	\frac{\det_M \bigl\{\re(\n_k - \m_j) (1 + \fa_n (\n_k|h))
	                    - \re(\m_j - \n_k) - \re(\n_k - \m_j)\bigr\}}
             {\det_M \Bigl\{\frac{1 + \fa_n (\n_k|h)}{\sin(\n_k - \m_j)}\Bigr\}}
\end{multline}
and use (\ref{defk}), the Bethe Ansatz equations (\ref{baes}) and
l'Hospital's rule, and similarly for the other determinant in the
denominator.
\end{proof}

\section{\boldmath Bethe root patterns and eigenvalues in the low-$T$ limit}
In \cite{DGKS15b} we considered the low-$T$ spectrum of the quantum
transfer matrix of the XXZ chain in the antiferromagnetic massive regime.
In this regime a characteristic feature of the sets of Bethe roots
for fixed $h > 0$ and $T \rightarrow 0+$ is the absence of strings,
regular patterns in the complex plane which would appear at $h = 0$
or for the usual transfer matrix. All excitations of the quantum
transfer matrix in the low-$T$ limit can rather be interpreted as
particle-hole excitations.

In \cite{Sakai07,GKKKS17} it was observed that the spectra and Bethe
root patterns of the quantum transfer matrix and of the dynamical
quantum transfer matrix coincide in the Trotter limit, $N \rightarrow \infty$.
This follows from the fact that the driving terms in the associated
non-linear integral equations have the same Trotter limit which is
independent of $t$. For this reason the results of \cite{DGKS15b}
can be used in the dynamical case as well. We shall recall them
below. They are based on a thorough analysis of the solutions of the
non-linear integral equations in the entire complex plane. As not
only their solutions but the non-linear integral equations themselves 
will be needed in the sequel, we include a short a posteriori derivation
at the end of this section.

\subsection{A reminder of the set of functions that determine the
low-temperature spectrum of correlation lengths and the universal part}
\label{sec:lowtfun}
The basic functions that eventually appeared in our previous low-$T$
analysis of the correlation lengths and form factors of the XXZ chain in the
antiferromagnetic massive regime were $q$-gamma and $q$-Barnes functions.
They may be expressed in terms of (infinite) $q$-multi factorials which,
for $|q_j| < 1$ and $a \in {\mathbb C}$, are defined as
\begin{equation}
     (a;q_1, \dots, q_p) =
        \prod_{n_1, \dots, n_p = 0}^\infty (1 - a q_1^{n_1} \dots q_p^{n_p}) \epp
\end{equation}
Based on this definition we introduce the $q$-gamma and $q$-Barnes
functions $\G_q$ and $G_q$,
\begin{subequations}
\label{defgammaqgq}
\begin{align}
     \G_q (x) & = (1 - q)^{1 - x} \frac{(q;q)}{(q^x;q)} \epc \\[1ex]
     G_q (x) & = (1 - q)^{- \2 (1 - x)(2 - x)} (q;q)^{x - 1}
               \frac{(q^x;q,q)}{(q;q,q)} \epp
\end{align}
\end{subequations}
They satisfy the normalisation conditions
\begin{equation}
     \G_q (1) = G_q (1) = 1
\end{equation}
and the basic functional equations
\begin{equation} \label{gammabarnesfun}
     [x]_q \G_q (x) = \G_q (x + 1) \epc \qd \G_q (x) G_q (x) = G_q (x + 1) \epc
\end{equation}
where
\begin{equation} \label{gaussnumber}
     [x]_q = \frac{1 - q^x}{1 - q} \epp
\end{equation}
is a familiar form of the $q$-number.

A closely related family of functions are the Jacobi theta functions
$\dh_j (x) = \dh_j (x|q)$, $j = 1, \dots, 4$, where
\begin{equation} \label{proddh4}
     \dh_4 (x|q) = (q^2;q^2) (e^{- 2\i x} q; q^2) (\re^{2 \i x} q; q^2)
\end{equation}
and
\begin{align} \label{defotherdhs}
     \dh_1 (x) & = - \i q^\4 \re^{\i x} \dh_4 (x + \i \g/2) \epc \qd
     \dh_2 (x) = q^\4 \re^{\i x} \dh_4 (x + \i \g/2 + \p/2) \epc
        \notag \\[1ex]
     \dh_3 (x) & = \dh_4 (x + \p/2) \epp
\end{align}
These functions are related to the $q$-gamma functions through the
second functional relation of the latter, which can, for instance,
be written as
\begin{equation} \label{scdgammafuneq}
     \frac{\dh_4 (x)}{\dh_4} =
        \frac{\G_{q^2}^2 \bigl(\tst{\2}\bigr)}
	     {\G_{q^2} \bigl(\tst{\2 - \frac{\i x}{\g}}\bigr)
	      \G_{q^2} \bigl(\tst{\2 + \frac{\i x}{\g}}\bigr)} \epp
\end{equation}
Here we have employed a common convention \cite{WhWa63} for
theta constants,
\begin{equation}
     \dh_1' = \dh_1' (0) \epc \qd
     \dh_j = \dh_j (0) \epc \qd j = 2, 3, 4 \epp
\end{equation}

Using the above set of functions we can proceed with those functions
that determine the physical properties of the XXZ chain in the
low-$T$ limit. These are the dressed momentum $p$, the dressed energy
$\e$ and the dressed phase $\ph$. Dressed momentum and dressed
energy are defined as
\begin{align} \label{ptheta4}
      p(\la) & = \frac{\p}{2} + \la
             - \i \ln \biggl(
	       \frac{\dh_4 (\la + \i \g/2| q^2)}{\dh_4 (\la - \i \g/2| q^2)}
	       \biggr) \epc \\[1ex] \label{dressede}
     \e(\la|h) & = \frac h 2 - 2 J \sh (\g) \dh_3 \dh_4
                           \frac{\dh_3 (\la)}{\dh_4 (\la)} \epp
\end{align}
The dressed phase is the function
\begin{equation} \label{dressedphase}
     \ph(\la_1, \la_2) = \i \Bigl( \frac \p 2 + \la_{12} \Bigr)
        + \ln \Biggl\{ \frac{\G_{q^4} \bigl(1 + \frac{\i \la_{12}}{2\g}\bigr)
	                     \G_{q^4} \bigl(\2 - \frac{\i \la_{12}}{2\g}\bigr)}
		            {\G_{q^4} \bigl(1 - \frac{\i \la_{12}}{2\g}\bigr)
			     \G_{q^4} \bigl(\2 + \frac{\i \la_{12}}{2\g}\bigr)}
			     \Biggr\} \epc
\end{equation}
where $\la_{12} = \la_1 - \la_2$, $|\Im \la_2| < \g$.

In order to define the shift function, which fixes the Bethe root
patterns in the low-$T$ limit, we first introduce a convention for
sums and products that will be used throughout this work. Let 
${\cal X}, {\cal Y} \subset {\mathbb C}$ be two discrete sets.
%associate a map $d: {\cal X} \cup {\cal Y} \mapsto {\mathbb C}$,
%\begin{equation}
%     d_{{\cal X}, {\cal Y}} (z) =
%        \begin{cases}
%	   1 & z \in {\cal X} \setminus ({\cal X} \cap {\cal Y}) \\
%	   0 & z \in {\cal X} \cap {\cal Y} \\
%	   - 1 & z \in {\cal Y} \setminus ({\cal X} \cap {\cal Y})
%        \end{cases}
%\end{equation}
We shall write
\begin{equation} \label{defsetsumprod}
     \sum_{\la \in {\cal X} \ominus {\cal Y}} f(\la)
        = \sum_{\la \in {\cal X}} f(\la)
          - \sum_{\la \in {\cal Y}} f(\la) \epc \qd
     \prod_{\la \in {\cal X} \ominus {\cal Y}} f(\la)
        = \frac{\prod_{\la \in {\cal X}} f(\la)}{\prod_{\la \in {\cal Y}} f(\la)} \epp
\end{equation}
Then the shift function $F$ is defined as
\begin{equation}
    F(\la|{\cal X}, {\cal Y}) = \frac{1}{2 \p \i}
        \sum_{\m \in {\cal Y} \ominus {\cal X}} \ph(\la, \m) \epp
\end{equation}

Note that the dressed energy depends parametrically on the magnetic field
$h$. The condition $\e(\frac{\p}{2}|h_\ell) = 0$ determines the lower
critical field
\begin{equation} \label{hlow}
     h_\ell = 4 J \sh(\g) \dh_4^2 \epc
\end{equation}
i.e.\ the boundary between the antiferromagnetic massive and massless
regimes (see Figure~\ref{fig:phasediagram}). The curves $\Re \e (\la|h) 
= 0$ are shown in Figure~\ref{fig:regimes}. These are the curves on which
the Bethe roots `condense' at low $T$.
\begin{figure}
\begin{center}
\includegraphics[width=.70\textwidth]{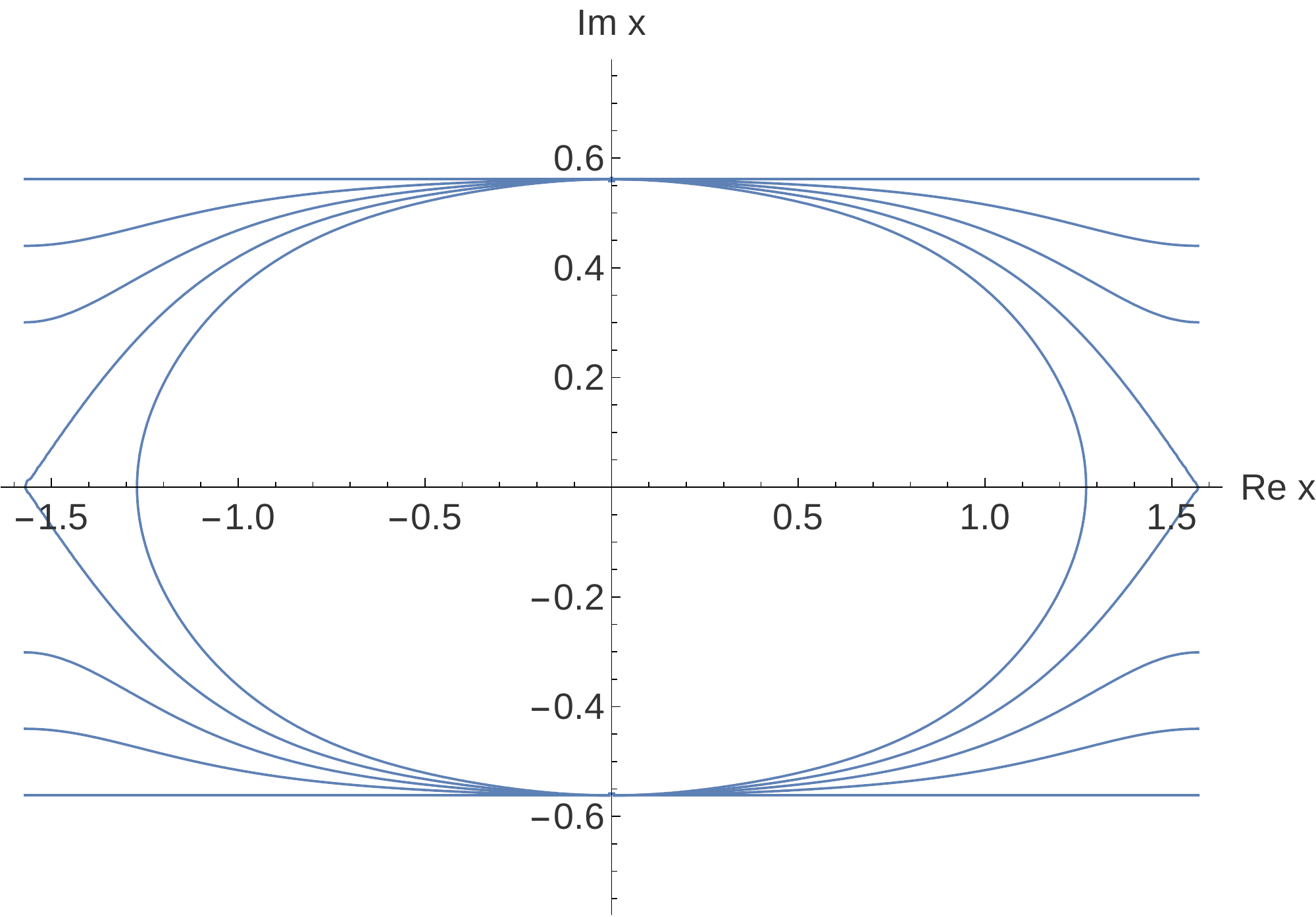}
\end{center}
\caption{\label{fig:regimes} The curves $\Re \e (\la|h) = 0$ for various
values of the magnetic field. Here $\D = 1.7$, $h_\ell/J = 0.76$.
The values of the magnetic field decrease proceeding from the inner
to the outer curve: $h/h_\ell = 1.34, 1, 2/3, 1/3, 0$. In the
critical regime $h_\ell < h < h_u = 4J (1 + \D)$ the curves are
closed. In this regime $\e$ is not given by (\ref{dressede}), but
is rather defined as a solution of a linear integral equation.
At the lower critical field $h = h_\ell$ the closed curves develop
two cusps, and a gap opens for $0 < h < h_\ell$, which is the
parameter regime considered in this work.}
\end{figure}

\begin{conjecture}
\label{con:nostrings}
{\bf \boldmath Low-$T$ Bethe root patterns at $0 < h < h_\ell$ \cite{DGKS15b}.}
\begin{enumerate}
\item
All excitations of the quantum transfer matrix at low $T$ and
large enough Trotter number $N$ can be parameterized by an even
number of complex parameters located inside the strip $|\Im \la| < \g/2$
and by an index $\fe \in {\mathbb Z}\slash 2 {\mathbb Z} $.
Referring to \cite{DGKS15b} we call the parameters in the upper
half plane particles, the parameters in the lower half plane holes.
We shall denote the set of particles by $\cal Y$, the set of
holes by $\cal X$.
\item
For excited states with $M_n = M = N + 1$ we have $\card {\cal X} =
\card {\cal Y} = \nex$.
\item
Up to corrections of the order $T^\infty$ the particles and holes
are determined by the low-$T$ higher-level Bethe Ansatz equations
\begin{subequations}
\label{hlbaes}
\begin{align}
     \e(y|h) & = 2 \p \i T \bigl(\ell_y + \fe/2 + F(y|{\cal X}, {\cal Y})\bigr) \epc
			     && \forall\ y \in {\cal Y} \epc \\[1ex]
     \e(x|h) & = 2 \p \i T \bigl(m_x + \fe/2 + F(x|{\cal X}, {\cal Y})\bigr) \epc
			     && \forall\ x \in {\cal X} \epc
\end{align}
\end{subequations}
where $\ell_y, m_x \in {\mathbb Z}$ and where the $\ell_y$ and the
$m_x$ are mutually distinct.
\item
For solution ${\cal X}_n$, ${\cal Y}_n$ of (\ref{hlbaes}) with index
$\fe = \fe_n$ we denote $F (\la|{\cal X}_n, {\cal Y}_n)  = F_n (\la)$. The
auxiliary function corresponding to this solution is
\begin{equation} \label{anlowt}
     \fa_n (\la|h) = (-1)^{\fe_n} \re^{- \e(\la|h)/T + 2\p\i F_n (\la)}
                      \bigl(1 + {\cal O} (T^\infty)\bigr)
\end{equation}
uniformly for $|\Im \la| < \g/2$, away from the points $\pm \i \g/2$.
\end{enumerate}
\end{conjecture}

In the limit $T \rightarrow 0+$ the higher-level Bethe Ansatz equations
(\ref{hlbaes}) decouple, $\i \p \ell_y T$ and $\i \p m_x T$ turn into
independent continuous variables, and the particles and holes become
free parameters on the curves
\begin{equation} \label{defbpm}
     {\cal B}_\pm = \bigl\{ \la \in {\mathbb C} \big|
                            \Re \e (\la|h) = 0, - \p/2 \le \Re \la \le \p/2,
			    0 < \pm \Im \la < \g \bigr\} \epp
\end{equation}
These curves are shown in Figure~\ref{fig:regimes}. As we can see, the
massive regime is distinguished from the massless regime by the
opening of a `mass gap' at the critical field $h_\ell$.

\subsection{Low-temperature limit of the eigenvalue ratios}
\label{sec:evratlowt}
The main result of our work \cite{DGKS15b} was an explicit formula
for all correlation lengths, or rather all eigenvalue ratios, in
the low-temperature regime. At low enough temperatures all excitations
are parameterized by solutions ${\cal X}_n$, ${\cal Y}_n$ of the
higher-level Bethe Ansatz equations (\ref{hlbaes}). Let $k = \fe_n -
\fe_0$. Then, for $- \g < \Im \la < 0$, the corresponding eigenvalue
ratios in the Trotter limit and at finite magnetic field
\cite{DGKS15b,DGKS16b} can be expressed as
\begin{equation} \label{evarat}
     \r_n (\la|h, h') = (-1)^k 
        \exp \biggl\{\i \mspace{-8.mu}
	       \sum_{z \in {\cal X}_n \ominus {\cal Y}_n} \mspace{-16.mu}
	                        p(\la - z + \i \g/2)\biggr\} \epc
\end{equation}
this being valid up to multiplicative corrections of the form
$\bigl(1 + {\cal O} (T^\infty)\bigr)$. Note that the value of the
magnetic field $h'$ enters here only through the particle and hole
parameters.

From equation (\ref{evarat}) we obtain the eigenvalue ratios entering
the form factor series (\ref{defgenfun}),
\begin{equation} \label{evaratzero}
     \r_n (- \i \g/2|h, h') =
        (-1)^k \exp \biggl\{\i 
	       \mspace{-8.mu}
	        \sum_{z \in {\cal Y}_n \ominus {\cal X}_n}
	       \mspace{-16.mu} p(z) \biggr\}
	= (-1)^{k} \mspace{-8.mu} \prod_{z \in {\cal Y}_n \ominus {\cal X}_n}
	           \frac{\dh_1(z - \i \g/2| q^2)}{\dh_4(z - \i \g/2| q^2)} \epc
\end{equation}
and
\begin{multline}
     \lim_{\substack{N \rightarrow \infty\\ \e \rightarrow 0}}
        \biggl(\frac{\r_n (- \i \g/2 + \i t/\ks N|h,h')}
	            {\r_n (- \i \g/2 - \i t/\ks N|h,h')}\biggr)^\frac{N}{2} =
        \exp \bigl\{(\i t/\ks) \6_\la \ln \r_n (\la|h, h')\bigr\}_{\la = - \i \g/2} \\[1ex]
     = \exp \biggl\{2 \i t J \sh(\g) \mspace{-16.mu}
                    \sum_{z \in {\cal X}_n \ominus {\cal Y}_n}
	            \mspace{-16.mu} p'(z) \biggr\}
     = \exp \biggl\{\i t \mspace{-14.mu}
                     \sum_{z \in {\cal Y}_n \ominus {\cal X}_n}
	            \mspace{-16.mu} \e(z|h) \biggr\} \epp
\end{multline}
Both formula are again valid up to multiplicative temperature
corrections of the form $\bigl(1 + {\cal O} (T^\infty)\bigr)$.

\subsection{Non-linear integral equations}
\label{subsec:nlie}
The low-$T$ analysis can be based on a non-linear integral
equation satisfied by the auxiliary functions $\fa_n (\cdot|h)$.
Let us briefly recall this equation. We start with definitions
and useful notation. Let
\begin{equation}
     \th(\la|\g) = \int_{\G_\la} \rd \m \:
                    \bigl(\ctg(\m - \i \g) - \ctg(\m + \i \g)\bigr) \epc
\end{equation}
where $\G_\la$ is a piecewise straight contour beginning at
$- \frac{\p}{2}$, running parallel to the imaginary axis to
$- \frac{\p}{2} + \Im \la$, then continuing parallel to the
real axis from $- \frac{\p}{2} + \Im \la$ to $\la$. This
function has branch cuts along $(- \infty \pm \i \g,
- \p \pm \i \g] \cup [\pm \i \g, + \infty \pm \i \g)$
and satisfies the quasi-periodicity, quasi-reflection and
asymptotic conditions
\begin{subequations}
\label{qppkth}
\begin{align}
     \th (\la + \p|\g) & = \th(\la|\g) +
                          \begin{cases} 
			     2 \p \i & |\Im \la| < \g \\
                             0 & |\Im \la| > \g,
                          \end{cases}\\[1ex]
     \label{thetaminus}
     \th (- \la|\g) & = - \th(\la|\g) +
                         \begin{cases} 
			    2 \p \i & |\Im \la| < \g \\
                            0 & |\Im \la| > \g,
                         \end{cases} \\[1ex]
     \lim_{\Re \la \rightarrow \pm \infty} \th(\la|\g) & = \mp 2 \g \epp
\end{align}
\end{subequations}
It is constructed in such a way that
\begin{equation} \label{etotheta}
     \re^{\th(\la|\g)} = \frac{\sin(\la - \i \g)}{\sin(\la + \i \g)} \epp
\end{equation}
For $\g > 0$ we also set
\begin{equation}
     K_0 (\la|\g) = \frac{\th'(\la|\g)}{2 \p \i} \epp
\end{equation}
Then $K_0 (\la|\g) > 0$ for $\la \in {\mathbb R}$. Now define
\begin{align} \label{defepsnulneps}
     & \e^{(N,\e)}_0 (\la|h) =
        h + T \bigl\{\th\bigl(\la - \e\big|\tst{\frac{\g}{2}}\bigr)
                     - \th\bigl(\la + \e\big|\tst{\frac{\g}{2}}\bigr)\bigr\} \\[1ex]
	\notag & \mspace{.9mu}
        + \frac{NT}{2} \Bigl\{
	  \th\bigl(\la + \tst{\frac{\i t}{\ks N}}\big|\tst{\frac{\g}{2}}\bigr)
	  - \th\bigl(\la - \tst{\frac{\i t}{\ks N}}\big|\tst{\frac{\g}{2}}\bigr)
	  + \th\bigl(\la - \tst{\frac{\i t + 1/T}{\ks N}}\big|\tst{\frac{\g}{2}}\bigr)
	  - \th\bigl(\la + \tst{\frac{\i t + 1/T}{\ks N}}\big|\tst{\frac{\g}{2}}\bigr)
	  \Bigr\} \epp
\end{align}
This function has the point-wise limit
\begin{equation} \label{defepszero}
     \lim_{N \rightarrow \infty} \lim_{\e \rightarrow 0} \e^{(N,\e)}_0 (\la|h) =
        h - \th' \bigl(\la|\tst{\frac{\g}{2}}\bigr)/\ks =
	h - 4 \p J \sh(\g) K_0 \bigl(\la|\tst{\frac{\g}{2}}\bigr) =
	\e_0 (\la|h)
\end{equation}
which is independent of $t$ and $T$.

For a shortcut in the derivation of the non-linear integral
equation we shall now use our insight from previous work
\cite{DGKS15b} summarized above as an input. Recall that we
have fixed a solution $\{\m_j\}_{j=1}^M = \{\la_j^{(n)}\}_{j=1}^{M_n}$
of the Bethe ansatz equations (\ref{baes}) which determines an
auxiliary function $\fa_n(\cdot|h')$ and that we are keeping the
magnetic field $h$ pertaining to the dominant state and $h'$
pertaining to the $n$th excited state independent. Let us now
assume that for fixed small $T$ and $N$ large enough all Bethe
roots $\m_j$ are close to ${\cal B}_\pm$, such that $|\Im \m_j|
< \g/2$, $j = 1, \dots, M$. Using (\ref{etotheta}) and
(\ref{defepsnulneps}) in (\ref{baaux}) we see that
\begin{multline} \label{deflnan}
     \ln \fa_n (\la|h') = - \frac{\e_0^{(N,\e)}(\la|h')}{T}
        + \sum_{j=1}^{N+1} \th(\la - \m_j|\g) 
	\\
        - \frac{N}{2} \Bigl\{
	  \th\bigl(\la + \tst{\frac{\i \g}{2}} - \tst{\frac{\i t}{\ks N}}\big|\g\bigr)
	  + \th\bigl(\la + \tst{\frac{\i \g}{2}} + \tst{\frac{\i t + 1/T}{\ks N}}\big|\g\bigr)
	  \Bigr\}
	- \th\bigl(\la + \tst{\frac{\i \g}{2}} + \e \big|\g\bigr)
\end{multline}
is a possible definition of the logarithm of $\fa_n (\cdot|h')$.
This definition guarantees that $\ln \fa_n (\cdot|h')$ is $\p$-periodic
and analytic for $|\Im \la| < \g/2$ in the following sense. For a given
$\la$ with $|\Im \la| < \g/2$ and small $T > 0$ there is an $N_0 \in
{\mathbb N}$ such that $\ln \fa_n (\cdot|h')$ is analytic in $\la$ and
\begin{equation}
     \ln \fa_n (\la + \p|h') = \ln \fa_n (\la|h')
\end{equation}
for all $N > N_0$.

We fix a simple closed and positively oriented contour ${\cal C} = {\cal C}_l
+ {\cal C}_- + {\cal C}_r + {\cal C}_+$, where
\begin{align}
     {\cal C}_l & = [-\p/2, -\p/2 - \i \g/2 - \i 0] \epc &
     {\cal C}_r & = - {\cal C}_l + \p \epc \\[1ex]
     {\cal C}_+ & = [\p/2, -\p/2] \epc &
     {\cal C}_- & = - {\cal C}_+ - \i \g/2 - \i 0  \epp
\end{align}
Here the regularization $\i 0$ means that the pole of $\fa_n (\la|h')$
at $\la = - \frac{\i \g}{2} + \frac{\i t }{\ks N}$ is inside ${\cal C}$.

Let
\begin{equation} \label{defzset}
     {\cal Z}_n = \{\la \in {\mathbb C}|1 + \fa_n (\la|h') = 0\}
                  \cap \Int {\cal C} \epp
\end{equation}
With $\{\m_j\}_{j=1}^M$ we associate sets of `particles' ${\cal Y}_n$
and `holes' ${\cal X}_n$ relative to ${\cal C}$, setting
\begin{equation} \label{defxyset}
     {\cal Y}_n = \{\m_j\}_{j=1}^M \cap \Ext {\cal C} \epc \qd
     {\cal X}_n = {\cal Z}_n \setminus \bigl(\{\m_j\}_{j=1}^M
                  \cap \Int {\cal C}\bigr) \epp
\end{equation}
Using the residue theorem in (\ref{deflnan}) then implies that
\begin{equation} \label{prenlie}
     \ln \fa_n (\la|h') =
        - \frac{\e_0^{(N, \e)} (\la|h')}{T}
	             + \mspace{-18.mu}
		          \sum_{z \in {\cal Y}_n \ominus {\cal X}_n}
		       \mspace{-18.mu} \th(\la - z|\g)
		     + \int_{\cal C} \frac{\rd \m}{2 \p \i} \th(\la - \m|\g)
		       \frac{\6_\m \fa_n (\m|h')}{1 + \fa_n (\m|h')}
\end{equation}
for all $\la \in {\mathbb C}$ with $|\Im \la| < \g/2$. Here
we have used that, due to our assumptions, $\m_j - \i \g
\notin \Int {\cal C}$.

Equation (\ref{prenlie}) turns into a non-linear integral
equation for $\fa_n (\cdot|h')$ upon applying partial integration
to the integral. For this purpose a proper definition of
the logarithm of $1 + \fa_n (\cdot|h')$ is required. Let us
assume (\ref{anlowt}) to hold approximately for small $T$ and
$N$ large enough. Then $|\fa_n (\la|h')| = 1$, $- \g/2 < \Im \la < 0$
determines a curve ${\cal B}_n$ close to ${\cal B}_-$ which
intersects the line $\Re \la = - \frac{\p}{2}$ in a point
$- \frac{\p}{2} - \i x_0$ with $- \g/2 < x_0 < 0$. Due to the
$\p$-periodicity of $\fa_n (\cdot|h')$ the curve ${\cal B}_n$
passes through $\frac{\p}{2} - \i x_0$ as well. Define a domain
${\cal D}_\auf$ located between ${\cal B}_n$ and the real axis
and a domain ${\cal D}_\ab$ located between the line $\Im \la
= - \frac{\g}{2}$ and ${\cal B}_n$. Then
\begin{equation}
     |\fa_n (\la|h')| \begin{cases}
		         > 1 & \text{for $\la \in {\cal D}_\auf$} \\
		         < 1 & \text{for $\la \in {\cal D}_\ab$.}
                      \end{cases}
\end{equation}
This allows us to define the logarithm of $1 + \fa_n (\cdot|h')$ on
${\cal D}_\auf \cup {\cal D}_\ab$ as an analytic and $\p$-periodic
function,
\begin{equation} \label{deflogoneplusa}
     \Ln (1 + \fa_n) (\la| h') =
        \begin{cases}
	     \ln \fa_n (\la|h') + \ln \bigl(1 + 1/\fa_n (\la|h')\bigr) &
	     \text{for $\la \in {\cal D}_\auf$} \\[.5ex]
	     \ln \bigl(1 + \fa_n (\la|h')\bigr) & \text{for $\la \in {\cal D}_\ab$,}
        \end{cases}
\end{equation}
where $\ln \fa_n (\cdot|h')$ is defined by (\ref{deflnan}) and the
other logarithms on the right hand side are defined by the principal
branch. By construction $\Ln (1 + \fa_n) (\cdot|h')$ has jump
discontinuities across ${\cal B}_n$. In particular, at $\mp \frac{\p}{2}
- \i x_0$
\begin{equation} \label{jumplnoneplusa}
     \Ln (1 + \fa_n) \bigl(\mp \tst{\frac{\p}{2}} - \i x_0^-\big| h'\bigr)
        - \Ln (1 + \fa_n) \bigl(\mp \tst{\frac{\p}{2}} - \i x_0^+\big| h'\bigr)
	= 2 \p \i \fe_n
\end{equation}
for some $\fe_n \in {\mathbb Z}$.

Using the properties of the function $\Ln (1 + \fa_n) (\cdot|h')$ we
may now calculate the index of $1 + \fa_n (\cdot|h')$ along ${\cal C}$,
\begin{multline} \label{indexoneplusa}
     \int_{\cal C} \frac{\rd \m}{2 \p \i}
        \frac{\6_\m \fa_n (\m|h')}{1 + \fa_n (\m|h')} =
        \Ln (1 + \fa_n) \bigl(- \tst{\frac{\p}{2}} - \i x_0^-\big| h'\bigr)
        - \Ln (1 + \fa_n) \bigl(\tst{\frac{\p}{2}} - \i x_0^-\big| h'\bigr) \\
        + \Ln (1 + \fa_n) \bigl(\tst{\frac{\p}{2}} - \i x_0^+\big| h'\bigr)
        - \Ln (1 + \fa_n) \bigl(- \tst{\frac{\p}{2}} - \i x_0^+\big| h'\bigr)
	= 0 \epp
\end{multline}
Here we have used the analyticity of $\Ln (1 + \fa_n) (\cdot|h')$
in ${\cal D}_\auf$ und ${\cal D}_\ab$ in the first equation and
(\ref{jumplnoneplusa}) in the second equation. Calculating the same
integral by means of the residue theorem we see that
\begin{equation}
     \int_{\cal C} \frac{\rd \m}{2 \p \i}
        \frac{\6_\m \fa_n (\m|h')}{1 + \fa_n (\m|h')} =
	\card {\cal Y}_n - \card {\cal X}_n \epc
\end{equation}
implying that
\begin{equation}
	\card {\cal Y}_n = \card {\cal X}_n = \nex \epp
\end{equation}

Performing a similar calculation as in (\ref{indexoneplusa}) and using
also the quasi-periodicity (\ref{qppkth}) of the function $\th(\cdot|\g)$
we may now rewrite the integral on the right hand side of
(\ref{prenlie}),
\begin{equation}
     \int_{\cal C} \frac{\rd \m}{2 \p \i} \th(\la - \m|\g)
        \frac{\6_\m \fa_n (\m|h')}{1 + \fa_n (\m|h')} = 2 \p \i \fe_n +
        \int_{\cal C} \rd \m \: K_0 (\la - \m|\g) \Ln (1 + \fa_n) (\m|h') \epp
\end{equation}
Inserting this back into (\ref{prenlie}) we have converted the latter
equation into a non-linear integral equation for the function
$\fa_n (\cdot|h')$,
\begin{multline} \label{nlie1}
     \ln \fa_n (\la|h') =
        - \frac{\e_0^{(N, \e)} (\la|h')}{T} \\ + 2 \p \i \fe_n 
	             + \mspace{-18.mu}
		          \sum_{z \in {\cal Y}_n \ominus {\cal X}_n}
		       \mspace{-18.mu} \th(\la - z|\g) +
          \int_{\cal C} \rd \m \: K_0 (\la - \m|\g) \Ln (1 + \fa_n) (\m|h') \epp
\end{multline}
This equation determines $\fa_n (\cdot|h')$ on ${\cal C}$ and then,
by analytic continuation in the entire complex plane. Recall that,
$\fa_n (x|h') = \fa_n (y|h') = - 1$ for all $x \in {\cal X}_n$,
$y \in {\cal Y}_n$ by construction.
\begin{Remark}
The non-linear integral equation (\ref{nlie1}) is appropriate for
taking the Trotter limit of the function $\fa_n (\cdot|h')$. Formally
we just have to replace the function $\e_0^{(N,\e)} (\cdot|h')$
by its limit $\e_0 (\cdot|h')$ defined in (\ref{defepszero}).
\end{Remark}
%\begin{Remark}
%The non-linear integral equation (\ref{nlie1}) is presented in
%a form that is also optimized for taking the low-$T$ limit
%after the Trotter limit. For this purpose one introduces a
%function $u_n (\la|h') = - T \ln \fa_n (\la|h')$. For more details
%compare \cite{DGKS15b}.
%\end{Remark}
\begin{Remark}
Reversing the steps in the derivation of (\ref{indexoneplusa})
and (\ref{nlie1}) we can go back to the Bethe Ansatz equations
(\ref{baes}). This means that every solution of the non-linear integral
equation for which $1 + \fa_n (\cdot|h')$ has index zero corresponds
to a solution of the Bethe Ansatz equation and to an eigenstate of
the dynamical quantum transfer matrix.  In \cite{DGKS15b} we argued
that, in the low-$T$ limit, all eigenstates can be obtained this way.
\end{Remark}

We shall also need an `off-shell version' of the auxiliary function
$\fa_n (\cdot|h')$ which is defined by means of a non-linear integral
equation similar to (\ref{nlie1}). In order to construct it we first
of all set
\begin{align} \label{defepsnulnepsbar}
     & \overline{\e}_0^{(N,\e)} (\la|h) =
        h + T \bigl\{\th\bigl(\la + \i \g + \e\big|\tst{\frac{\g}{2}}\bigr)
                     - \th\bigl(\la + \i \g - \e\big|\tst{\frac{\g}{2}}\bigr)\bigr\}
        \notag \\[1ex]
	\notag & \mspace{72.mu}
        + \frac{NT}{2} \Bigl\{
	  \th\bigl(\la + \i \g - \tst{\frac{\i t}{\ks N}}\big|\tst{\frac{\g}{2}}\bigr)
	  - \th\bigl(\la + \i \g + \tst{\frac{\i t}{\ks N}}\big|\tst{\frac{\g}{2}}\bigr)
        \\ & \mspace{162.mu}
	  + \th\bigl(\la + \i \g + \tst{\frac{\i t + 1/T}{\ks N}}\big|\tst{\frac{\g}{2}}\bigr)
	  - \th\bigl(\la + \i \g - \tst{\frac{\i t + 1/T}{\ks N}}\big|\tst{\frac{\g}{2}}\bigr)
	  \Bigr\}
\end{align}
Then, for every pair of sets ${\cal U} \subset \Int {\cal C}$,
${\cal V} \subset \Ext {\cal C}$ with $\card {\cal U} = \card {\cal V}
= \nex$, the function $\fa^- (\cdot|{\cal U}, {\cal V}, h')$ is the
solution of the non-linear integral equation
\begin{multline} \label{nlieoffshell1}
     \ln \fa^- (\la|{\cal U}, {\cal V}, h') =
        \frac{\overline{\e}_0^{(N, \e)} (\la|h')}{T} \\ - 2 \p \i \fe_n'
	             - \mspace{-12.mu}
		          \sum_{z \in {\cal V} \ominus {\cal U}}
		       \mspace{-8.mu} \th(\la - z|\g) -
          \int_{\cal C} \rd \m \: K_0 (\la - \m|\g)
	     \Ln^- (1 + \fa^-) (\m|{\cal U}, {\cal V}, h') \epp
\end{multline}
Here
\begin{multline}
     \Ln^- (1 + \fa^-) (\la|{\cal U}, {\cal V}, h') \\ =
        \begin{cases}
	     \ln \bigl(1 + \fa^- (\la|{\cal U}, {\cal V}, h')\bigr) &
	     \text{for $\la \in {\cal D}_\auf^-$} \\[.5ex]
	     \ln \fa^- (\la|{\cal U}, {\cal V}, h')\bigr) +
	     \ln \bigl(1 + 1/\fa^- (\la|{\cal U}, {\cal V}, h')\bigr) &
	     \text{for $\la \in {\cal D}_\ab^-$,}
        \end{cases}
\end{multline}
and ${\cal D}_\auf^-$ and ${\cal D}_\ab^-$ are separated by the
line, where $|\fa^- (\la|{\cal U}, {\cal V}, h')\bigr)| = 1$.
The contour ${\cal C}$ is the same as in (\ref{nlie1}) but
properly regularized such as to include the point
$- \frac{\i \g}{2} + \frac{\i t + 1/T}{\ks N}$ for sufficiently
large values of $N$. The integral equation (\ref{nlieoffshell1})
is constructed in such a way that
\begin{equation}
     \fa_n (\cdot|h') = 1/\fa^-(\cdot|{\cal X}_n, {\cal Y}_n, h')
\end{equation}
if ${\cal X}_n$ and ${\cal Y}_n$ are solutions of the subsidiary
conditions $\fa^-(x|{\cal X}_n, {\cal Y}_n, h') = -1$
for all $x \in {\cal X}_n$, and $\fa^-(y|{\cal X}_n, {\cal Y}_n, h') = - 1$
for all $y \in {\cal Y}_n$.

Hence, it is natural to define an off-shell function
\begin{equation} \label{defapm}
     \fa^+ (\cdot|{\cal U}, {\cal V}, h') = 1/\fa^- (\cdot|{\cal U}, {\cal V}, h') \epp
\end{equation}
Then
\begin{subequations}
\begin{align}
     \fa_n (\cdot|h') & = \fa^+ (\cdot|{\cal X}_n, {\cal Y}_n, h')
                        = \fa_n^+ (\cdot|h') \epc \\[1ex]
     1/\fa_n (\cdot|h') & = \fa^- (\cdot|{\cal X}_n, {\cal Y}_n, h')
                        = \fa_n^- (\cdot|h') \epp
\end{align}
\end{subequations}

\section{Low-temperature analysis of the amplitudes}
In this section we derive explicit expressions for the amplitudes
(\ref{ampgenfun}) in the low-temperature limit.
\subsection{A factorization of determinants}
Denoting the elements of ${\cal X}_n$ by $x_j$, the elements of
${\cal Y}_n$ by $y_j$ and the elements of ${\cal Z}_n$ by $z_j$
(cf.\ (\ref{defzset}), (\ref{defxyset})) we may label them in such
a way that
\begin{align}
     z_j & = \m_j \epc & & j = 1, \dots, M - \nex \epc \notag \\[1ex]
     y_j & = \m_{M - \nex + j} \epc & & j = 1, \dots, \nex \epc \notag \\[1ex]
     x_j & = z_{M - \nex + j} \epc & & j = 1, \dots, \nex \epp
\end{align}

Using this notation we shall separate the particle- hole contributions
from the determinants on the right hand side of (\ref{start}). We first
define a number of auxiliary matrices through their matrix elements. For
these we shall use a convention in which upper indices refer to rows
and lower indices to columns. Let
\begin{align}
     K^j_k & = \frac{K(\la_j - \la_k)}{\fa' (\la_k|h)} \epc && j,k = 1, \dots, M \epc 
               \notag \\[1ex]
     V^j_k & = \frac{\re(\la_j - \m_k)}{1 + \fa(\m_k|h)} -
               \frac{\re(\m_k - \la_j)}{1 + 1/\fa (\m_k|h)} \epc && j,k = 1, \dots, M \epp 
\end{align}
Further define
\begin{align}
     {V^{(i)}}^j_k & = V^j_k \epc
                       && j = 1, \dots, M ; k = 1, \dots, M - \nex \epc \notag \\[1ex]
     {V^{(p)}}^j_k & = V^j_{M - \nex + k} \epc
                       && j = 1, \dots, M ; k = 1, \dots, \nex \epc \notag \\[1ex]
     {V^{(h)}}^j_k & = \frac{\re(\la_j - x_k)}{1 + \fa(x_k|h)} -
                       \frac{\re(x_k - \la_j)}{1 + 1/\fa (x_k|h)} \epc
		       && j = 1, \dots, M ; k = 1, \dots, \nex
\end{align}
and
\begin{equation} \label{defgmat}
     G = (I_M + K)^{-1} \bigl(V^{(i)}, V^{(h)}\bigr) \epc \qd
     D = (0, I_{\nex}) G^{-1} (I_M + K)^{-1} V^{(p)} \epp
\end{equation}
Here and in the following $I_n$ denotes the $n \times n$ unit matrix,
and we use a block-matrix notation. E.g., in $\bigl(V^{(i)}, V^{(h)}\bigr)$
we combine the $M \times (M - \ell)$ block $V^{(i)}$ and the
$M \times \ell$ block $V^{(h)}$ into an $M \times M$ square matrix.
Similarly, $(0, I_{\nex})$ consists of an $\ell \times (M - \ell)$
block of zeros, combined with the $\ell \times \ell$ unit matrix
into an $\ell \times M$ matrix. This convention is particularly 
convenient for the proofs of the two following lemmata.

\begin{lemma} \label{lem:det1}
{\bf A factorization of determinants.}\\
We have the following factorization of determinants,
\begin{equation} \label{firstfactorsplit}
     \frac{\det_M \Bigl\{\frac{\re(\la_j - \m_k)}{1 + \fa (\m_k|h)}
           - \frac{\re(\m_k - \la_j)}{1 + 1/\fa (\m_k|h)}\Bigr\}}
          {\det_M \Bigl\{\de^j_k + \frac{K (\la_j - \la_k)}{\fa' (\la_k|h)}\Bigr\}} =
	  \det_M \{G\} \det_{\nex} \{D\} \epp
\end{equation}
\end{lemma}
\begin{proof}
With the notations above we can write
\begin{equation}
     \frac{\det_M \Bigl\{\frac{\re(\la_j - \m_k)}{1 + \fa (\m_k|h)}
           - \frac{\re(\m_k - \la_j)}{1 + 1/\fa (\m_k|h)}\Bigr\}}
          {\det_M \Bigl\{\de^j_k + \frac{K (\la_j - \la_k)}{\fa' (\la_k|h)}\Bigr\}} =
        \det_M \bigl((I_M + K)^{-1}\bigr) \det_M \bigl((V^{(i)},V^{(p)})\bigr) \epp
\end{equation}
Now
\begin{equation}
     \det_M \bigl((V^{(i)},V^{(p)})\bigr) =
        \det_{M + \nex}
	   \begin{bmatrix}
	      V^{(i)} & V^{(p)} & V^{(h)} \\
	      0 & 0 & I_\nex
	   \end{bmatrix} =
        \det_{M + \nex}
	   \begin{bmatrix}
	      V^{(i)} & V^{(h)} & V^{(p)} \\
	      0 & - I_\nex & 0
	   \end{bmatrix}.
\end{equation}
Hence,
\begin{multline}
     \det_M \bigl((I_M + K)^{-1}\bigr) \det_M \bigl((V^{(i)},V^{(p)})\bigr) =
        \det_{M + \nex}
	   \begin{bmatrix}
	      G & (I_M + K)^{-1} V^{(p)} \\
	      (0, - I_\nex) & 0
	   \end{bmatrix} \\[1ex] =
        \det_M \{G\}
        \det_{M + \nex}
	   \begin{bmatrix}
	      I_M & G^{-1} (I_M + K)^{-1} V^{(p)} \\
	      (0, - I_\nex) & 0
	   \end{bmatrix} \\[1ex] =
        \det_M \{G\}
        \det_{M + \nex}
	   \begin{bmatrix}
	      I_M & G^{-1} (I_M + K)^{-1} V^{(p)} \\
	      0 & (0, I_\nex) G^{-1} (I_M + K)^{-1} V^{(p)}
	   \end{bmatrix} 
\end{multline}
which implies (\ref{firstfactorsplit}).
\end{proof}

In order to perform a similar calculation with the second such ratio
on the right hand side of (\ref{start}) we have to decompose the
norm determinant in the denominator first. For this purpose let
\begin{align}
     L^j_k & = \frac{K(\m_j - \m_k)}{\fa_n' (\m_k|h')} \epc
               && j,k = 1, \dots, M \epc \notag \\[1ex]
     {L^{(ii)}}^j_k & = \frac{K(\m_j - \m_k)}{\fa_n' (\m_k|h')} \epc
                        && j,k = 1, \dots, M - \nex \epc \notag \\[1ex]
     {L^{(ip)}}^j_k & = \frac{K(\m_j - y_k)}{\fa_n' (y_k|h')} \epc
                        && j = 1, \dots, M - \nex; k = 1, \dots, \nex \epc \notag \\[1ex]
     {L^{(pi)}}^j_k & = \frac{K(y_j - \m_k)}{\fa_n' (\m_k|h')} \epc
                        && j = 1, \dots, \nex; k = 1, \dots, M - \nex \epc \notag \\[1ex]
     {L^{(pp)}}^j_k & = \frac{K(y_j - y_k)}{\fa_n' (y_k|h')} \epc
                        && j, k = 1, \dots, \nex \epc \notag \\[1ex]
     {L^{(ih)}}^j_k & = \frac{K(\m_j - x_k)}{\fa_n' (x_k|h')} \epc
                        && j = 1, \dots, M - \nex; k = 1, \dots, \nex \epc \notag \\[1ex]
     {L^{(hi)}}^j_k & = \frac{K(x_j - \m_k)}{\fa_n' (\m_k|h')} \epc
                        && j = 1, \dots, \nex; k = 1, \dots, M - \nex \epc \notag \\[1ex]
     {L^{(hh)}}^j_k & = \frac{K(x_j - x_k)}{\fa_n' (x_k|h')} \epc
                        && j, k = 1, \dots, \nex \epc \notag \\[1ex]
     {L^{(ph)}}^j_k & = \frac{K(y_j - x_k)}{\fa_n' (x_k|h')} \epc
                        && j, k = 1, \dots, \nex \epc \notag \\[1ex]
     {L^{(hp)}}^j_k & = \frac{K(x_j - y_k)}{\fa_n' (y_k|h')} \epc
                        && j, k = 1, \dots, \nex
\end{align}
and
\begin{align}
     & K^* = \begin{pmatrix}
                L^{(ii)} & L^{(ih)} \\ L^{(hi)} & L^{(hh)}
             \end{pmatrix}, \notag \\[1ex]
     & C = \begin{pmatrix}
              L^{(hi)} & L^{(hh)} \\ L^{(pi)} & L^{(ph)}
           \end{pmatrix},\
     B = \begin{pmatrix}
              L^{(ih)} & - L^{(ip)} \\ L^{(hh)} & - L^{(hp)}
           \end{pmatrix},\
     D = I_{2 \nex} + \begin{pmatrix}
                         - L^{(hh)} & L^{(hp)} \\ - L^{(ph)} & L^{(pp)}
                       \end{pmatrix}, \notag \\[1ex]
     & J = D + C (I_M + K^*)^{-1} B \epp
\end{align}
Finally define
\begin{align}
     {W^{(i)}}^j_k & = \frac{\re(\m_j - \la_k)}{1 + \fa_n (\la_k|h')} -
                       \frac{\re(\la_k - \m_j)}{1 + 1/\fa_n (\la_k|h')} \epc
		       && j = 1, \dots, M - \nex; k = 1, \dots, M \epc \notag \\[1ex]
     {W^{(p)}}^j_k & = \frac{\re(y_j - \la_k)}{1 + \fa_n (\la_k|h')} -
                       \frac{\re(\la_k - y_j)}{1 + 1/\fa_n (\la_k|h')} \epc
		       && j = 1, \dots, \nex; k = 1, \dots, M \epc \notag \\[1ex]
     {W^{(h)}}^j_k & = \frac{\re(x_j - \la_k)}{1 + \fa_n (\la_k|h')} -
                       \frac{\re(\la_k - x_j)}{1 + 1/\fa_n (\la_k|h')} \epc
		       && j = 1, \dots, \nex; k = 1, \dots, M
\end{align}
and
\begin{equation} \label{defgstarmat}
     G^* = (I_M + K^*)^{-1} \begin{pmatrix} W^{(i)} \\ W^{(h)} \end{pmatrix} \epc \qd
     D^* = W^{(p)} {G^*}^{-1} (I_M + K^*)^{-1} \binom{0}{I_{\nex}} \epp
\end{equation}
With this notation we can state our next lemma.
\begin{lemma} \label{lem:det2}
{\bf Another factorization of determinants.}\\
The remaining determinants in (\ref{start}) admit the factorization
\begin{subequations}
\label{secondsplit}
\begin{align} \label{splitexcnorm}
     & \det_M \biggl\{\de^j_k + \frac{K (\m_j - \m_k)}{\fa_n' (\m_k|h')}\biggr\} =
        \det_M \{I_M + K^*\} \det_{2\nex} \{ J \} \epc
	\\[2ex]
     & \frac{\det_M \Bigl\{\frac{\re(\m_j - \la_k)}{1 + \fa_n (\la_k|h')}
                         - \frac{\re(\la_k - \m_j)}{1 + 1/\fa_n (\la_k|h')}\Bigr\}}
          {\det_M \{I_M + K^*\}} = \det_M \{G^*\} \det_{\nex} \{D^*\} \epp
	  \label{secondfactorsplit}
\end{align}
\end{subequations}
\end{lemma}
\begin{proof}
The proof is similar to the proof of Lemma~\ref{lem:det1} and relies
on elementary row- and column manipulations of determinants.
\end{proof}

In lemma~\ref{lem:det1} and \ref{lem:det2} we have factorized the
determinants in the original representation (\ref{start}) of the
form factor amplitudes for the generating function in a way that
will allow us to take the Trotter limit and the zero temperature
limit. Two of the determinants are of size $M$, which goes to infinity
for $N \rightarrow \infty$. We shall call them the `large determinants'.
The three remaining determinants are of size $\nex$, or $2 \nex$ and
will be called the `small determinants'. In the following subsections
the five determinants will be considered one by one. We shall see that
the matrix elements of $G$, $G^*$, $(I_M + K)^{-1}$ and $(I_M + K^*)^{-1}$
can be expressed in terms of functions that solve linear integral
equations and can be explicitly calculated in the limit $T \rightarrow 0+$.

\subsection{The small determinant in the denominator}
We start with the $2 \nex \times 2 \nex$ determinant in (\ref{splitexcnorm})
whose structure is familiar to us from previous work \cite{DGKS16b,GKKKS17}.
We first of all observe that
\begin{multline} \label{densmalldet2}
     J = I_{2 \nex} \\[1ex] +
         \begin{pmatrix}
	    - L^{(hh)} + L^{(hv)} (I_M + K^*)^{-1} L^{(vh)} &
	    L^{(hp)} - L^{(hv)} (I_M + K^*)^{-1} L^{(vp)} \\
	    - L^{(ph)} + L^{(pv)} (I_M + K^*)^{-1} L^{(vh)} &
	    L^{(pp)} - L^{(pv)} (I_M + K^*)^{-1} L^{(vp)}
	 \end{pmatrix} \epc
\end{multline}
where
\begin{subequations}
\begin{align}
     L^{(hv)} & = \bigl(L^{(hi)}, L^{(hh)}\bigr) \epc &&
     L^{(pv)} = \bigl(L^{(pi)}, L^{(ph)}\bigr) \epc \\[1ex]
     L^{(vh)} & = \binom{L^{(ih)}}{L^{(hh)}} \epc &&
     L^{(vp)} = \binom{L^{(ip)}}{L^{(hp)}} \epp
\end{align}
\end{subequations}

In order to simplify (\ref{densmalldet2}) we define a resolvent kernel as the
solution of the linear integral equation
\begin{equation} \label{rstarrightint}
     R^* (\la,\m) = K(\la - \m) - \int_{\cal C} \frac{\rd \nu}{2 \p \i} 
                               \frac{K(\la - \nu) R^* (\nu, \m)}{1 + \fa_n (\nu|h')} \epp
\end{equation}
Recall that $R^* (\la,\m)$ then also satisfies the integral equation
\begin{equation} \label{rstarleftint}
     R^* (\la,\m) = K(\la - \m) - \int_{\cal C} \frac{\rd \nu}{2 \p \i} 
                               \frac{R^* (\la, \nu) K(\nu - \m)}{1 + \fa_n (\nu|h')} \epp
\end{equation}
Let $\la \in \Int {\cal C}$. The contour ${\cal C}$ is constructed in
such a way that $\m \in \Int {\cal C} \cup {\cal Y}_n$ implies that
$\m \pm \i \g \in \Ext {\cal C}$. Hence, $R^* (\cdot, \m)$ is holomorphic
in $\Int {\cal C}$ and on the boundary of this domain. Then the residue
theorem applied to (\ref{rstarrightint}) implies that
\begin{equation} \label{rstarrightsum}
     R^* (\la,\m) = K(\la - \m)
               - \sum_{l = 1}^M \frac{K(\la - z_l) R^* (z_l, \m)}
	                                {\fa_n' (z_l|h')} 
\end{equation}
for all $\m \in {\cal Z}_n \cup {\cal Y}_n$ and for all $\la \in \Int {\cal C}$.
Equation (\ref{rstarrightsum}) also defines the analytic continuation
of $R^* (\cdot, \m)$ to the complex plane. Similarly, (\ref{rstarleftint})
implies that
\begin{equation} \label{rstarleftsum}
     R^* (\la,\m) = K(\la - \m)
        - \sum_{l = 1}^M \frac{R^* (\la, z_l) K(z_l - \m)}
	                         {\fa_n' (z_l|h')}
\end{equation}
for all $\la \in {\cal Z}_n \cup {\cal Y}_n$, for all $\m \in
\Int {\cal C}$, and also for $\m \in {\mathbb C}$ for which the
right hand side is defined by analytic continuation.

Upon defining a matrix $R^*$ with matrix elements
\begin{equation}
     {R^*}^j_k = \frac{R^* (z_j, z_k)}{\fa_n' (z_k|h')}
\end{equation}
equation (\ref{rstarrightsum}) implies that
\begin{equation} \label{rstark}
     (I_M + K^*)^{-1} = I_M - R^* \epp
\end{equation}
This can be used to derive the following
\begin{lemma} \label{lem:measurefromdet}
{\bf Jacobian of higher level Bethe Ansatz equations \cite{GKKKS17}.}\\
The small determinant $\det_{2 \nex} \{J\}$ has
the representations
\begin{equation} \label{detjrform}
     \det_{2 \nex} \{ J\} = \det_{2 \nex}
        \begin{pmatrix}
	   \de^j_k - \frac{R^* (x_j, x_k)}{\fa_n' (x_k|h')} &
	   \frac{R^* (x_j, y_k)}{\fa_n' (y_k|h')} \\[2ex]
	   - \frac{R^* (y_j, x_k)}{\fa_n' (x_k|h')} &
	   \de^j_k + \frac{R^* (y_j, y_k)}{\fa_n' (y_k|h')}
        \end{pmatrix}
\end{equation}
and
\begin{equation}
     \det_{2 \nex} \{J\} =
        \Biggl[\prod_{j=1}^{\nex} \frac{-1}{\fa_n' (x_j|h') \fa_n' (y_j|h')}\Biggr]
        \det_{2 \nex} \{{\cal J}\} \epc
\end{equation}
where
\begin{equation} \label{defjtilde}
     \det_{2 \nex} \{{\cal J}\} = \det_{2 \nex}
	\begin{pmatrix}
	   \6_{u_k} \fa^- (u_j |{\cal U},{\cal V},h') &
	   \6_{v_k} \fa^- (u_j |{\cal U},{\cal V},h') \\
	   \6_{u_k} \fa^+ (v_j |{\cal U},{\cal V},h') &
	   \6_{v_k} \fa^+ (v_j |{\cal U},{\cal V},h')
	\end{pmatrix}_{\substack{{\cal U} = {\cal X}_n \\ {\cal V} = {\cal Y}_n}} \epp
\end{equation}
\end{lemma}
\begin{proof}
Let us consider, for instance, the upper left element of $J$. First
of all, (\ref{rstarrightsum}) implies that
\begin{equation}
     \sum_{l = 1}^M {L^{(hv)}}^j_l \bigl[(I_M + K^*)^{-1}\bigr]^l_k =
       \sum_{l = 1}^M \frac{K(x_j - z_l)}{\fa_n' (z_l|h')} (\de^l_k - {R^*}^l_k) =
       \frac{R^* (x_j, z_k)}{\fa_n' (z_k|h')}
\end{equation}
for $j = 1, \dots, \nex$, $k = 1, \dots, M$. Using (\ref{rstarleftsum})
we infer
\begin{multline}
     \sum_{l, m = 1}^M {L^{(hv)}}^j_l \bigl[(I_M + K^*)^{-1}\bigr]^l_m {L^{(vh)}}^m_k \\
        = \sum_{m=1}^M \frac{R^* (x_j, z_m)}{\fa_n' (z_m|h')}
                     \frac{K(z_m - x_k)}{\fa_n' (x_k|h')} =
        \frac{K(x_j - x_k)}{\fa_n' (x_k|h')} - \frac{R^* (x_j, x_k)}{\fa_n' (x_k|h')}
\end{multline}
which is equivalent to
\begin{equation}
     \de^j_k - {L^{(hh)}}^j_k + \bigl[{L^{(hv)}} (I_M + K^*)^{-1} {L^{(vh)}}\bigr]^j_k
        = \de^j_k - \frac{R^* (x_j, x_k)}{\fa_n' (x_k|h')} \epp
\end{equation}
Thus, we have obtained the upper left element of $J$. The other
elements are obtained in a similar way.

Taking the derivative of (\ref{nlieoffshell1}) with respect to $u_k$ we
obtain
\begin{multline}
     \frac{\6_{u_k} \fa^- (\la|{\cal U},{\cal V},h')}{\fa^- (\la|{\cal U},{\cal V},h')} =
        - \frac{\6_{u_k} \fa^+ (\la|{\cal U},{\cal V},h')}{\fa^+ (\la|{\cal U},{\cal V},h')}
	\\[1ex]
	= - K(\la - u_k)
	  - \int_{\cal C} \frac{\rd \m}{2 \p \i} K(\la - \m)
	    \frac{\6_{u_k} \fa^- (\m|{\cal U},{\cal V},h')}
	         {1 + \fa^- (\m|{\cal U},{\cal V},h')} \epp
\end{multline}
Using the first equation in the second equation and comparing with
equation (\ref{rstarrightint}) we conclude that
\begin{equation}
     R^* (\la, x_k) = \pm
        \frac{\6_{u_k} \fa^\pm (\la|{\cal U},{\cal V},h')}{\fa^\pm (\la|{\cal U},{\cal V},h')}
	\biggl|_{\substack{{\cal U} = {\cal X}_n \\ {\cal V} = {\cal Y}_n}} \epp
\end{equation}
Thus, using the Bethe Ansatz equations,
\begin{equation}
     \frac{\6_{u_k} \fa^+ (u_j |{\cal U},{\cal V},h')}{\fa_n' (x_k|h')}
	\biggl|_{\substack{{\cal U} = {\cal X}_n \\ {\cal V} = {\cal Y}_n}} =
     - \frac{\6_{u_k} \fa^- (u_j |{\cal U},{\cal V},h')}{\fa_n' (x_k|h')}
	\biggl|_{\substack{{\cal U} = {\cal X}_n \\ {\cal V} = {\cal Y}_n}} =
	\de^j_k - \frac{R^* (\la_j, x_k)}{\fa_n' (x_k|h')} \epp
\end{equation}
The other entries of the matrix $J$ can be treated in a similar way.
\end{proof}

Note that Lemma~\ref{lem:measurefromdet} will allow us to convert our
sum over excited states (\ref{defgenfun}) into a sum over classes of
excited states, each class being represented by a multiple integral.

\subsection{\boldmath The functions $G$ and $G^*$}
In order to prepare for the calculation of the remaining determinants
in (\ref{firstfactorsplit}) and in (\ref{secondfactorsplit}) we
introduce a pair of functions $G$, $G^*$ which, for every 
$\m \in \Int {\cal C}$, are defined as the solutions of the linear
integral equation
\begin{subequations}
\label{defggstar}
\begin{align} \label{defg}
     G (\la,\m) & = \re(\m - \la) - \int_{\cal C} \frac{\rd \nu}{2 \p \i} 
                              \frac{K(\la - \nu) G (\nu, \m)}{1 + \fa (\nu|h)} \epc \\[1ex]
     G^* (\la,\m) & = \re(\m - \la)
                 - \int_{\cal C} \frac{\rd \nu}{2 \p \i} 
		      \frac{K(\la - \nu) G^* (\nu, \m)}{1 + \fa_n (\nu|h')} \epp
		      \label{defgstar}
\end{align}
\end{subequations}
\begin{lemma} \label{lem:matfung}
{\bf \boldmath Matrices and functions $G$ and $G^*$.}\\ The matrix
elements of $G$ and $G^*$ defined in equations (\ref{defgmat})
and (\ref{defgstarmat}) are naturally expressed in terms of the
function $G$ and $G^*$ defined in (\ref{defggstar}),
\begin{equation} \label{ggstarmatfun}
     G^j_k = - G(\la_j, z_k) \epc \qd {G^*}^j_k = - G^* (z_j,\la_k) \epc
\end{equation}
$j, k = 1, \dots, M$.
\end{lemma}
\begin{proof}
Recall that $\re(\m - \la) = \ctg(\m - \la) - \ctg(\m - \la - \i \g)$.
Thus, $\m - \i \g \in \Ext {\cal C}$, if $\m \in \Int {\cal C}$ and
$G(\cdot, \m)$ then has a single simple pole with residue $- 1$ at $\m$.
It follows that
\begin{equation} \label{gsum}
     G(\la,\m) = \frac{\re(\m - \la)}{1 + 1/\fa (\m|h)} -
              \frac{\re(\la - \m)}{1 + \fa (\m|h)}
             - \sum_{l = 1}^M \frac{K(\la - \la_l) G (\la_l, \m)}
	                              {\fa' (\la_l|h)} \epp
\end{equation}
Here we have used the second equation (\ref{defk}). We compare
(\ref{gsum}) with equation (\ref{defgmat}), which is equivalent to
\begin{equation}
     G^j_k + \sum_{l = 1}^M \frac{K(\la_j - \la_l) G^l_k}{\fa' (\la_l|h)}
            = \frac{\re(\la_j - z_k)}{1 + \fa (z_k|h)}
	      - \frac{\re(z_k - \la_j)}{1 + 1/\fa (z_k|h)} \epp
\end{equation}
Then the first equation (\ref{ggstarmatfun}) follows.

Similarly,
\begin{equation} \label{gstarsum}
     G^* (\la,\m) = \frac{\re(\m - \la)}{1 + 1/\fa_n (\m|h')} -
                 \frac{\re(\la - \m)}{1 + \fa_n (\m|h')}
                - \sum_{l = 1}^M \frac{K(\la - z_l) G^* (z_l, \m)}
	                                 {\fa_n' (z_l|h')} \epc
\end{equation}
while, on the other hand, the definition of the matrix $G^*$ in
(\ref{defgstarmat}) is equivalent to
\begin{equation} \label{gstarmatsum}
     {G^*}^j_k + \sum_{l = 1}^M \frac{K(z_j - z_l) {G^*}^l_k}
                                        {\fa_n' (z_l|h')}
        = \frac{\re(z_j - \la_k)}{1 + \fa_n (\la_k|h')} -
	  \frac{\re(\la_k - z_j)}{1 + 1/\fa_n (\la_k|h')} \epp
\end{equation}
Comparison of (\ref{gstarsum}) and (\ref{gstarmatsum}) shows that
the second equation (\ref{ggstarmatfun}) holds as well.
\end{proof}

For the actual calculation of the determinants of the matrices $G$, $G^*$
and of their inverses that appear in the remaining small determinants,
we need a closer characterization of the functions $G$, $G^*$. Important
tools in this context are the resolvent kernels $R^*$ defined in
(\ref{rstarrightint}) and its counterpart associated with the dominant
state, which may be defined as the solution of the integral equation
\begin{equation} \label{rrightint}
     R(\la,\m) = K(\la - \m) - \int_{\cal C} \frac{\rd \nu}{2 \p \i} 
                            \frac{R(\la, \nu) K(\nu - \m)}{1 + \fa (\nu|h)} \epp
\end{equation}
\begin{lemma} \label{lem:gresolve}
{\bf Representation of the functions $G$ and $G^*$ by means of resolvent
kernels and implications.}
\begin{enumerate}
\item
The functions $G$ and $G^*$ introduced in (\ref{defggstar}) can be
represented as
\begin{subequations}
\label{ggstarresolv}
\begin{align}
     G (\la,\m) & = \re(\m - \la) - \int_{\cal C} \frac{\rd \nu}{2 \p \i} 
                              \frac{R(\la,\nu) \re (\m - \nu)}{1 + \fa (\nu|h)} \epc \\[1ex]
     G^* (\la,\m) & = \re(\m - \la)
                 - \int_{\cal C} \frac{\rd \nu}{2 \p \i} 
		      \frac{R^* (\la,\nu) \re(\m - \nu)}{1 + \fa_n (\nu|h')} \epp
\end{align}
\end{subequations}
\item
For $\la \in \Int {\cal C}$ the functions $G(\la, \cdot)$ and $G^* (\la, \cdot)$
are meromorphic inside ${\cal C}$, where they both have a single simple pole with
residue $+1$ at $\la$.
\end{enumerate}
\end{lemma}
\begin{proof}
(i) is an immediate consequence of the definitions of the resolvent kernels.
(ii) holds, since $R(\la, \cdot)$ and $R^*(\la, \cdot)$ are holomorphic on
and inside ${\cal C}$ as can be seen from (\ref{rstarrightint}),
(\ref{rrightint}).
\end{proof}
Define two functions
\begin{multline}
     R_0 (\la) = \frac{1}{2\p} \\ + \frac{1}{4 \p \g}
                 \Bigl\{\ps_{q^4} \Bigl(\tst{1 + \frac{\i \la}{2 \g}}\Bigr)
                         + \ps_{q^4} \Bigl(\tst{1 - \frac{\i \la}{2 \g}}\Bigr)
                         - \ps_{q^4} \Bigl(\tst{\2 + \frac{\i \la}{2 \g}}\Bigr)
                         - \ps_{q^4} \Bigl(\tst{\2 - \frac{\i \la}{2 \g}}\Bigr)\Bigr\} \epc
\end{multline}
where $\ps_q (\la) = \6_\la \ln \G_q (\la)$ is the $q$-digamma function, and
\begin{equation} \label{defgalpha}
     g_\a (\la,\m) = \frac{\dh_1' \dh_2 (\m - \la - \a)}{\dh_2 (\a) \dh_1 (\m - \la)} \epc
\end{equation}
where $\a \in {\mathbb C}$ is a parameter (recall the definitions of the
$q$-gamma, $q$-Barnes and theta functions in Section~\ref{sec:lowtfun}).
\begin{lemma} \label{lem:ggstarlowt}
{\bf \boldmath Low-$T$ form of the integral equations for $G$ and $G^*$.}
\begin{enumerate}
\item
Let $\m \in \Int {\cal C}$. The functions $G(\cdot, \m)$, $G^*(\cdot, \m)$
satisfy the linear integral equations
\begin{subequations}
\label{ggstarlowtform}
\begin{align} \label{glowtform}
     G(\la,\m) & = g_0 (\la,\m) -
	\sum_{\s = \pm} \s
	   \int_{{\cal C}_\s} \rd \nu
	      \frac{R_0 (\la - \nu) G (\nu, \m)}{1 + \fa^\s (\nu|h)} \epc \\[1ex]
     G^* (\la,\m) & = g_0 (\la,\m) -
	\sum_{\s = \pm} \s
	   \int_{{\cal C}_\s} \rd \nu
	      \frac{R_0 (\la - \nu) G^* (\nu, \m)}{1 + \fa_n^\s (\nu|h')} \epp
	      \label{gstarlowtform}
\end{align}
\end{subequations}
\item
If $\m \in \Int {\cal C}$, then
\begin{equation} \label{glim}
     G(\la, \m) = g_0 (\la, \m) + \CO(T^\infty) \epc \qd
     G^* (\la, \m) = g_0 (\la, \m) + \CO(T^\infty) \epp
\end{equation}
\item
The functions $G(\la, \cdot)$, $G^* (\la, \cdot)$ can be analytically
continued to the upper half plane. In particular, for $\m \in {\cal Y}_n$
we have the representation
\begin{equation} \label{gcontinuedmu}
     G(\la,\m) = g_0 (\la,\m) + \frac{2 \p \i R_0 (\la - \m)}{1 + \fa (\m|h)} - 
	\sum_{\s = \pm} \s
	   \int_{{\cal C}_\s} \rd \nu
	      \frac{R_0 (\la - \nu) G (\nu, \m)}{1 + \fa^\s (\nu|h)} \epc
\end{equation}
implying the low-$T$ asymptotic behaviour
\begin{equation} \label{glimyparticle}
     G(\la, \m) = g_0 (\la,\m)
        + \frac{2 \p \i R_0 (\la - \m)}{1 - \frac{\fa(\m|h)}{\fa_n (\m|h')}}
	+ \CO(T^\infty) \epp
\end{equation}
\end{enumerate}
\end{lemma}
\begin{proof}
(i) The functions under the integral on the right hand side of (\ref{defg})
are $\p$-periodic. Hence, the contributions from ${\cal C}_l$
and ${\cal C}_r$ cancel each other, and we may replace ${\cal C}$ by
${\cal C}_+ + {\cal C}_-$. For $T \rightarrow 0+$ the function
$\fa (\la|h)$ behaves as
\begin{equation}
     \fa (\la|h) = \re^{- \frac{\e(\la|h)}{T}} \bigl(1 + {\cal O} (T^\infty)\bigr) \epc
\end{equation}
(see (\ref{anlowt})) and $\Re \e(\la|h) < 0$ on ${\cal C}_+$, $\Re \e(\la|h) > 0$
on ${\cal C}_-$. This suggests to decompose
\begin{multline}
     \int_{\cal C} \frac{\rd \nu}{2 \p \i}
        \frac{K(\la - \nu) G (\nu, \m)}{1 + \fa (\nu|h)} = 
     \sum_{\s = \pm} \int_{{\cal C}_\s} \frac{\rd \nu}{2 \p \i}
                        \frac{K(\la - \nu) G (\nu, \m)}{1 + \fa (\nu|h)} \\ =
     \int_{{\cal C}_-} \frac{\rd \nu}{2 \p \i} K(\la - \nu) G (\nu, \m) +
     \sum_{\s = \pm} \s \int_{{\cal C}_\s} \frac{\rd \nu}{2 \p \i}
                           \frac{K(\la - \nu) G (\nu, \m)}{1 + \fa^\s (\nu|h)}
\end{multline}
and to rewrite the integral equation (\ref{defg}) as
\begin{multline} \label{intgsplit}
     G(\la,\m) + \int_{{\cal C}_-} \frac{\rd \nu}{2 \p \i} K(\la - \nu) G (\nu, \m) \\ =
        \re(\m - \la) -
	\sum_{\s = \pm} \s \int_{{\cal C}_\s} \frac{\rd \nu}{2 \p \i}
                              \frac{K(\la - \nu) G (\nu, \m)}{1 + \fa^\s (\nu|h)} \epp
\end{multline}

This can be further transformed using Fourier series techniques.
The functions $K(\la)$ and $\re(\m - \la)$ have Fourier series
representations
\begin{equation} \label{fourierke}
     K(\la) = \sum_{l = - \infty}^\infty K_l \re^{2 \i l \la} \epc \qd
     \re(\m - \la) = \sum_{l = - \infty}^\infty \re_l (\m) \re^{2 \i l \la} \epc
\end{equation}
where
\begin{equation}
     K_l = 2 \i \re^{- 2 \g |l|} \epc \qd
     \re_l (\m) = 2 \i \begin{cases}
                            0 & l \le 0 \\
			    \re^{- 2 \i l \m} - \re^{- 2 \i l (\m - \i \g)}
			    & l > 0
                          \end{cases} \epp
\end{equation}
For the derivation recall that $\m \in \Int {\cal C}$. Inserting
the Fourier series (\ref{fourierke}) into (\ref{intgsplit}) we see
that $G(\la,\m)$ has a Fourier series representation
\begin{equation}
     G(\la,\m) = \sum_{l = - \infty}^\infty g_l (\m) \re^{2 \i l \la}
\end{equation}
with Fourier coefficients that satisfy the equation
\begin{equation} \label{gfouriercoeffs}
     g_l (\m) + \frac{K_l}{2\i}
        \int_{{\cal C}_-} \frac{\rd \nu}{\p} \re^{- 2 \i l \nu} G (\nu, \m) =
        \re_l (\m) -
	\sum_{\s = \pm} \s
	   \int_{{\cal C}_\s} \frac{\rd \nu}{2 \p \i}
	      \frac{K_l \re^{- 2 \i l \nu} G (\nu, \m)}{1 + \fa^\s (\nu|h)} \epp
\end{equation}
Using that $\m \in \Int {\cal C}$ we see that
\begin{equation}
     \int_{{\cal C}_-} \frac{\rd \nu}{\p} \re^{- 2 \i l \nu} G (\nu, \m) =
        - 2 \i \re^{- 2 \i l \m} + g_l (\m) \epp
\end{equation}
Inserting this back into (\ref{gfouriercoeffs}), solving for $g_l (\m)$
and performing the back transformation we arrive at (\ref{glowtform}).

The integral equation (\ref{defgstar}) for $G^* (\cdot, \m)$ is of the
same form as the integral equation (\ref{defg}) for $G(\cdot,\m)$, just
$\fa (\cdot|h)$ is replaced by $\fa_n (\cdot|h')$ whose low-$T$ behaviour
is displayed in (\ref{anlowt}). As above, one can use the $\p$-periodicity
of the functions $G^* (\cdot, \m)$, $\re$ and $K$ to argue that
(\ref{gstarlowtform}) must hold.

(ii) The limits (\ref{glim}) are an obvious implication of equations
(\ref{ggstarlowtform}).

(iii) Using Lemma~\ref{lem:gresolve} and analytically continuing
$G(\la, \cdot)$ through ${\cal C}_+$ we see that (\ref{gcontinuedmu})
must hold. The limit (\ref{glimyparticle}) follows from (\ref{gcontinuedmu}),
using (\ref{anlowt}) and the fact that $\fa_n (\m|h') = - 1$, since
$\m \in {\cal Y}_n$.
\end{proof}
Following the same procedure as above we obtain the low-$T$ form
of the integral equation for $R^*$ which will be needed below.
\begin{lemma}
{\bf \boldmath Low-$T$ limit of $R^*$.}
\begin{equation} \label{rstarlowtform}
     R^* (\la,\m) = 2 \p \i R_0 (\la - \m) -
	\sum_{\s = \pm} \s \int_{{\cal C}_\s} \rd \n \:
	   \frac{R_0 (\la - \n) R^* (\n,\m)}{1 + \fa_n^\s (\n|h')} \epc
\end{equation}
implying that
\begin{equation} \label{rstarzero}
     R^* (\la,\m) = 2 \p \i R_0 (\la - \m) + \CO(T^\infty) \epp
\end{equation}
\end{lemma}

\subsection{\boldmath Determinant and inverse of $g_0$}
The determinants of the matrices with matrix elements ${G_0}^j_k
= g_0 (\la_j, z_k)$ and ${G_0^*}^j_k = g_0 (z_j, \la_k)$ and the
inverses of these matrices can be evaluated in explicit form.
This will be important for our further reasoning. In order to write
the resulting expressions compactly we introduce the shorthand notation
\begin{equation} \label{defsi}
     \Si = \sum_{j=1}^M (\la_j - z_j)
\end{equation}
and the function
\begin{equation} \label{defphi}
     \Ph (\la) = \prod_{j=1}^M
                 \frac{\dh_1 (\la - \la_j)}{\dh_1 (\la - z_j)} \epp
\end{equation}

\begin{lemma}
{\bf \boldmath Properties of the matrices $G_0$ and $G_0^*$.}
\begin{enumerate}
\item
The elliptic Cauchy determinant.
\begin{equation} \label{detgzeromain}
     \det_M \{ g_0 (\la_j, z_k) \} =
	\frac{\dh_2 (\Si) \dh_1'^M}{\dh_2}
        \frac{\prod_{1 \le j < k \le M} \dh_1 (\la_j - \la_k) \dh_1 (z_k - z_j)}
	     {\prod_{j,k = 1}^M \dh_1 (z_j - \la_k)} \epp
\end{equation}
\item
The inversion formulae.
\begin{equation} \label{gzerozerostarinverse}
     {G_0^{-1}}^j_k = \frac{g_\Si (z_j, \la_k)}{(1/\Ph)' (z_j) \Ph' (\la_k)} \epc \qd
     {{G_0^*}^{-1}}^j_k = \frac{g_{- \Si} (\la_j, z_k)}{\Ph' (\la_j) (1/\Ph)' (z_k)} \epp
\end{equation}
\item
The square of the determinant.
\begin{equation} \label{thesquare}
     \det_M \{ g_0 (\la_j, z_k) \} \det_M \{ g_0 (z_j, \la_k) \} =
	\frac{\dh_2^2 (\Si)}{\dh_2^2} \prod_{j=1}^M \Ph' (\la_j) (1/\Ph)' (z_j) \epp
\end{equation}
\end{enumerate}
\end{lemma}
\begin{proof}
(i) The determinant formula is a variant of a classical result due to
Frobenius \cite{Frobenius1882}. It appeared in the literature on the
XXZ chain in \cite{KMT99b} and has many interesting generalizations 
\cite{RoSc06}. For the sake of self-containedness of this work we
present a proof in Appendix~\ref{app:ellipticcauchy}.

(ii) is obtained by Cramer's rule and (\ref{detgzeromain}). (iii) is
a direct consequence of (\ref{detgzeromain}) and the definition
(\ref{defphi}) of $\Ph$.
\end{proof}

\subsection{The large determinants}
For $\la, \m \in \Int {\cal C}$ and else by analytic continuation
define a pair of kernel functions
\begin{subequations}
\label{defukernels}
\begin{align}
     U(\la,\m) & = \sum_{\s = \pm}
                \s \int_{{\cal C}_\s} \frac{\rd \x}{1 + \fa^\s (\x|h)}
		   \int_{\cal C} \frac{\rd \z}{2 \p \i \Ph (\z)} \,
		      g_\Si (\la, \z) R_0 (\z - \x) G(\x,\m) \epc \\[1ex]
     U^* (\la,\m) & = \sum_{\s = \pm}
                   \s \int_{{\cal C}_\s} \frac{\rd \x}{1 + \fa_n^\s (\x|h')}
		      \int_{\cal C} \frac{\rd \z \: \Ph (\z)}{2 \p \i}
		         g_{- \Si} (\la, \z) R_0 (\z - \x) G^* (\x,\m) \epp
\end{align}
\end{subequations}

\begin{lemma}
{\bf \boldmath Low-$T$ factorization of $G$ and $G^*$.}
\begin{subequations}
\begin{align} \label{lowfactorg}
     G^j_k & = - \sum_{l = 1}^M g_0 (\la_j, z_l)
        \biggl(\de^l_k - \frac{U(z_l, z_k)}{(1/\Ph)'(z_l)} \biggr) \epc \\
     {G^*}^j_k & = - \sum_{l = 1}^M g_0 (z_j, \la_l)
        \biggl(\de^j_k - \frac{U^* (\la_l, \la_k)}{\Ph'(\la_l)} \biggr) \epp
	\label{lowfactorgstar}
\end{align}
\end{subequations}
\end{lemma}
\begin{proof}
The proof is by straightforward calculation,
\begin{multline} \label{ginvonrzero}
     \sum_{l = 1}^M g_0 (\la_j, z_l) \frac{U(z_l, z_k)}{(1/\Ph)'(z_l)} \\
        = \sum_{\s = \pm} \s \int_{{\cal C}_\s} \frac{\rd \x}{1 + \fa^\s (\x|h)}
          \sum_{l = 1}^M \frac{g_0 (\la_j, z_l)}{(1/\Ph)'(z_l)}
          \sum_{m = 1}^M \frac{g_\Si (z_l, \la_m)}{\Ph'(\la_m)}
	                 R_0 (\la_m - \x) G(\x, z_k) \\[1ex]
        = \sum_{\s = \pm} \s \int_{{\cal C}_\s} \frac{\rd \x}{1 + \fa^\s (\x|h)}
	                     R_0 (\la_j - \x) G(\x, z_k)
        = g_0 (\la_j, z_k) - G(\la_j, z_k) \epp
\end{multline}
Here we have inserted the definition of $U$ in the first equation,
the inversion formula (\ref{gzerozerostarinverse}) in the second equation
and Lemma~\ref{lem:ggstarlowt} in the third equation. Equation
(\ref{lowfactorg}) then follows with Lemma~\ref{lem:matfung}. The proof
of (\ref{lowfactorgstar}) is similar.
\end{proof}

\begin{corollary}
{\bf \boldmath Low-$T$ form of the large determinants.}
\begin{subequations}
\label{detggstarlowt}
\begin{align}
     \det_M \{G\} & = \det_M \{- g_0 (\la_j, z_k)\}
        \det_M \biggl\{\de^j_k - \frac{U(z_j, z_k)}{(1/\Ph)'(z_j)} \biggr\} \epc \\[1ex]
     \det_M \{G^*\} & = \det_M \{- g_0 (z_j, \la_k)\}
        \det_M \biggl\{\de^j_k - \frac{U^* (\la_j, \la_k)}{\Ph'(\la_j)} \biggr\} \epp
\end{align}
\end{subequations}
\end{corollary}
This form is now suitable for taking the Trotter limit and the limit
$T \rightarrow 0+$, since the first factors on the right hand side
can be written as products (see (\ref{detgzeromain})) and since the
kernel functions $U$, $U^*$ vanish for $T \rightarrow 0+$.

\subsection{The small determinants}
The kernels $U(\la, y)$ and $U^* (\la,y)$ are holomorphic functions of
both of their arguments $\la$, $\m$ for $\la, \m \in \Int {\cal C}$. For
the first argument this is obvious from the definition (\ref{defukernels}),
for the second argument it follows from Lemma~\ref{lem:gresolve}.
Let us consider the corresponding resolvent kernels defined as solutions
of the linear integral equations
\begin{subequations}
\label{defskernels}
\begin{align}
     S(\la,\m) & =
        U(\la,\m) + \int_{\cal C} \frac{\rd \nu \: \Ph (\nu)}{2 \p \i}
	                          U(\la,\nu) S(\nu,\m) \epc \\[1ex]
     S^* (\la,\m) & =
        U^* (\la,\m) + \int_{\cal C} \frac{\rd \nu}{2 \p \i \Ph (\nu)}
                                    U^* (\la,\nu) S^* (\nu,\m) \epp
\end{align}
\end{subequations}
These resolvents can be used to derive representations of the matrices
$G^{-1}$ and ${G^*}^{-1}$ that are suitable for performing the low-$T$
limit.
\begin{lemma}
{\bf \boldmath Low-$T$ factorization of $G^{-1}$ and ${G^*}^{-1}$.}
The inverse matrices $G^{-1}$ and ${G^*}^{-1}$ admit the following
representation in terms of the resolvent $S$,
\begin{subequations}
\begin{align} \label{lowfactorginv}
     {G^{-1}}^j_k & =
        - \frac{g_\Si (z_j, \la_k)}{(1/\Ph)' (z_j) \Ph' (\la_k)}
	- \int_{\cal C} \frac{\rd \nu \: \Ph (\nu)}{2 \p \i}
          \frac{S(z_j, \nu) g_\Si (\nu, \la_k)}{(1/\Ph)' (z_j) \Ph' (\la_k)} \epc \\[1ex]
     {{G^*}^{-1}}^j_k & =
        - \frac{g_{- \Si} (\la_j, z_k)}{\Ph' (\la_j) (1/\Ph)' (z_k)}
	- \int_{\cal C} \frac{\rd \nu}{2 \p \i \Ph (\nu)}
          \frac{S^*(\la_j, \nu) g_{- \Si} (\nu, z_k)}{\Ph' (\la_j) (1/\Ph)' (z_k)} \epp
	  \label{lowfactorgstarinv}
\end{align}
\end{subequations}
\end{lemma}
\begin{proof}
Clearly $S(\la,\m)$ is a holomorphic function of $\la$ for $\la \in \Int {\cal C}$.
Thus,
\begin{equation}
     S(\la,\m) = U(\la,\m) + \sum_{l = 1}^M
                          \frac{U(\la,z_l)S(z_l,y)}{(1/\Ph)'(z_l)} \epc
\end{equation}
implying that
\begin{equation} \label{resu}
     \sum_{l = 1}^M
        \biggl(\de^j_l - \frac{U(z_j,z_l)}{(1/\Ph)'(z_j)}\biggr)
	\biggl(\de^l_k + \frac{S(z_l,z_k)}{(1/\Ph)'(z_l)}\biggr)
	   = \de^j_k \epp
\end{equation}
Combining the latter equation with (\ref{gzerozerostarinverse}) and
(\ref{lowfactorg}) we see that
\begin{equation}
     {G^{-1}}^j_k = - \sum_{l = 1}^M
                      \biggl(\de^j_l + \frac{S(z_j,z_l)}{(1/\Ph)'(z_j)}\biggr)
		      \frac{g_\Si (z_l, \la_k)}{(1/\Ph)' (z_l) \Ph' (\la_k)} \epc
\end{equation}
which is equivalent to (\ref{lowfactorginv}). The proof of
(\ref{lowfactorgstarinv}) is similar.
\end{proof}

We fix a contour ${\cal C}' \subset \Int {\cal C}$ in such a way
that $\{\la\} \subset \Int {\cal C}'$. In preparation of the next
lemma we define for every $\la \in \Ext {\cal C}'$ a function
\begin{equation} \label{defh}
     H(\la,x_k) = g_{- \Si} (\la,x_k)
        - \int_{{\cal C}'} \frac{\rd \z \: \Ph (\z)}{2 \p \i}
          g_{- \Si} (\la,\z) R^* (\z,x_k) \frac{(1/\Ph)' (x_k)}{\fa_n' (x_k|h')} \epp
\end{equation}

\begin{lemma}
{\bf \boldmath Low-$T$ form of the matrices in the small determinants.}
\begin{enumerate}
\item
The matrix in the first small determinant can be recast as
\begin{multline} \label{firstsmalldetlow}
     {\cal D}^j_k = (1/\Ph)'(x_j) D^j_k =
	   \int_{\cal C} \frac{\rd \z}{2 \p \i \Ph(\z)} g_\Si (x_j,\z) G(\z,y_k) \\[1ex]
	 + \int_{\cal C} \frac{\rd \x \: \Ph(\x)}{2 \p \i}
	   \int_{{\cal C}'} \frac{\rd \z}{2 \p \i \Ph(\z)}
	      S(x_j,\x) g_\Si (\x,\z) G(\z,y_k) \epc
\end{multline}
$j, k = 1, \dots, \nex$.
\item
The matrix in the second small determinant can be expressed as
\begin{multline} \label{scdsmalldetzero}
     {\cal D^*}^j_k = (1/\Ph)'(x_k) {D^*}^j_k =
        R^* (y_j,x_k) \frac{(1/\Ph)' (x_k)}{\fa_n' (x_k|h')} \\[1ex]
        + \int_{\cal C} \frac{\rd \x}{2 \p \i \Ph (\x)}
	  \biggl(
	  \frac{\re(\x - y_j)}{1 + 1/\fa_n (\x|h')} -
          \frac{\re(y_j - \x)}{1 + \fa_n (\x|h')} \biggr) \\ \times
	  \biggl(H(\x,x_k) + \int_{\cal C} \frac{\rd \z}{2 \p \i \Ph (\z)}
	                        S^* (\x,\z) H(\z,x_k) \biggr) \epc
\end{multline}
$j, k = 1, \dots, \nex$.
\end{enumerate}
\end{lemma}
\begin{proof}
(i) It follows from Lemma~\ref{lem:gresolve} that $G(\la,\cdot)$ can be analytically
continued to $\Ext {\cal C}$. Using (\ref{defg}) we conclude that
\begin{equation}
     \sum_{l=1}^M (\de^j_l + K^j_l) G(\la_l,y_k) =
        \frac{\re(y_k - \la_j)}{1 + 1/\fa (y_k|h)} -
	\frac{\re(\la_j - y_k)}{1 + \fa (y_k|h)} = - {V^{(p)}}^j_k \epp
\end{equation}
Hence,
\begin{equation} \label{smalldetg}
     D^j_k = \bigl[(0, I_{\nex}) G^{-1} (I_M + K)^{-1} V^{(p)}\bigr]^j_k =
        - \sum_{m=1}^M {G^{-1}}^{M - \nex +j}_m G(\la_m,y_k) \epc
\end{equation}
$j, k = 1, \dots, \nex$. Inserting (\ref{lowfactorginv}) on the right
hand side and using the residue theorem we end up with (\ref{firstsmalldetlow}).

(ii) For the second matrix $D^*$ we start with the definition (\ref{defgstarmat})
and insert (\ref{rstark}),
\begin{multline} \label{smalldetgstar}
     {D^*}^j_k =
     \biggl[ W^{(p)} {G^*}^{-1} (I_M + K^*)^{-1} \binom{0}{I_{\nex}} \biggr]^j_k = \\[1ex]
        \sum_{i, m = 1}^M \biggl(
          \frac{\re(y_j - \la_i)}{1 + \fa_n (\la_i|h')} -
	  \frac{\re(\la_i - y_j)}{1 + 1/\fa_n (\la_i|h')}\biggr)
	  {{G^*}^{-1}}^i_m \biggl(\de^m_{M - \nex + k} -
	  \frac{R^* (z_m,x_k)}{\fa_n' (x_k|h')}\biggr) \epc
\end{multline}
$j, k = 1, \dots, \nex$. Then we substitute (\ref{lowfactorgstarinv})
and convert the sums into integrals by means of the residue theorem.
The only slightly subtle part of the calculation is the following:
\begin{multline} \label{wpgid}
     \sum_{l = 1}^M \biggl(
        \frac{\re(y_j - \la_l)}{1 + \fa_n (\la_l|h')} -
	\frac{\re(\la_l - y_j)}{1 + 1/\fa_n (\la_l|h')}\biggr)
	\frac{g_{- \Si} (\la_l, z_m)}{\Ph' (\la_l)} \\
     - \int_{\cal C} \frac{\rd \x}{2 \p \i \Ph(\x)} \biggl(
        \frac{\re(y_j - \x)}{1 + \fa_n (\x|h')} -
	\frac{\re(\x - y_j)}{1 + 1/\fa_n (\x|h')}\biggr) g_{- \Si} (\x, z_m) \\[1ex]
     = \re(z_m - y_j) \lim_{\x \rightarrow z_m} \frac{1/\Ph(\x)}{1 + 1/\fa_n (\x|h')}
       - \re(y_j - z_m) \lim_{\x \rightarrow z_m} \frac{1/\Ph(\x)}{1 + \fa_n (\x|h')} \\[1ex]
     = K(y_j - z_m) \frac{(1/\Ph)' (z_m)}{\fa_n' (z_m|h')} \epp
\end{multline}
Here the zeros of $1 + \fa_n (\x|h')$ are canceled by the poles of
$\Ph (\x)$ and only the pole of $g_{- \Si} (\x, z_m)$ at $z_m$ contributes
to the sum over all residues. In the last equation we have used
(\ref{defk}). Taking into account (\ref{wpgid}) and the definition
(\ref{defh}) of the function $H$, we easily see, that (\ref{smalldetgstar})
implies (\ref{scdsmalldetzero}).
\end{proof}

\subsection{An intermediate summary}
At this point we have factorized the determinants in a way that
will allow us to take the Trotter limit and low-$T$ limit. What remains
to be done is to collect all the prefactors and to also rewrite them in
a form that is appropriate for taking these limits. For this purpose we
introduce the function
\begin{equation} \label{defphis}
     \Ph_s (\la) = \prod_{k = 1}^M \frac{\sin (\la - \la_k)}{\sin (\la - z_k)}
\end{equation}
and two matrices $U$ and $U^*$ with matrix elements
\begin{equation} \label{defmatsu}
        U^j_k = \frac{U(z_j, z_k)}{(1/\Ph)'(z_j)} \epc \qd
        {U^*}^j_k = \frac{U^* (\la_j, \la_k)}{\Ph'(\la_j)} \epp
\end{equation}
Then the amplitudes (\ref{ampgenfun}) can be represented as
\begin{multline} \label{prelimsum}
     A_n (h,h') = \det_M (I_M - U) \det_M (I_M - U^*)
        \frac{\det_\nex ({\cal D}) \det_\nex ({\cal D}^*)}{\det_{2 \nex} ({\cal J})}
	\\ \times
	(-1)^\nex \biggl[ \prod_{j=1}^\nex (1 + \fa (x_j|h))(1 + \fa(y_j|h)) \biggr]
	\biggl[ \prod_{j=1}^\nex
	        \frac{\fa_n' (x_j|h')}
		     {(1 + \fa(x_j|h))(1/\Ph)' (x_j)} \biggr]^2 \\ \times
	\biggl[ \prod_{j=1}^M \frac{\Ph_s (z_j + \i \g)}{\Ph_s (\la_j + \i \g)} \biggr]
	\biggl[ \prod_{\m \in {\cal Y}_n \ominus {\cal X}_n}
	        \Ph_s (\m - \i \g) \Ph_s (\m + \i \g) \biggr]
	\biggl[ \prod_{\m, \nu \in {\cal Y}_n \ominus {\cal X}_n}
	        \frac{1}{\sin(\m - \nu + \i \g)} \biggr] \\ \times
	\frac{\dh_2^2 (\Si)}{\dh_2^2} \biggl[ \prod_{j=1}^M \Ph' (\la_j) (1/\Ph)' (z_j) \biggr]
	\biggl[ \prod_{j=1}^M \frac{1 + \fa_n (\la_j|h')}{\fa' (\la_j|h)}
	                      \frac{1 + \fa (z_j|h)}{\fa_n' (z_j|h')} \biggr] \epp
\end{multline}

In order to see this one has to first realize that
\begin{multline}
     \frac{\prod_{j=1}^M \frac{\r_n (\la_j|h, h')}{\r_n (\m_j|h, h')}}
          {\det_M \Bigl\{\frac{1}{\sin(\la_j - \m_k)}\Bigr\}
	   \det_M \Bigl\{\frac{1}{\sin(\m_j - \la_k)}\Bigr\}} \\[1ex]
        = \prod_{j=1}^M
	  \frac{1 + \fa_n (\la_j|h')}{\fa' (\la_j|h)}
	  \frac{1 + \fa (\m_j|h)}{\fa_n' (z_j|h')}
	  \prod_{k=1}^M
	  \frac{\sin(\la_j - \m_k + \i \g) \sin(\m_j - \la_k + \i \g)}
	       {\sin(\la_j - \la_k + \i \g) \sin(\m_j - \m_k + \i \g)} \epc
\end{multline}
then use equations (\ref{start}), (\ref{firstfactorsplit}),
(\ref{secondsplit}), (\ref{defjtilde}), (\ref{detggstarlowt}),
(\ref{thesquare}), (\ref{firstsmalldetlow}), (\ref{scdsmalldetzero}),
and the definitions (\ref{defphis}), (\ref{defmatsu}).

\subsection{The Trotter limit}
Next we shall rewrite the remaining sums and products over all Bethe
roots in terms of integrals involving the auxiliary functions $\fa$
and $\fa_n$. In fact, we will be always dealing with differences
between excited state quantities and quantities pertaining to the
dominant state. These will be taken care of by means of the function
\begin{equation}
     \fz (\la) = \frac{1}{2\p\i}
                 \bigl( \Ln (1 + \fa_n) (\la|h') - \Ln (1 + \fa) (\la|h) \bigr) \epc
\end{equation}
where the logarithms were defined in (\ref{deflogoneplusa}).
\begin{lemma} \label{lem:sumroots}
{\bf Summation lemma.}\\ If $f$ is holomorphic on and inside
${\cal C}$, then
\begin{equation}
     \sum_{j=1}^M \bigl(f(z_j) - f(\la_j)\bigr) =
        - k \bigl(\tst{f(\frac{\p}{2} - \i x_0) - f(- \frac{\p}{2} - \i x_0)}\bigr)
	- \int_{\cal C} \rd \m \: f'(\m) \fz(\m) \epc
\end{equation}
where $x_0$ was introduced in Section~\ref{subsec:nlie} and
$k = \fe_n - \fe_0$.
\end{lemma}
\begin{proof}
Decompose the contour as in the calculation of the index in
(\ref{indexoneplusa}) and use partial integration.
\end{proof}
In order to have a compact notation in the following corollary
and below we introduce the `indicator function' $\one$ defined
by
\begin{equation}
     \one_{\rm condition} = \begin{cases}
                               1 & \text{if condition is satisfied} \\
			       0 & \text{else.}
                            \end{cases}
\end{equation}
\begin{corollary}
{\bf \boldmath Integral representations for $\Si$, $\Ph$ and $\Ph_s$.}
\begin{enumerate}
\item
The constant $\Si$ defined in (\ref{defsi}) has the integral
representation
\begin{equation} \label{sigmaint}
     \Si = k \p + \int_{\cal C} \rd \m \: \fz(\m) \epp
\end{equation}
\item
For the function $\Ph_s$ defined in (\ref{defphis}) we have for
all $\la \in {\mathbb C} \setminus {\cal C}$
\begin{equation} \label{phsint}
     \Ph_s (\la) =
        \biggl(\frac{1 + \fa (\la|h)}
	            {1 + \fa_n (\la|h')}\biggr)^{\one_{\la \in \Int {\cal C}}} (-1)^k
        \exp \biggl\{- \int_{\cal C} \rd \m \: \ctg (\la - \m) \fz(\m) \biggr\} \epc
\end{equation}
\item
while for $\Ph$ defined in (\ref{defphi}) it holds for any $\la$ with
$|\Im \la| < \g/2$ that
\begin{align} \label{phint}
     \Ph (\la) =
        \biggl(\frac{1 + \fa (\la|h)}
	            {1 + \fa_n (\la|h')}\biggr)^{\one_{\la \in \Int {\cal C}}} (-1)^k
        \exp \biggl\{- \int_{\cal C} \rd \m \:
	             \frac{\dh_1' (\la - \m)}{\dh_1 (\la - \m)} \fz(\m) \biggr\} \epp
\end{align}
\end{enumerate}
\end{corollary}
\begin{proof}
For (i) apply Lemma~\ref{lem:sumroots} to $f(\la) = - \la$. For the
proof of (ii) consider any holomorphic function $\ph$ that is $\p$
anti-periodic and has a single zero inside its periodicity strip
at $\m = 0$. Then apply Lemma~\ref{lem:sumroots} to $\ln \ph (\la - \m)$
considered as a functions of $\m$ with $\la \in \Ext {\cal C}$. This
gives (\ref{phsint}) for $\la \in \Ext {\cal C}$. The prefactor for
$\la \in \Int {\cal C}$ is obtained by means of analytic continuation.
The proof of (iii) is similar. The restriction $|\Im \la| < \g/2$
guarantees that the zeros of $\dh_1$ at $\pm m \i \g$, $m \in {\mathbb N}$
stay outside ${\cal C}$.
\end{proof}
\begin{corollary}
{\bf The double integrals.}
\begin{align} \label{prodphisbyphis}
     & \prod_{j=1}^M \frac{\Ph_s (z_j + \i \g)}{\Ph_s (\la_j + \i \g)} =
        \exp \biggl\{\int_{\cal C} \rd \la \int_{\cal C} \rd \m \:
	             \ctg' (\la - \m + \i \g) \fz(\la) \fz(\m) \biggr\} \epc \\
     & \prod_{j=1}^M \Ph' (\la_j) (1/\Ph)' (z_j)
        \frac{1 + \fa_n (\la_j|h')}{\fa' (\la_j|h)}
	\frac{1 + \fa (z_j|h)}{\fa_n' (z_j|h')} \notag \\ & \mspace{180.mu} =
        \exp \biggl\{- \int_{{\cal C}' \subset {\cal C}} \rd \la \int_{\cal C} \rd \m \:
	               [\ln \dh_1]'' (\la - \m) \fz(\la) \fz(\m) \biggr\} \epp
		       \label{largedets}
\end{align}
Here ${\cal C}'$ is a simple closed contour, tightly enclosed by
${\cal C}$ in such way that $\{\la\}, {\cal Z}_n \subset \Int {\cal C}'$.
\end{corollary}
\begin{proof}
Equation (\ref{prodphisbyphis}) follows from Lemma~\ref{lem:sumroots}
and (\ref{phsint}), since $\la_j + \i \g, z_j + \i \g \in \Ext {\cal C}$
for $j = 1, \dots, M$. For the proof of (\ref{largedets}) notice that
\begin{multline}
     \Ph' (\la_j) (1/\Ph)' (z_j)
        \frac{1 + \fa_n (\la_j|h')}{\fa' (\la_j|h)}
	\frac{1 + \fa (z_j|h)}{\fa_n' (z_j|h')} \\ =
     \lim_{\substack{\la \rightarrow \la_j\\z \rightarrow z_j}}
        \frac{(1 + \fa_n (\la|h'))\Ph(\la)}{1 + \fa (\la|h)}
	\frac{1 + \fa (z|h)}{(1 + \fa_n (z|h')) \Ph(z)} \\ =
     \exp \biggl\{\int_{\cal C} \rd \m \:
                  \biggl(\frac{\dh_1' (z_j - \m)}{\dh_1 (z_j - \m)} -
                         \frac{\dh_1' (\la_j - \m)}{\dh_1 (\la_j - \m)}\biggr)
			 \fz(\m) \biggr\} \epc
\end{multline}
where we have used (\ref{phint}) in the last equation. Using once
more (\ref{phint}) and the properties of the contour ${\cal C}'$
we arrive at (\ref{largedets}).
\end{proof}
\begin{corollary}
\begin{equation}
     \biggl[ \prod_{j=1}^\nex
             \frac{\fa_n' (x_j|h')}
	          {(1 + \fa(x_j|h))(1/\Ph)' (x_j)} \biggr]^2 =
     \exp \biggl\{- 2 \sum_{\la \in {\cal X}_n} \int_{\cal C} \rd \m \:
	             \frac{\dh_1' (\la - \m)}{\dh_1 (\la - \m)} \fz(\m) \biggr\} \epp
\end{equation}
\end{corollary}
\begin{proof}
Here the same reasoning as in the proof of the previous corollary
applies.
\end{proof}

Let us collect the single and double integrals in the exponents
and denote their sums
\begin{subequations}
\begin{align}
     \D_{1, {\cal C}} [\fz] & =
        - 2 \sum_{\la \in {\cal X}_n}
	    \int_{\cal C} \rd \m \:
	    \frac{\dh_1' (\la - \m)}{\dh_1 (\la - \m)} \fz(\m) \notag \\ & \mspace{153.mu}
	    - \sum_{\la \in {\cal Y}_n \ominus {\cal X}_n} \sum_{\s = \pm}
              \int_{\cal C} \rd \m \: \ctg (\la - \m + \s \i \g) \fz(\m) \epc \\[1ex]
     \D_{2, {\cal C}} [\fz] & =
        \int_{{\cal C}' \subset {\cal C}} \rd \la \int_{\cal C} \rd \m \:
	   \bigl(\ctg' (\la - \m + \i \g) - [\ln \dh_1]'' (\la - \m)\bigr)
	   \fz(\la) \fz(\m) \epp
\end{align}
\end{subequations}
\begin{lemma}
{\bf Re-interpretation of the remaining large determinants as Fredholm
determinants.}
\begin{equation}
	\det_M (I_M - U) = \det_{{\cal C}'} (\id - \widehat U) \epc \qd
	\det_M (I_M - U^*) = \det_{{\cal C}'} (\id - \widehat{U}^*) \epc
\end{equation}
where the expressions of the right hand side of these equations are
the Fredholm determinants of the integral operators $\widehat U$
and $\widehat{U}^*$ acting on functions $f$ holomorphic on $\Int {\cal C}$
and defined by
\begin{subequations}
\begin{align}
     \widehat{U} f (\x) & =
        \int_{{\cal C}'} \frac{\rd \z}{2 \p \i} \Ph(\z) f(\z) U(\z,\x) \epc \\[1ex]
     \widehat{U}^* f (\x) & =
        \int_{{\cal C}'} \frac{\rd \z}{2 \p \i} \Ph(\z) f(\z) U^* (\z,\x) \epc
\end{align}
where ${\cal C}'$ is again a simple closed contour tightly enclosed
by $\cal C$.
\end{subequations}
\end{lemma}
Using this lemma and the above notation and corollaries as well
as the definition (\ref{gaussnumber}) of the $q$-numbers in equation
(\ref{prelimsum}), we obtain our final result for the amplitudes
at finite Trotter number.
\begin{proposition} \label{prop:amps}
{\bf Factorized amplitudes at finite (and infinite) Trotter number.}
\begin{multline} \label{ampsfinalfiniten}
     A_n (h,h') = (-1)^\nex \re^{\D_{1, {\cal C}} [\fz] + \D_{2, {\cal C}} [\fz]}
        \frac{\det_\nex ({\cal D}) \det_\nex ({\cal D}^*)}{\det_{2 \nex} ({\cal J})}
	\det_{{\cal C}'} (\id - \widehat U) \det_{{\cal C}'} (\id - \widehat{U}^*)
	\\[1ex] \times
	\frac{\dh_2^2 (\Si)}{\dh_2^2}
        \biggl[\prod_{\m \in {\cal X}_n \cup {\cal Y}_n}
	       \biggl(1 - \frac{\fa (\m|h)}{\fa_n (\m|h')}\biggr)\biggr]
        \biggl[\prod_{\m, \nu \in {\cal Y}_n \ominus {\cal X}_n}
	       \frac{1}{\bigl[\2 - \frac{\i (\m - \n)}{2 \g}\bigr]_{q^4}} \biggr] \epp
\end{multline}
\end{proposition}

Equation (\ref{ampsfinalfiniten}) still is a finite Trotter number
representation of the amplitudes. Its derivation is only based on
certain assumptions on the location of the Bethe roots $\{\la\}$ of
the dominant state and $\{\m\}$ of the excited states relative
to the reference contour ${\cal C}$, that are described above.
Conjecture~\ref{con:nostrings}, cited from our previous work \cite{DGKS15b},
implies that these assumptions are satisfied for all excited states
if $T$ is low enough and $N$ is large enough and that they continue to
hold in the Trotter limit if $T$ is low enough. This means that in the
Trotter limit the auxiliary functions $\fa$ and $\fa_n$, which fully
determine (\ref{ampsfinalfiniten}), are the solutions of (\ref{nlie1})
with $\e^{(N,\e)}$ replaced by its limit $\e_0$ (see (\ref{defepszero})).
Hence, the Trotter limit of (\ref{ampsfinalfiniten}) is obtained by
replacing $\fa$ and $\fa_n$ by their limit functions.

\subsection{\boldmath The low-$T$ limit}
The calculation of the low-$T$ limit of the form factor amplitudes
is based on Conjecture~\ref{con:nostrings} and on the following corollary.
\begin{corollary} \label{cor:fzzero}
The function $\fz$ is $\p$-periodic and on ${\cal C}_+ \cup {\cal C}_-$
has the low-$T$ asymptotics
\begin{equation} \label{zzzero}
     \fz (\la) = \fz_0 (\la) \one_{\la \in {\cal C}_+}
                 + \one_{\la \in {\cal C}_+ \cup {\cal C}_-}
		   \CO \bigl(T^\infty\bigr) \epc
\end{equation}
where
\begin{equation}
     \fz_0 (\la) = \frac{k}{2} + \frac{h - h'}{4 \p \i T} + F_n (\la)
\end{equation}
with $F_n$ as defined above (\ref{anlowt}).
\end{corollary}
In fact, by means of (\ref{zzzero}), most of the integrals that remain
in the low-$T$ limit reduce to convolution-type integrals involving
$\fz_0$. In order to deal with those integrals we introduce the notation
\begin{equation} \label{defsizero}
     \Si_0 = \Si_0 ({\cal X}_n, {\cal Y}_n) = - \frac{\p k}{2} - \frac{h - h'}{4 \i T}
           + \2 % \mspace{-18.mu}
	        \sum_{\m \in {\cal Y}_n \ominus {\cal X}_n} \mspace{-10.mu} \m \epc \qd
     \ups (\la) = \frac{(\re^{2 \i \la} q^4;q^4)}
                         {(\re^{2 \i \la} q^2;q^4)} \epp
\end{equation}
Then
\begin{equation} \label{zzeroups}
     2 \p \i \fz_0 (\la) = - 2 \i \Si_0
       + \sum_{\m \in {\cal Y}_n \ominus {\cal X}_n}
         \bigl\{ \ln\bigl(\ups(\la - \m)\bigr)
	         - \ln\bigl(\ups(\m - \la)\bigr) \bigr\} \epp
\end{equation}
\begin{lemma} \label{lem:baseint}
{\bf Basic integration lemma.}\\ Let $\s = \sign \bigl(\Im (\la)\bigr)$.
For any $\la \notin {\cal C}$
\begin{equation} \label{ctgint}
     \int_{\cal C} \rd \m \: \ctg(\la - \m) \fz (\m) =
        - \s \i \Si_0
        + \sum_{\m \in {\cal Y}_n \ominus {\cal X}_n}
          \ln\bigl(\ups(\s(\la - \m))\bigr) + \CO \bigl(T^\infty\bigr) \epp
\end{equation}
\end{lemma}
\begin{proof}
Equation (\ref{zzzero}) and (\ref{zzeroups}) imply that
\begin{multline}
     \int_{\cal C} \rd \m \: \ctg(\la - \m) \fz (\m) =
     - 2 \i \Si_0 \int_{- \frac{\p}{2}}^\frac{\p}{2} \frac{\rd \m}{2 \p \i} \ctg(\m - \la) \\
       + \sum_{\n \in {\cal Y}_n \ominus {\cal X}_n}
         \int_{- \frac{\p}{2}}^\frac{\p}{2} \frac{\rd \m}{2 \p \i}
         \bigl\{ \ln\bigl(\ups(\m - \n)\bigr) - \ln\bigl(\ups(\n - \m)\bigr) \bigr\}
	 \ctg(\m - \la) + \CO \bigl(T^\infty\bigr) \epp
\end{multline}
The remaining integrals can be evaluated by means of the residue
theorem. For this purpose one may use that $\ctg$ is a $\p$-periodic
function with its only simple poles at $\la = 0 \mod \p$ and
residue $1$ at $\la = 0$, which behaves asymptotically as
$\lim_{\Im \la \rightarrow \pm \infty} \ctg (\la) = \mp \i$. With
this information the first integral on the right hand side may
be calculated by shifting the integration contour, for instance,
to $[- \p/2,\p/2] + \i \infty$.

For the evaluation of the second integral note that $\ups$ is
a meromorphic, $\p$-periodic function on ${\mathbb C}$ which is
free of poles and zeros in ${\mathbb H}_+ - \i \g$ and decays
exponentially fast for $\Im \la \rightarrow + \infty$. Shifting
the integration contour to $[- \p/2,\p/2] + \i \infty$ for the
first term in the second integral on the right hand side and
to  $[- \p/2,\p/2] - \i \infty$ for the second term in the
second integral on the right hand side we obtain a single
residue contribution from the pole of the cotangent function,
which is either in ${\mathbb H}_+$ or ${\mathbb H}_-$, and we
have arrived at (\ref{ctgint}).
\end{proof}

For $\s = \pm$ define
\begin{equation} \label{defphpm}
     \Ph^{(\s)} (\la) =
        \re^{\s \i \Si_0}
	\prod_{\m \in {\cal X}_n \ominus {\cal Y}_n}
	\tst{\G_{q^4} \Bigl(\2 - \frac{\s \i (\la - \m)}{2 \g}\Bigr)
	     \G_{q^4} \Bigl(1 + \frac{\s \i (\la - \m)}{2 \g}\Bigr)} \epp
\end{equation}

\begin{corollary} \label{cor:ethpectg}
For any $\la \notin {\cal C}$ and $\s = \sign \bigl(\Im (\la)\bigr)$
\begin{multline} \label{intzctg}
     \exp \biggl\{- \int_{\cal C} \rd \m \: \ctg (\la - \m) \fz(\m) \biggr\} = \\
        \re^{\s \i \Si_0}
	\prod_{\m \in {\cal X}_n \ominus {\cal Y}_n}
	\frac{\G_{q^4} \Bigl(\2 - \frac{\s \i (\la - \m)}{2 \g}\Bigr)}
	     {\G_{q^4} \Bigl(1 - \frac{\s \i (\la - \m)}{2 \g}\Bigr)}
	     \bigl(1 + \CO (T^\infty)\bigr) \epp
\end{multline}
If, in addition, $|\Im \la| < \g/2$, then
\begin{equation} \label{intzlogpth}
     \exp \biggl\{- \int_{\cal C} \rd \m \:
                    \frac{\dh_1' (\la - \m)}{\dh_1 (\la - \m)} \fz(\m) \biggr\} =
		    \Ph^{(\s)} (\la) \bigl(1 + \CO (T^\infty)\bigr) \epp
\end{equation}
\end{corollary}
\begin{proof}
Equation (\ref{intzctg}) is a direct consequence of (\ref{ctgint}) and 
of the definition (\ref{defgammaqgq}) of the $q$-gamma function.

For the proof of (\ref{intzlogpth}) we note that
\begin{equation} \label{dhctgdecomp}
     \frac{\dh_1' (\la)}{\dh_1 (\la)} =
        \ctg(\la) + \sum_{k=1}^\infty \bigl(\ctg(\la - k \i \g) + \ctg(\la + k \i \g)\bigr)
\end{equation}
which follows from (\ref{proddh4}), (\ref{defotherdhs}). Assuming that
$|\Im \la| < \g/2$, multiplying by $\fz (\la)$, integrating over $\cal C$
and using (\ref{ctgint}), we obtain
\begin{multline} \label{dhintups}
     \int_{\cal C} \rd \m \: \frac{\dh_1' (\la - \m)}{\dh_1 (\la - \m)} \fz(\m) =
        - \s \i \Si_0 + \sum_{\m \in {\cal Y}_n \ominus {\cal X}_n}
	                   \ln \bigl(\ups(\s(\la - \m))\bigr) \\
        + \sum_{\m \in {\cal Y}_n \ominus {\cal X}_n} \sum_{k=1}^\infty
	  \Bigl\{ \ln \bigl(\ups(\la - \m + k \i \g)\bigr)
	          + \ln \bigl(\ups( \m - \la + k \i \g)\bigr) \Bigr\} \epp
\end{multline}
By straightforward direct calculation
\begin{equation} \label{produps}
     \prod_{k=1}^\infty \ups (\la - \m + k \i \g) =
        \frac{1}{(\re^{2\i (\la - \m)} q^4; q^4)} \epp
\end{equation}
Finally, (\ref{dhintups}), (\ref{produps}) and (\ref{ctgint})
imply (\ref{intzlogpth}).
\end{proof}

\begin{corollary}
\label{cor:siphilowt}
If $\la \in {\cal C}_+ \cup {\cal C}_-$, then
\begin{equation} \label{phicpcm}
     \Ph (\la) = (-1)^k \bigl\{\Ph^{(+)} (\la) \one_{\la \in {\cal C}_+}
                               + \Ph^{(-)} (\la) \one_{\la \in {\cal C}_-}\bigr\}
			       \bigl(1 + \CO(T^\infty)\bigr) \epp
\end{equation}
Moreover,
\begin{equation} \label{silim}
     \Si = k \p + \Si_0 + \CO(T^\infty) \epp
\end{equation}
\end{corollary}
\begin{proof}
Equation (\ref{phicpcm}) follows directly from (\ref{phint}) and
(\ref{intzlogpth}), taking into account that $\Ph$ is analytic
across ${\cal C}_+$ and develops a jump singularity across ${\cal B}_-$
in the limit $T \rightarrow 0+$.

Equation (\ref{silim}) is obtained from (\ref{ctgint}) by taking the
limit $\Im \la \rightarrow + \infty$ and taking into account that
in this limit $\ctg (\m - \la) \rightarrow \i$ and
$\ln \bigl(\ups (\la - \m)\bigr) \rightarrow 0$.
\end{proof}

\begin{corollary}
\label{cor:singleint}
\begin{equation}
     \re^{\D_{1, {\cal C}} [\fz]} =
        \biggl[ \prod_{\la \in {\cal X}_n} \Ph^{(-)} (\la) \biggr]^2
        \Biggl[ \prod_{\la, \m \in {\cal X}_n \ominus {\cal Y}_n}
	\frac{\G_{q^4} \Bigl(\frac{3}{2} - \frac{\i (\la - \m)}{2 \g}\Bigr)}
	     {\G_{q^4} \Bigl(1 - \frac{\i (\la - \m)}{2 \g}\Bigr)}
	        \Biggr]^2
	     \bigl(1 + \CO (T^\infty)\bigr) \epp
\end{equation}
\end{corollary}
\begin{proof}
This follows easily from Corollary~\ref{cor:ethpectg}.
\end{proof}

\begin{lemma} \label{lem:doubleint}
\begin{equation}
     \re^{\D_{2, {\cal C}} [\fz]} =
        \Biggl[ \prod_{\la, \m \in {\cal X}_n \ominus {\cal Y}_n}
	\frac{G_{q^4} \Bigl(1 - \frac{\i (\la - \m)}{2 \g}\Bigr)
	      G_{q^4} \Bigl(2 - \frac{\i (\la - \m)}{2 \g}\Bigr)}
	     {G_{q^4} \Bigl(\frac{3}{2} - \frac{\i (\la - \m)}{2 \g}\Bigr)
	      G_{q^4} \Bigl(\frac{5}{2} - \frac{\i (\la - \m)}{2 \g}\Bigr)}
	        \Biggr]
	     \bigl(1 + \CO (T^\infty)\bigr) \epp
\end{equation}
\end{lemma}
\begin{proof}
We first of all evaluate the integral
\begin{equation}
     I_2 (\tau) =
        \int_{{\cal C}' \subset {\cal C}} \rd \la \int_{\cal C} \rd \m \:
        \ctg' (\la - \m + \i \tau) \fz (\la) \fz (\m) \epc
\end{equation}
where $\tau \in {\mathbb C}$ is such that $\la - \m + \i \tau \ne 0$
for $\m \in {\cal C}$, $\la \in {\cal C}'$. The regularization ${\cal C}'
\subset {\cal C}$ may be lifted if $\Im \tau \ne 0$ which we shall
assume in the following. Since $\ctg' = - 1/\sin^2$ is an even function,
$I_2 (\tau)$ is even in $\tau$, whence $I_2 (\tau) = I_2 (\s \tau)$,
where $\s = \sign(\Im \tau)$. Recall that $\fz$ is $\p$-periodic.
Hence, a partial integration gives
\begin{equation}
     I_2 (\tau) =
        - \int_{{\cal C}' \subset {\cal C}} \rd \la \fz' (\la) 
	  \int_{\cal C} \rd \m \: \ctg (\la - \m + \i \s \tau) \fz (\m) \epp
\end{equation}
Now we use Lemma~\ref{lem:baseint} for the $\m$-integral and
Corollary~\ref{cor:fzzero} and equation (\ref{zzeroups}) for
$\fz (\la)$. Then
\begin{multline} \label{itwores}
     I_2 (\tau) = \mspace{-18.mu}
        \sum_{\m, \n \in {\cal X}_n \ominus {\cal Y}_n}
	\int_{- \frac{\p}{2}}^\frac{\p}{2} \frac{\rd \la}{2 \p \i}
	\ln \bigl( \ups(\la - \m + \i \s \tau) \bigr)
	\6_\la \bigl\{\ln \bigl( \ups(\la - \n) \bigr)
	                - \ln \bigl( \ups(\n - \la) \bigr)\bigr\} \\[-2ex]
			\mspace{450.mu} + \CO (T^\infty) \\ = \mspace{-18.mu}
        \sum_{\m, \n \in {\cal X}_n \ominus {\cal Y}_n} \sum_{n=1}^\infty
	\bigl\{
	\ln \bigl( \ups(\n - \m + \i \s \tau + (2n - 1) \i \g) \bigr) \\
	- \ln \bigl( \ups(\n - \m + \i \s \tau + 2n \i \g) \bigr)\bigr\}
	+ \CO (T^\infty) \epp
\end{multline}
In the second equation we have shifted the integration contour
to $[-\p/2,\p/2] + \i \infty$ and have used the residue theorem
and the analytic properties of $\ln \ups$ described in the proof of
Lemma~\ref{lem:baseint}.

Equation (\ref{itwores}) implies that
\begin{equation} \label{reizwo}
     \re^{I_2 (\tau)} =
        \biggl[ \prod_{\la, \m \in {\cal X}_n \ominus {\cal Y}_n}
	        A(\la - \m + \i \s \tau) \biggr] \bigl(1 + \CO(T^\infty)\bigr) \epc
\end{equation}
where
\begin{equation}
     A (\x) = (\re^{2 \i \x} q^4; q^4)
              \frac{(\re^{2 \i \x} q^6; q^4, q^4)^2}
	           {(\re^{2 \i \x} q^4; q^4, q^4)^2} \epp
\end{equation}

Note that
\begin{equation} \label{dzwopartone}
     A (\x + \i \g) =
        (\re^{2 \i \x} q^6; q^4)
	\frac{(\re^{2 \i \x} q^8; q^4, q^4)^2}{(\re^{2 \i \x} q^6; q^4, q^4)^2} =
	\frac{(\re^{2 \i \x} q^6; q^4) (\re^{2 \i \x} q^4; q^4, q^4)^2}
	     {(\re^{2 \i \x} q^4; q^4)^2 (\re^{2 \i \x} q^6; q^4, q^4)^2}
\end{equation}
and
\begin{equation} \label{dzwopartzwo}
    A (\x) \prod_{k=1}^\infty A(\x + \i k \g) A(\x - \i k \g) =
       \frac{1}{(\re^{2 \i \x} q^4; q^4)}
\end{equation}
by a straightforward calculation. Using (\ref{dhctgdecomp}) and
equations (\ref{reizwo}), (\ref{dzwopartone}) and (\ref{dzwopartzwo})
it follows that
\begin{multline}
     \re^{\D_{2, {\cal C}} [\fz]} = \re^{I_2 (\g)}
        \re^{- I_2 (0) - \sum_{k=1}^\infty (I_2 (k \g) + I_2 (- k \g))} \\[1ex]
        = \biggl[ \prod_{\la, \m \in {\cal X}_n \ominus {\cal Y}_n}
	\frac{(\re^{2 \i (\la - \m)} q^6; q^4) (\re^{2 \i (\la - \m)} q^4; q^4, q^4)^2}
	     {(\re^{2 \i (\la - \m)} q^4; q^4) (\re^{2 \i (\la - \m)} q^6; q^4, q^4)^2}
	        \biggr] \bigl(1 + \CO(T^\infty)\bigr) \epc
\end{multline}
Replacing the $q$-factorials by $q$-gamma and $q$-Barnes functions
by means of (\ref{defgammaqgq}) and using the functional equation
(\ref{gammabarnesfun}) we establish the claim.
\end{proof}

It remains to consider the low-$T$ limit of the determinants. All our
efforts so far were driven by the desire to be left with Fredholm
determinants that trivialize in the low-$T$ limit.

\begin{lemma}
{\bf \boldmath Low-$T$ limit of the Fredholm determinants.}
\begin{equation}
   \det_{{\cal C}'} (\id - \widehat U) = 1 + \CO (T^\infty) \epc \qd
   \det_{{\cal C}'} (\id - \widehat{U}^*) = 1 + \CO (T^\infty) \epp
\end{equation}
\end{lemma}
\begin{proof}
This holds, since the kernels of the operators $\widehat U$ and
$\widehat{U}^*$ vanish exponentially fast in the low-$T$ limit
as can be seen from the definitions in (\ref{defukernels}).
\end{proof}

In order to express the low-$T$ limit of the finite determinants
in a compact way we introduce two functions $\Om_n (\la,\m) =
\Om (\la, \m|{\cal X}_n, {\cal Y}_n)$ and $\overline{\Om}_n (\la,\m) =
\overline{\Om} (\la,\m|{\cal X}_n, {\cal Y}_n)$ setting
\begin{subequations}
\label{defombarom}
\begin{align}
     \Om (\la, \m|{\cal X}_n, {\cal Y}_n) &
        = g_{\Si_0} (\la, \m)
	  - \int_{- \p/2}^{\p/2} \frac{\rd \z}{2 \p \i}
	                         \frac{\Ph^{(-)} (\m)}{\Ph^{(+)} (\z)} \,
				 g_{\Si_0} (\la, \z) \re(\z - \m) \epc \\[1ex]
     \overline{\Om} (\la,\m|{\cal X}_n, {\cal Y}_n) &
        = g_{- \Si_0} (\la, \m)
	  - \int_{- \p/2}^{\p/2} \frac{\rd \z}{2 \p \i} \:
	                         \frac{\Ph^{(-)} (\z)}{\Ph^{(+)} (\m)} \,
				 g_{- \Si_0} (\la, \z) \re(\m - \z) \epp
\end{align}
\end{subequations}
Then we can formulate the following lemma.
\begin{lemma} \label{lem:smalldetslowt}
{\bf \boldmath Low-$T$ limit of the matrices in the small determinants.}
\begin{subequations}
\begin{align} \label{dlow}
     {\cal D}^l_m & =
        \frac{(-1)^k \, \Om_n (x_l, y_m)}
	     {\Ph^{(+)} (y_m) \bigl(1 - \re^{2 \p \i \fz_0 (y_m)}\bigr)}
	     \bigm(1 + \CO (T^\infty)\bigr) \epc
	     \\[1ex] \label{dstarlow}
     {\cal D^*}^l_m & =
        \frac{(-1)^{k+1} \, \overline{\Om}_n (y_l,x_m)}
	     {\Ph^{(+)} (y_l) \bigl(1 - \re^{2 \p \i \fz_0 (x_m)}\bigr)}
	     \bigm(1 + \CO (T^\infty)\bigr) \epp
\end{align}
\end{subequations}
\end{lemma}
\begin{proof}
$S(x_j, \z) = \CO (T^\infty)$ by (\ref{defukernels}) and (\ref{defskernels}).
Inserting this together with (\ref{glimyparticle}), (\ref{phicpcm}) and
(\ref{silim}) into (\ref{firstsmalldetlow}) we obtain the following
low-$T$ limit for the elements of the `first small determinant',
\begin{multline} \label{smalldettzero}
     {\cal D}^l_m = \int_{\cal C} \frac{\rd \z}{2 \p \i \Ph(\z)}
        g_\Si (x_l,\z) G(\z,y_m) + \CO (T^\infty) = \\ (-1)^k
	\sum_{\s = \pm} \int_{{\cal C}_\s} \frac{\rd \z}{2 \p \i \Ph^{(\s)} (\z)} \,
	   g_{\Si_0} (x_l,\z)
	   \biggl[g_0 (\z, y_m)
	          + \frac{2 \p \i R_0 (\z - y_m)}
		         {1 - \re^{- 2 \p \i \fz_0 (y_m)}} \biggr]
           + \CO (T^\infty) \epp
\end{multline}

The integral can be simplified exploiting the properties of the
functions in the integrand under shifts by $\i \g$,
\begin{subequations}
\label{shiftsbyigamma}
\begin{align}
     g_\a (\la,\m) & = - \re^{2 \i \a} g_\a (\la + \i \g, \m)
                  = - \re^{- 2 \i \a} g_\a (\la, \m + \i \g) \epc \\[1ex]
     2 \p \i R_0 (\la) & = - 2 \p \i R_0 (\la + \i \g) - \re(\la + \i \g) \epc \\[1ex]
     \Ph^{(-)} (\la) & = \re^{- 2 \i \Si_0} \Ph^{(+)} (\la + \i \g) \epp
\end{align}
\end{subequations}
Inserting these relations into the integral over ${\cal C}_-$ we obtain
\begin{multline}
     \int_{{\cal C}_-} \frac{\rd \z}{2 \p \i \Ph^{(-)} (\z)} \,
	   g_{\Si_0} (x_l,\z)
	   \biggl[g_0 (\z, y_m)
	          + \frac{2 \p \i R_0 (\z - y_m)}
		         {1 - \re^{- 2 \p \i \fz_0 (y_m)}} \biggr] \\[1ex] =
     \int_{{\cal C}_- + \i \g} \frac{\rd \z}{2 \p \i \Ph^{(+)} (\z)} \,
	   g_{\Si_0} (x_l,\z)
	   \biggl[g_0 (\z, y_m)
	          + \frac{2 \p \i R_0 (\z - y_m) + \re (\z - y_m)}
		         {1 - \re^{- 2 \p \i \fz_0 (y_m)}} \biggr] \\[1ex] =
     \frac{g_{\Si_0} (x_l, y_m)}{\Ph^{(+)} (y_m)}
        \biggl[1 - \frac{1}{1 - \re^{- 2 \p \i \fz_0 (y_m)}}\biggr] \mspace{288.mu} \\
     - \int_{{\cal C}_+} \frac{\rd \z}{2 \p \i \Ph^{(+)} (\z)} \,
	   g_{\Si_0} (x_l,\z)
	   \biggl[g_0 (\z, y_m)
	          + \frac{2 \p \i R_0 (\z - y_m) + \re (\z - y_m)}
		         {1 - \re^{- 2 \p \i \fz_0 (y_m)}} \biggr] \epp
\end{multline}
Substituting this back into (\ref{smalldettzero}) and using that
\begin{equation} \label{splitetoz}
     \re^{2 \p \i \fz_0 (x)} = \frac{\Ph^{(-)} (x)}{\Ph^{(+)} (x)}
\end{equation}
we arrive at (\ref{dlow}).

As for the entries of the `second small determinant' we first of
all evaluate the low-$T$ limit of the function $H$ introduced
in (\ref{defh}). We start with
\begin{multline} \label{phpanp}
     \frac{(1/\Ph)' (x_m)}{\fa_n' (x_m|h')} =
	\lim_{\la \rightarrow x_m} \frac{1}{\Ph(\la) (1 + \fa_n (\la|h')} =
	\frac{(-1)^k \exp \Bigl\{
	         \int_{\cal C} \rd \m \:
		 \frac{\dh_1' (x_m - \m)}{\dh_1 (x_m - \m)} \fz (\m)\Bigr\}}
             {1 - \frac{\fa (x_m|\k)}{\fa_n (x_m|\k')}} \\[1ex] =
	\frac{(-1)^k (1 + \CO (T^\infty))}
             {\Ph^{(-)} (x_m) (1 - \re^{- 2 \p \i \fz_0 (x_m)})} =
        \frac{(-1)^k (1 + \CO (T^\infty))}{\Ph^{(-)} (x_m) - \Ph^{(+)} (x_m)} \epp
\end{multline}
Here we have employed (\ref{phint}) and the Bethe Ansatz equations in
the second equation, (\ref{intzlogpth}) in the third equation and
(\ref{splitetoz}) in the last equation. Using the equation (\ref{phpanp})
as well as (\ref{rstarzero}), (\ref{phicpcm}), and (\ref{silim}) we obtain
\begin{equation}
     H (\la, x_m) = H_0 (\la, x_m) + \CO(T^\infty) \epc
\end{equation}
where
\begin{multline} \label{defhzero}
     H_0 (\la,x_m) = g_{- \Si_0} (\la,x_m) \\
        - \biggl[\int_{{\cal C}_+'} \rd \z \: \Ph^{(+)} (\z) +
                 \int_{{\cal C}_-'} \rd \z \: \Ph^{(-)} (\z)\biggl]
        \frac{g_{- \Si_0} (\la,\z) R_0(\z - x_m)}{\Ph^{(-)} (x_m) - \Ph^{(+)} (x_m)}
	\epp
\end{multline}

Since $S^* (\x, \z) = \CO (T^\infty)$ by (\ref{defukernels}) and
(\ref{defskernels}), this allows us to conclude that
\begin{multline} \label{d2zero}
     {\cal D^*}^l_m = (-1)^k \biggl\{
        \frac{2 \p \i R_0 (y_l - x_m)}{\Ph^{(-)} (x_m) - \Ph^{(+)} (x_m)}
        + \int_{{\cal C}_+} \frac{\rd \x}{2 \p \i \Ph^{(+)} (\x)} \,
          \re(\x - y_l) H_0 (\x,x_m) \\[1ex]
        - \int_{{\cal C}_-} \frac{\rd \x}{2 \p \i \Ph^{(-)} (\x)} \,
	  \re(y_l - \x) H_0 (\x, x_m) \biggr\} + \CO(T^\infty) \\[1ex] =
        (-1)^k \biggl\{
        \frac{2 \p \i R_0 (y_l - x_m)}{\Ph^{(-)} (x_m) - \Ph^{(+)} (x_m)}
	- \frac{H_0 (y_l, x_m)}{\Ph^{(+)} (y_l)} \biggr\} + \CO(T^\infty) \epp
\end{multline}
Next we use (\ref{shiftsbyigamma}) and $\re(- \la) = \re (\la + \i \g)$
and apply similar manipulations as in the above proof of (\ref{dlow})
to equation (\ref{defhzero}). After some calculation we arrive at
(\ref{dstarlow}).
\end{proof}

\subsection{\boldmath The amplitudes in the low-$T$ limit}
If we now use Corollaries~\ref{cor:siphilowt}, \ref{cor:singleint}
and Lemmata \ref{lem:doubleint}-\ref{lem:smalldetslowt} in
Proposition~\ref{prop:amps} together with the basic functional equations
(\ref{gammabarnesfun}), some cancellations occur, and we obtain our
final result for the amplitudes in the low-$T$ limit. It assumes
a more compact shape if we introduce a function
\begin{equation} \label{defpsi}
     \Ps (\la) = \G_{q^4} \bigl(\tst{\2} - \tst{\frac{\i \la}{2\g}}\bigr)
               \G_{q^4} \bigl(1 - \tst{\frac{\i \la}{2\g}}\bigr)
               \frac{G^2_{q^4} \bigl(1 - \tst{\frac{\i \la}{2\g}}\bigr)}
                    {G^2_{q^4} \bigl(\2 - \tst{\frac{\i \la}{2\g}}\bigr)} \epp
\end{equation}
We summarize it in the following theorem.

\begin{theorem}
{\bf \boldmath Amplitudes at low $T$.}\\
The amplitudes in the form factor expansion of the generating
function of the longitudinal correlations function have the
low-$T$ asymptotic behaviour
\begin{multline} \label{fine}
     A_n (h, h') = 
        \frac{\dh_2^2 (\Si_0)}{\dh_2^2}
	\biggl[
	\prod_{\la, \m \in {\cal X}_n \ominus {\cal Y}_n}
	\mspace{-18.mu} \Ps (\la - \m) \biggr]
       \frac{\det_{\nex} \bigl\{\Om_n (x_j,y_k)\bigr\}
             \det_{\nex} \bigl\{\overline{\Om}_n (y_j,x_k)\bigr\}}
            {\det_{2 \nex} \{{\cal J}\}} \\[-1ex] \times
	\bigl(1 + \CO (T^\infty)\bigr) \epp
\end{multline}
\end{theorem}
This is the main result of this work. It has many implications.
We would like to emphasize that (\ref{fine}) is explicit and is
defined entirely in terms of functions of the q-gamma family.
In particular, as compared to previous expressions for form-%
factor amplitudes, there are neither multiple integrals nor
Fredholm determinants involved. We like to think of (\ref{fine})
as of a `factorized form' of the amplitudes. The function $\Om$
remotely resembles the function $\om$ (cf.\ equation (50)
of \cite{BoGo09}, where the function is denoted $\Ps$ rather
than $\om$) appearing in the theory of factorized static
correlation functions \cite{JMS08}. The functions $\Om$,
$\overline \Om$ are still defined as single integrals over
explicit functions. In fact, the integrals can be evaluated
in terms of $q$-hypergeometric functions. We further comment
on this and give explicit examples below. The general case will
be considered in a separate subsequent work. We finally note
that taking the isotropic limit toward the Heisenberg chain
is straightforward. This limit will be considered in an extra
section below.

\section{The form factor series}
\label{sec:ffseries}
\subsection{The series for the generating function}
\label{sec:genfunseries}
The amplitudes (\ref{fine}) together with the eigenvalue ratios
at finite Trotter number are the input for the form factor
series (\ref{defgenfun}) of the generating function. A key
feature of the amplitudes is the appearance of the factor
$\det_{2 \nex} \{{\cal J}\}$ in the denominator. It is this
factor that allows us to convert the series over all excitations
into a series over multiple integrals, each accounting for a
whole class of particle-hole excitations. We have described this
in some detail in \cite{DGKS16b,GKKKS17}. The main idea is that
\begin{equation}
     {\cal J} ({\cal U}, {\cal V}, k) = 
	\begin{pmatrix}
	   \6_{u_k} \fa^- (u_j |{\cal U},{\cal V},h') &
	   \6_{v_k} \fa^- (u_j |{\cal U},{\cal V},h') \\
	   \6_{u_k} \fa^+ (v_j |{\cal U},{\cal V},h') &
	   \6_{v_k} \fa^+ (v_j |{\cal U},{\cal V},h')
	\end{pmatrix}
\end{equation}
is the Jacobi matrix $\6(\Uv,\Vv)/\6(\uv,\vv)$ of a transformation
${\mathbb C}^{2 \nex} \mapsto {\mathbb C}^{2 \nex}$, $(\uv, \vv)
\mapsto (\Uv, \Vv)$, where
\begin{subequations}
\label{aasamap}
\begin{align}
     U_j (\uv, \vv) & = 1 + \fa^- (u_j|{\cal U},{\cal V},h') \epc
                        \qd j = 1, \dots, \nex \epc \\[1ex]
     V_k (\uv, \vv) & = 1 + \fa^+ (v_k|{\cal U},{\cal V},h') \epc
                        \qd k = 1, \dots, \nex \epp
\end{align}
\end{subequations}
This transformation maps solutions of the higher level Bethe
Ansatz equations to the origin in ${\mathbb C}^{2 \nex}$.
For this reason the terms in the thermal form factor series
of the generating functions may be interpreted as multiple
residues. Under certain assumptions which are discussed in
detail in Appendix~D of \cite{GKKKS17} the sum over all
excitations then becomes a sum over multiple integrals. For
the derivation of the final series it is important to start
with multiple integral representations of the summands at
finite Trotter number and perform the Trotter limit on the
integrals. This is feasible due to the fact, that the Trotter
number enters our expressions for the amplitudes and eigenvalue
ratios only parametrically through the functions $\fa^\pm$.

For the following consideration let us make the Trotter number
dependence of the various functions explicit by providing an
extra index $N$. We shall denote, in particular,
\begin{equation}
     {\cal A}_N ({\cal U}, {\cal V}|k) =
        A_n (h, h') \det_{2 \nex} \{{\cal J}\}
	\Bigr|_{\substack{{\cal X}_n \rightarrow {\cal U}\\
	                  {\cal Y}_n \rightarrow {\cal V}}}
\end{equation}
the `off-shell' finite Trotter number form factor density.
In a similar way $\r_N (\la|{\cal U}, {\cal V}|k)$ are the
off-shell eigenvalue ratios at finite Trotter number. For
the limits $N \rightarrow \infty$, $\e \rightarrow 0$ we
introduce the notation
\begin{equation}
     \lim_{N \rightarrow \infty} \lim_{\e \rightarrow 0}
        {\cal A}_N ({\cal U}, {\cal V}|k) = {\cal A} ({\cal U}, {\cal V}|k) \epp
\end{equation}

Recall that ${\cal B}_\pm$ are the curves in ${\mathbb H}_\pm$
on which the particle and hole roots condense in the low-$T$
limit, i.e.\ for $T$ very small and $N$ large enough. Let
us assume that these curves are oriented in the direction of
growing real part. We define two simple positively oriented
curves $\G({\cal B}_\pm)$ which go around ${\cal B}_\pm$ and
enclose all particles and holes, respectively, if $T$ is
small enough and $N$ is large enough. Then
\begin{multline} \label{intrepgfiniten}
     {\cal G} (m, t, T, h, h') =
        \lim_{N \rightarrow \infty} \lim_{\e \rightarrow 0} \biggl\{
	{\cal A}_N (\emptyset, \emptyset|0) + {\cal A}_N (\emptyset, \emptyset|1) (-1)^m \\
        + \sum_{\substack{\nex \in {\mathbb N}\\k = 0, 1}}
	  \frac{(-1)^{km}}{(\nex !)^2}
	     \int_{\G({\cal B}_+)^\nex} \frac{\rd^\nex u}{(2\p \i)^\nex}
	     \int_{\G({\cal B}_-)^\nex} \frac{\rd^\nex v}{(2\p \i)^\nex} \:
             \frac{\r_N \bigl(- \frac{\i \g}{2} + \frac{\i t}{\k N} \big|
	           {\cal U}, {\cal V}\big|k\bigr)^{N/2}}
	          {\r_N \bigl(- \frac{\i \g}{2} - \frac{\i t}{\k N}\big|
		   {\cal U}, {\cal V}\big|k\bigr)^{N/2}} \\
          \frac{{\cal A}_N ({\cal U}, {\cal V}|k) \,
	        \r_N \bigl(- \frac{\i \g}{2}\big|{\cal U}, {\cal V}\big|k\bigr)^m}
               {\prod_{j=1}^\nex \bigl(1 + \fa^- (u_j|{\cal U}, {\cal V}, h')\bigr)
                                 \bigl(1 + \fa^+ (v_j|{\cal U}, {\cal V}, h')\bigr)}
				 \biggr\} \epp
\end{multline}
The reason why we had to come back to finite Trotter number here
is that we needed to control the singularities of the function
$\r_N \bigl(- \frac{\i \g}{2} - \frac{\i t}{\k N}\big| {\cal U},
{\cal V}\big|k\bigr)^{-N/2}$. As a function of $u_j$ it has an
$N/2$-fold pole at $- \frac{\i \g}{2} - \frac{\i t}{\k N}$ which
is compensated by a similar pole in $\fa^- (u_j|{\cal U},
{\cal V}, h')$. As a function of $v_k$ is has an $N/2$-fold pole
at $\frac{\i \g}{2} - \frac{\i t}{\k N}$ which is compensated by
a similar pole in $\fa^+ (v_k|{\cal U}, {\cal V}, h')$. If we
would have taken the Trotter limit first, the poles stemming from
the amplitude ratios would have developed into essential
singularities, and we would not have been able to use the residue
theorem in order to transform the sum into a sum over integrals.

In (\ref{intrepgfiniten}) the Trotter limit and the limit
$\e \rightarrow 0$ can be taken, resulting in
\begin{multline}
     {\cal G} (m, t, T, h, h') =
	{\cal A} (\emptyset, \emptyset|0) + {\cal A} (\emptyset, \emptyset|1) (-1)^m \\[1ex]
        + \sum_{\substack{\nex \in {\mathbb N}\\k = 0, 1}}
	  \frac{(-1)^{km}}{(\nex !)^2}
	     \int_{\G({\cal B}_+)^\nex} \frac{\rd^\nex u}{(2\p \i)^\nex}
	     \int_{\G({\cal B}_-)^\nex} \frac{\rd^\nex v}{(2\p \i)^\nex} \:
             {\cal A} ({\cal U}, {\cal V}|k) \\[-1ex]
          \frac{\r \bigl(- \frac{\i \g}{2}\big|{\cal U}, {\cal V}\big|k\bigr)^m
                \exp\bigl\{\frac{\i t}{\k} \6_\la \ln \r \bigl(\la |{\cal U},
		           {\cal V}|k)\bigr|_{\la = - \i \g/2}\bigr\}}
               {\prod_{j=1}^\nex \bigl(1 + \fa^- (u_j|{\cal U}, {\cal V}, h')\bigr)
                                 \bigl(1 + \fa^+ (v_j|{\cal U}, {\cal V}, h')\bigr)} \epp
\end{multline}
Finally we can also take the limit $T \rightarrow 0+$. For the
amplitudes we can then use (\ref{fine}) and for the eigenvalue
ratios the results of section~\ref{sec:evratlowt}. The integrals
further simplify, since the functions $1/(1 + \fa^\pm)$ behave
like Fermi functions for particles and holes. Using (\ref{anlowt})
and the fact that the dressed energy $\e(\la|h)$ is negative
for $\la$ between ${\cal B}_-$ and ${\cal B}_+$ and positive
for $\la$ below ${\cal B}_-$ or above ${\cal B}_+$ we see that
$1/\bigl(1 + \fa^\pm (\la|{\cal U}, {\cal V}, h')\bigr) = 1 + \CO (T^\infty)$
if $\la \in {\cal B}_\pm + \i 0$, whereas these functions are
of order $\CO(T^\infty)$ if $\la \in {\cal B}_\pm - \i 0$. Thus, the
integrals over $\G({\cal B}_\pm)$ can be replaced by integrals
over ${\cal B}_\pm + \i 0$ in the low-$T$ limit. Because of the
$\p$-periodicity of the integrands, these contours can then be
deformed into a `particle contour' ${\cal C}_p$ which lies
above $\frac{\i \g}{2}$ and a `hole contour' ${\cal C}_h$ located
above $- \frac{\i \g}{2}$ both of which, to a certain extend,
can be chosen to our convenience. For the sake of definiteness
we fix them to
\begin{equation} \label{defchcp}
     {\cal C}_p = \tst{[- \frac{\p}{2}, \frac{\p}{2}] + \frac{\i \g (1 + 0_+)}{2}} \epc \qd
     {\cal C}_h = \tst{[- \frac{\p}{2}, \frac{\p}{2}] - \frac{\i \g (1 + 0_-)}{2}}
\end{equation}
here. Hence, in the zero-temperature limit, all remaining dependence
on the magnetic fields $h, h'$ is in $\a = \frac{h - h'}{2 \g T}$
which has to be seen as an independent variable. For this reason
we shall use the notation
\begin{equation}
     {\cal G}^{(0)} (m, t, \a) =
        \lim_{T \rightarrow 0+} {\cal G} (m, t, T, h, h')
\end{equation}
for the generating function in the zero-temperature limit. As
a consequence of the above discussion the latter can be characterized
by the following theorem.

\begin{theorem}
{\bf Generating function at zero temperature.}\\
The generating function of the longitudinal dynamical two-point 
function in the zero-temperature limit has the series representation
\begin{multline} \label{genfunfinalseries}
     {\cal G}^{(0)} (m, t, \a) =
        \frac{\dh_2^2 (\i \g \a/2)}{\dh_2^2}
	+ \frac{\dh_1^2 (\i \g \a/2)}{\dh_2^2}
	  (-1)^m \\[1ex]
        + \sum_{\substack{\nex \in {\mathbb N}\\k = 0, 1}}
	  \frac{(-1)^{km}}{(\nex !)^2}
	     \int_{{\cal C}_h^\nex} \frac{\rd^\nex u}{(2\p \i)^\nex}
	     \int_{{\cal C}_p^\nex} \frac{\rd^\nex v}{(2\p \i)^\nex} \:
             {\cal A} ({\cal U},{\cal V}|k)
	     \re^{- \i \sum_{\la \in {\cal U} \ominus {\cal V}}
	     (m p(\la) + t \e (\la|h))} \epc
\end{multline}
where
\begin{multline}
     {\cal A} ({\cal U},{\cal V}|k) =
        \frac{\dh_2^2 (\Si_0 ({\cal U},{\cal V}))}{\dh_2^2}
	\biggl[\prod_{\la, \m \in {\cal U}\ominus {\cal V}}
	       \mspace{-18.mu} \Ps (\la - \m)\biggr] \\ \times
        \det_{\nex} \{\Om (u_j,v_k|{\cal U}, {\cal V})\}
	\det_{\nex} \{\overline{\Om} (v_j,u_k|{\cal U}, {\cal V})\} \epp
\end{multline}
\end{theorem}
\begin{remark}
Note that
\begin{equation} \label{omalphazero}
     \det_{\nex} \{\Om (u_j,v_k|{\cal U}, {\cal V})\}\bigr|_{\a = 0} =
	\det_{\nex} \{\overline{\Om} (v_j,u_k|{\cal U}, {\cal V})\}\bigr|_{\a = 0} = 0
\end{equation}
which follows from Lemma~\ref{lem:ombaromleftevs} below. Using
(\ref{omalphazero}) we conclude with (\ref{genfunfinalseries})
that
\begin{equation}
     {\cal G}^{(0)} (m, t, 0) = 1 \epc
\end{equation}
in accordance with our expectation.
\end{remark}
\begin{lemma}
\label{lem:ombaromleftevs}
{\bf \boldmath Left null vectors of $\Om$ and $\overline \Om$ at $\a = 0$.}\\
Fix two sets ${\cal U} \subset {\mathbb D}_- = \bigl\{z \in {\mathbb C}\big|
|\Re z| \le \frac{\p}{2}, - \frac{\g}{2} \le \Im z \le 0\bigr\}$, ${\cal V}
\subset {\mathbb D}_+ = {\mathbb D}_- + \frac{\i \g}{2}$ such that
$\card {\cal U} = \card {\cal V} = \nex$. Let
\begin{equation}
     f(\la) = \prod_{\m \in {\cal V} \ominus {\cal U}}
              \dh_1 (\la - \m|q^2)
\end{equation}
and, for all $u \in {\cal U}$, $v \in {\cal V}$,
\begin{equation}
     t_u = \frac{1}{(1/f)' (u)} \epc \qd
     \overline{t}_v = \frac{1}{f' (v)} \epp
\end{equation}
Then
\begin{equation} \label{ombaromleftnulls}
     \sum_{w \in {\cal U}} t_w
        \Om (w, v|{\cal U}, {\cal V})\bigr|_{\a = 0} = 0 \epc \qd
     \sum_{w \in {\cal V}} \overline{t}_w
        \overline{\Om} (w, u|{\cal U}, {\cal V})\bigr|_{\a = 0} = 0 \epc
\end{equation}
for all $u \in {\cal U}$, $v \in {\cal V}$.
\end{lemma}
\begin{proof}
Using (\ref{defotherdhs}) and (\ref{scdgammafuneq}) we can recast
the function $f$ in the form
\begin{equation}
     f(\la) = (-1)^k \re^{- 2 \i \Si_0 - \g \a}
              \prod_{w \in {\cal U} \ominus {\cal V}}
	         \tst{\G_{q^4} \Bigl(1 - \frac{\i (\la - w)}{2\g}\Bigr)
	              \G_{q^4} \Bigl(\frac{\i (\la - w)}{2\g}\Bigr)} \epp
\end{equation}
Then the definitions (\ref{defphpm}) of the functions $\Ph^{(\pm)}$
together with the first functional equation (\ref{gammabarnesfun}) of
the $q$-gamma function imply that
\begin{subequations}
\begin{align}
     \frac{f(\la)}{\Ph^{(+)} (\la)} & =
        (-1)^k \re^{\i \Si_0 + \g \a}
	\prod_{w \in {\cal V} \ominus {\cal U}}
	\frac{\G_{q^4} \Bigl(\2 - \frac{\i (\la - w)}{2\g}\Bigr)}
	     {\G_{q^4} \Bigl(- \frac{\i (\la - w)}{2\g}\Bigr)} \epc \\[1ex]
     \frac{f(\la + \i \g)}{\Ph^{(+)} (\la)} & =
        (-1)^k \re^{\i \Si_0 + \g \a}
	\prod_{w \in {\cal V} \ominus {\cal U}}
	\frac{\G_{q^4} \Bigl(1 + \frac{\i (\la - w)}{2\g}\Bigr)}
	     {\G_{q^4} \Bigl(\2 + \frac{\i (\la - w)}{2\g}\Bigr)} \epc \\[1ex]
     \frac{\Ph^{(-)} (\la)}{f(\la)} & =
        (-1)^k \re^{\i \Si_0 + \g \a}
	\prod_{w \in {\cal V} \ominus {\cal U}}
	\frac{\G_{q^4} \Bigl(\frac{\i (\la - w)}{2\g}\Bigr)}
	     {\G_{q^4} \Bigl(\2 + \frac{\i (\la - w)}{2\g}\Bigr)} \epc \\[1ex]
     \frac{\Ph^{(-)} (\la)}{f(\la - \i \g)} & =
        (-1)^k \re^{\i \Si_0 + \g \a}
	\prod_{w \in {\cal V} \ominus {\cal U}}
	\frac{\G_{q^4} \Bigl(\2 - \frac{\i (\la - w)}{2\g}\Bigr)}
	     {\G_{q^4} \Bigl(1 - \frac{\i (\la - w)}{2\g}\Bigr)} \epp
\end{align}
\end{subequations}
From these representations we can read off the analytic properties
of the functions on the left hand side which will become important
below.

Further note that
\begin{subequations}
\begin{align} \label{piperiodfg}
     f(\la + \p)\, g_{\pm \Si_0} (\la + \p, \m) &
        = f(\la)\, g_{\pm \Si_0} (\la, \m) \epc \\[1ex] \label{gammapperiodfg}
     f(\la + 2 \i \g)\, g_{\Si_0} (\la + 2 \i \g, \m) &
        = \re^{2 \g \a} f(\la)\, g_{\Si_0} (\la, \m) \epc \\[1ex]
     \frac{g_{- \Si_0} (\la - 2 \i \g, \m)}{f(\la - 2 \i \g)} &
        = \re^{2 \g \a} \frac{g_{- \Si_0} (\la, \m)}{f(\la)} \epp \label{gammamperiodfg}
\end{align}
\end{subequations}

Denote by ${\cal L} \subset {\mathbb C}$ the rectangle with corners
$- \frac{\p}{2} - \frac{\i \g}{2}$, $\frac{\p}{2} - \frac{\i \g}{2}$,
$\frac{\p}{2} + \frac{3 \i \g}{2}$, $- \frac{\p}{2} + \frac{3 \i \g}{2}$
and by $\6 {\cal L}$ its positively oriented boundary. Then, for
$\m \in {\mathbb D}_+$,
\begin{multline} \label{omrectint}
     \int_{\6 {\cal L}} \frac{\rd \z}{2 \p \i} f(\z) g_{\Si_0} (\z, \m) =
        \sum_{w \in {\cal U}} t_w g_{\Si_0} (w, \m) - f(\m)
	+ \re^{- 2 \i \Si_0} f(\m + \i \g) \\[-1ex]
     = \bigl(1 - \re^{2 \g \a}\bigr)
       \int_{- \frac{\p}{2}}^{\frac{\p}{2}} \frac{\rd \z}{2 \p \i}
          \tst{f\bigl(\z - \frac{\i \g}{2}\bigr)
	       g_{\Si_0} \bigl(\z - \frac{\i \g}{2}, \m\bigr)} \epp
\end{multline}
Here we have used the residue theorem in the first equation and
the quasi periodicity (\ref{piperiodfg}), (\ref{gammapperiodfg})
in the second equation. Using (\ref{omrectint}) and the fact that
$f(v) = 0$ for all $v \in {\cal V}$, we see that
\begin{multline} \label{sumtompre}
     \sum_{w \in {\cal U}} t_w \Om (w, v) =
        - \re^{- 2 \i \Si_0} f(v + \i \g) - {\cal I} (v) \\[-2ex]
	+ \bigl(1 - \re^{2 \g \a}\bigr)
	  \int_{- \frac{\p}{2}}^{\frac{\p}{2}} \frac{\rd \z}{2 \p \i}
          \tst{f\bigl(\z - \frac{\i \g}{2}\bigr) \Om \bigl(\z - \frac{\i \g}{2}, \m\bigr)} \epc
\end{multline}
where
\begin{equation}
     {\cal I} (v) = \Ph^{(-)} (v)
        \int_{- \frac{\p}{2}}^{\frac{\p}{2}} \frac{\rd \z}{2 \p \i}
	   \biggl\{\frac{f(\z)}{\Ph^{(+)} (\z)}
	           - \frac{\re^{- 2 \i \Si_0} f(\z + \i \g)}{\Ph^{(+)} (\z)} \biggr\}
		   \re(\z - v) \epp
\end{equation}
The latter integral can be calculated by means of the residue
theorem, since $f(\z)/\Ph^{(+)}(\z)$ is holomorphic and bounded
in ${\mathbb H}_+$, whereas $f(\z + \i \g)/\Ph^{(+)}(\z)$ is
holomorphic and bounded in ${\mathbb H}_-$. For this reason we
can shift the integration contour to $[-\p/2,\p/2] + \i \infty$
for the calculation of the first contribution to the integral
and to  $[-\p/2,\p/2] - \i \infty$ for the second contribution.
Altogether we obtain
\begin{equation}
     {\cal I} (v) = \frac{\Ph^{(-)} (v) f(v)}{\Ph^{(+)} (v)} -
                    \frac{\Ph^{(-)} (v) f(v + \i \g)}{\Ph^{(+)} (v + \i \g)}
                  = - \re^{- 2 \i \Si_0} f(v + \i \g) \epp
\end{equation}
In the second equation we have used equation (\ref{shiftsbyigamma})
and the fact that $f(v) = 0$ for all $v \in {\cal V}$. Substituting
this back into (\ref{sumtompre}) we see that
\begin{equation} \label{sumtom}
     \sum_{w \in {\cal U}} t_w \Om (w, v) =
	\bigl(1 - \re^{2 \g \a}\bigr)
	\int_{- \frac{\p}{2}}^{\frac{\p}{2}} \frac{\rd \z}{2 \p \i}
        \tst{f\bigl(\z - \frac{\i \g}{2}\bigr) \Om \bigl(\z - \frac{\i \g}{2}, \m\bigr)}
\end{equation}
for all $v \in {\cal V}$ which implies the first equation
(\ref{ombaromleftnulls}).

For the proof of the second equation we start by integrating
the function $g_{- \Si_0} (\cdot, \m)/f(\cdot)$ with $\m \in
{\mathbb D}_-$ over the boundary of the rectangle $\6 \overline {\cal L}$
with corners $- \frac{\p}{2} + \frac{\i \g}{2}$, $- \frac{\p}{2}
- \frac{3 \i \g}{2}$, $\frac{\p}{2} - \frac{3 \i \g}{2}$,
$\frac{\p}{2} + \frac{\i \g}{2}$ and use similar arguments as above.
\end{proof}

\subsection{The series for the two-point function}
It follows from (\ref{szszgenerated}) and (\ref{defalpha}) that
\begin{equation}
     \bigl\<\s_1^z \s_{m+1}^z (t)\bigr\> =
        \tst{\2} \6_{\g \a}^2 D_m^2 {\cal G}^{(0)} (m+1, t, \a)\bigr|_{\a = 0} \epp
\end{equation}
Using Lemma~\ref{lem:ombaromleftevs} we thus obtain the following
form factor series for the two-point function.
\begin{corollary}
In the zero-temperature limit the longitudinal dynamical two-point 
function has the series representation
\begin{multline}
     \bigl\<\s_1^z \s_{m+1}^z (t)\bigr\> = (-1)^m \frac{\dh_1'^2}{\dh_2^2} \\[.5ex]
        + \sum_{\substack{\nex \in {\mathbb N}\\k = 0, 1}}
	  \frac{(-1)^{km}}{(\nex !)^2}
	     \int_{{\cal C}_h^\nex} \frac{\rd^\nex u}{(2\p)^\nex}
	     \int_{{\cal C}_p^\nex} \frac{\rd^\nex v}{(2\p)^\nex} \:
	     {\cal A}^{zz} ({\cal U}, {\cal V}|k)
	     \re^{- \i \sum_{\la \in {\cal U} \ominus {\cal V}}
	             (m p(\la) + t \e (\la|h))} \epc
\end{multline}
where
\begin{multline} \label{azz}
     {\cal A}^{zz} ({\cal U}, {\cal V}|k) =
	(-1)^{(\nex - 1)}
	4 \sin^2 \biggl(\frac{k \p}{2}
	                + \sum_{\la \in {\cal U}\ominus {\cal V}} \frac{p(\la)}{2}\biggr)
        \frac{\dh_2^2 (\Si_0)}{\dh_2^2}
	\biggl[\prod_{\la, \m \in {\cal U}\ominus {\cal V}}
	       \mspace{-18.mu} \Ps (\la - \m)\biggr] \\[.5ex] \times
        \Bigl(\6_{\g \a} \det_{\nex} \{\Om (u_j,v_k|{\cal U}, {\cal V})\}\Bigr)
	\Bigl(\6_{\g \a} \det_{\nex} \{\overline{\Om} (v_j,u_k|{\cal U}, {\cal V})\}
	      \Bigr)\Bigr|_{\a = 0} \epp
\end{multline}
\end{corollary}
In the latter expression we take the $\a$-derivatives explicitly
and use the symmetry of the integrand to simplify the series.
As a result we obtain a formula in which the first columns of
the determinants of $\Om$ and $\overline \Om$ are modified.
Defining
\begin{equation}
     \Si_0^0 = \Si_0\bigr|_{\a = 0} \epc \qd
     g_{\Si_0^0}^{\rm mod} (\la, \m) =
        g_{\Si_0^0} (\la, \m) \frac{\dh_2' (\m - \la - \Si_0^0)}{\dh_2 (\m - \la - \Si_0^0)}
\end{equation}
and, for $j = 1, \dots, \nex$,
\enlargethispage{2ex}
\begin{subequations}
\begin{align}
     \Om^{\rm mod} (u_j, v_k|{\cal U}, {\cal V}) & =
          \begin{cases}
	  - 2 \i g_{\Si_0^0} (u_j, v_k) + g_{\Si_0^0}^{\rm mod} (u_j, v_k) \\[.5ex]
	  \mspace{-2.mu}
	  - \int_{- \p/2}^{\p/2} \frac{\rd \z}{2 \p \i}
	                         \frac{\Ph^{(-)} (v_k)}{\Ph^{(+)} (\z)} \,
				 g_{\Si_0^0}^{\rm mod} (u_j, \z) \re(\z - v_k)\
				 \text{for $k = 1$,}\\[2ex]
	  \Om (u_j, v_k|{\cal U}, {\cal V})_{\a = 0}\
	  \text{for $k = 2, \dots, \nex$,}
	  \end{cases} \\[1.5ex]
     \overline{\Om}^{\rm mod} (v_j, u_k|{\cal U}, {\cal V}) & =
          \begin{cases}
	  2 \i g_{- \Si_0^0} (v_j, u_k) + g_{- \Si_0^0}^{\rm mod} (v_j, u_k) \\[.5ex]
	  \mspace{-6.mu}
	  - \int_{- \p/2}^{\p/2} \frac{\rd \z}{2 \p \i}
	                         \frac{\Ph^{(-)} (u_k)}{\Ph^{(+)} (\z)} \,
				 g_{- \Si_0^0}^{\rm mod} (v_j, \z) \re(u_k - \z)\
				 \text{for $k = 1$,}\\[2ex]
	  \overline{\Om} (u_j, v_k|{\cal U}, {\cal V})_{\a = 0}\
	  \text{for $k = 2, \dots, \nex$}
	  \end{cases}
\end{align}
\end{subequations}
we can formulate the following theorem.
\begin{theorem}
{\bf Longitudinal two-point functions at zero temperature.}\\
The series
\begin{multline} \label{seriousseries}
     \bigl\<\s_1^z \s_{m+1}^z (t)\bigr\> = (-1)^m \frac{\dh_1'^2}{\dh_2^2} \\[1ex]
        - \sum_{\substack{\nex \in {\mathbb N}\\k = 0, 1}}
	  \frac{(-1)^{km + \nex}}{((\nex - 1) !)^2}
	     \int_{{\cal C}_h^\nex} \frac{\rd^\nex u}{(2\p)^\nex}
	     \int_{{\cal C}_p^\nex} \frac{\rd^\nex v}{(2\p)^\nex} \:
	     \re^{- \i \sum_{\la \in {\cal U} \ominus {\cal V}}
	             (m p(\la) + t \e (\la|h))} \\ \times
	\sin^2 \biggl(\frac{k \p}{2}
	              + \sum_{\la \in {\cal U}\ominus {\cal V}} \frac{p(\la)}{2}\biggr)
        \frac{\dh_2^2 (\Si_0^0)}{\dh_2^2}
	\biggl[\prod_{\la, \m \in {\cal U}\ominus {\cal V}}
	       \mspace{-18.mu} \Ps (\la - \m)\biggr] \\[1ex] \times
        \det_{\nex} \{\Om^{\rm mod} (u_j,v_k|{\cal U}, {\cal V})\}
	\det_{\nex} \{\overline{\Om}^{\rm mod} (v_j,u_k|{\cal U}, {\cal V})\}
\end{multline}
represents the longitudinal dynamical two-point functions of the
XXZ chain in the antiferromagnetic massive regime at zero temperature.
\end{theorem}
\subsection{\boldmath On the explicit evaluation of $\det_\nex \Om$ and
$\det_\nex \overline \Om$}
\label{sec:basehype}

\newcommand{\np}{\nonumber\\}
\def\pn2{\vskip 0.6cm\par\noindent}
\def\rg{\frac{3}{14}}
\def\CL{ {\cal L}}
\def\pn{\par\noindent}

\def\nt{\negthickspace}

%\begin{document}
%\noindent{\large {\bf The  explicit examples of $\Omega_n$  }}[JS ]
%\pn

In this section we shall present explicit expressions for
$\det_\nex \Omega$ and $\det_\nex \overline{\Omega}$
for  $\nex = 1, 2$.

The matrix elements $\Omega(x_j,y_k|\{x\}, \{y\})$ and 
$\overline{\Omega} (y_k,x_j|\{x\}, \{y\})$ with $\card \{x\} =
\card \{y\} = \nex$ are evaluated from (\ref{defombarom}).
Reflecting the equi-distant patterns of poles in the integrands, 
they assume very regular forms when written in terms of the basic
hypergeometric series \cite{GaRa04}, 
\begin{equation}
\phantom{ }_r \Phi_s 
\left (
{a_1, a_2, \cdots, a_r \atop
b_1, b_2, \cdots, b_s}
;q,z
\right)  
=\sum_{k=0}^{\infty}\frac{ (a_1, \cdots, a_r;q)_k}{ (b_1, \cdots, b_s,q;q)_k} 
\Bigl((-1)^k q^{\frac{k(k-1)}{2}}\Bigr)^{s+1-r} z^k.
\end{equation}
Here and hereafter we adopt the notation
\begin{equation}
     (a;q)_m = \prod_{k=0}^{m-1} (1 - a q^k) \epc \qd
(a_1, a_2,\cdots, a_k;q)_m = (a_1;q)_m   (a_2;q)_m \cdots  (a_k;q)_m
\end{equation}
for $q$-Pochhammer symbols.

The resulting $\Omega (x_j,y_k|\{x\},\{y\})$ is given by a linear
combination of $\phantom{}_{2 \nex} \Phi_{2 \nex - 1}(*;q^4,\zeta)$,
where arguments $*$ depend on $ \{ y_j\},  \{ x_j\}$ and $\alpha$.
The argument $\zeta$ also depends on $\alpha$.

The determinants, $\det_\nex \Omega$ and $\det_\nex \overline{\Omega}$,  
are already proven to be null if $\alpha=0$ in Section~\ref{sec:genfunseries}.
One can check that the row vectors (and column vectors) of the matrix
$\{\Omega (x_j,y_k|\{x\},\{y\})\}_{j,k=1}^\nex$ are not simply proportional
to each other.
This implies the existence of nontrivial identities among the $\nex$-products
of $\phantom{}_{2 \nex} \Phi_{2 \nex - 1}$ at $\alpha=0$. Even for $\nex = 2$,
interesting identities, such as non-terminating $q$-Saalsch{\"u}tz formulae,
naturally appear.
This might pose an interesting problem in the theory of basic hypergeometric series.
 
We, however, need to go beyond $\alpha=0$, as the derivatives of the
determinants $\det_\nex \Om$ and $\det_\nex \overline \Om$ at $\alpha=0$
are of our concern. There are many equivalent expressions for
$\Omega (x_j,y_k|\{x\},\{y\})$. This arbitrariness does not matter
much for $\nex = 1$, whereas one needs to take suitable linear
combinations and to select appropriate forms, in order to proceed with
$\nex \ge 2$ and  $\alpha \ne 0$. Their derivation thus becomes slightly
involved, and we decided to present the details in a separate publication
which will also include explicit numerical results.
Here we shall summarize the formulae for $\nex = 1, 2$ and their relation
to previous results.

We will sometimes use exponentiated variables, 
\[
     P_j={\rm e}^{2\i y_j} \epc \quad  H_k={\rm e}^{2\i x_k} \epp
 \]
For any function $f(y_1,\cdots,  x_1\cdots,\alpha)$ we write
\begin{equation}\label{defbar}
     \overline{f}(y_1,\cdots, x_1\cdots,\alpha)
        = f(-y_1,\cdots, -x_1\cdots,-\alpha) \epp
\end{equation}
The same notation is applied to a function with exponentiated variables,
\begin{equation}%\label{defbar}
     \overline{F}({P_1},\cdots, {H_1}, \cdots, {\rm e}^{\alpha \gamma})
        = F\biggl(\frac{1}{P_1},\cdots,\frac{1}{H_1},\cdots,{\rm e}^{-\alpha \gamma}\biggr) \epp
\end{equation}
By $f(\{y_1,y_2\},\cdots)$ we mean that $f$ is symmetric w.r.t.\
$y_1 \leftrightarrow y_2$.  

We then introduce
\begin{equation}\label{defmu}
     \mu(\{y\},\{x\},\alpha)
        = \frac{\vartheta'_1(0|q)  \vartheta_2(0|q) }
	       {2\vartheta_4(0|q^2) \vartheta_1(\sum_k (y_k - x_k)
	        + \i \alpha \gamma|q^2) } \epc
\end{equation}
where $\{y\}$ and $\{x\}$ mean that the function is symmetric with respect
to permutations of $\{y\}$ and $\{x\}$.% We will abbreviate $\mu$ to mean $\mu(\{y\},\{x\},0)$.

We first state the result for $\nex = 1$. Set further
\begin{equation}\label{defo_ne1}
     \mathfrak{o}_1(y_1, x_1)=
        \frac{1}{\bigl(q^2 \frac{H_1}{P_1}, q^4\frac{P_1}{H_1};q^4\bigr)_{\infty}}
\end{equation}
and
\begin{equation} \label{defPhi_ne1}
     \Phi^{\nex = 1}_1 (w_1,w_2,\alpha) =
        \phantom{}_2 \Phi_1 \Biggl(
	\begin{array}{@{}r@{}}
	   q^{-2}, \frac{Z_1}{Z_2} \\ q^2 \frac{Z_1}{Z_2}
        \end{array};
	q^4, q^4 {\rm e}^{-2\alpha \gamma} \Biggr) \epc  
\end{equation}
where $Z_j={\rm e}^{2\i w_j}$.
%$\Phi^{n_e=1}_1 (y_1,y_2)$ stands for $\Phi^{n_e=1}_1 (y_1,y_2,0)$.

\begin{lemma}
Let $\nex = 1$. For arbitrary $\alpha$ the function $\Omega (x_1,y_1|\{x\},\{y\})$
is explicitly given by 
\begin{multline}\label{Omega_ne1}
     \Omega (x_1,y_1|\{x\},\{y\}) = \mu(y_1,x_1,\alpha)
        \frac{\vartheta_1(\i \alpha \gamma|q^2)}{\vartheta_1(y_1-x_1|q^2)}
        \biggl(\overline{\Phi}^{\nex = 1}_1 (y_1,x_1,\alpha) \\[1ex]
        - {\rm e}^{-2i\Sigma_0} \frac{\mathfrak{o}_1(y_1,x_1) }
	                             { \overline{\mathfrak{o}}_1(y_1,x_1) }
				     {\Phi}^{\nex = 1}_1 (y_1,x_1,\alpha) \biggr) \epp
\end{multline}
\end{lemma}

This is clearly zero when $\alpha=0$. Thanks to the $q$-Gauss identity
we have a simple expression for the first term in the small-$\alpha$
expansion,
\begin{multline} \label{Omegaalpha1}
     \Omega (x_1, y_1|\{x\},\{y\}) = \i \alpha \gamma
        \frac{(1 -Z^{-1/2} (-1)^k) }{\sin (w)}
	\frac{\vartheta'_1(0|q)}{\vartheta_1(w|q)} 
     \frac{(q^4, q^2 Z; q^4)_{\infty}}{(q^2, q^4 Z;q^4)_{\infty}} +{\cal O}(\alpha^2) \\[1ex]
     =: \i \alpha \gamma\, \omega(x_1, y_1)  +{\cal O}(\alpha^2) \epc
\end{multline}
where $w=y_1-x_1$ and $Z={\rm e}^{2\i w}$. The expression for
$\overline{\Omega} (y_1,x_1|\{x\},\{y\})$ is simply related,
\begin{equation}\label{barOmega_ne1}
     \overline \Omega (y_1,x_1|\{x\},\{y\}) = - \Omega (x_1,y_1|\{x\},\{y\}) \epp
\end{equation}

We shall comment on the consistency with the result  from 
the vertex operator approach, which implies the form factor expansion
in the spinon basis. For simplicity we consider the formula in the
static case,
\begin{multline} \label{formffserieszzspinon}
\langle \sigma_1^z \sigma_{m+1}^z \rangle =
     (-1)^m \frac{(q^2;q^2)^4}{(-q^2;q^2)^4} \\[1ex] +
     \sum_{\substack{\nex \in {\mathbb N}\\k = 0, 1}} \frac{(-1)^{km}}{(2\nex)!}
     \int_{-\frac{\pi}{2}}^{\frac{\pi}{2}}  \frac{d^{2\nex} u}{(2\pi)^{2\nex}}  
     \re^{\i m \sum_{j=1}^{2\nex} p(u_j)}
     {\cal F}^{zz} (\{u_i\}_{i=1}^{2\nex}|k) \epc
\end{multline}
where  ${\cal F}^{zz} (\{u_i\}_{i=1}^{2\nex}|k)$ denotes the contribution
from $2\nex$ spinons. For arbitrary $\nex$, there exists a multiple-integral
formula with complicated integration contours for this quantity.
When $\nex = 1$, it reduces to a simple formula due to Lashkevich
\cite{Lashkevich02,DGKS16a},
\begin{multline} \label{ampfun}
    {\cal F}^{zz} (\{u_1, u_2\}|k) =
        \frac{ \sin^2 \2 \bigl(p(u_1) + p(u_2) + k\pi \bigr) 
	\sin^2 (u_{12})\, \vartheta_3^2
	   \bigl( \frac{u_{12}+\pi k}{2}|q \bigr)}
	               {\cos \frac{1}{2}(u_{12} + \i \gamma + k\pi)
		        \cos \frac{1}{2}(u_{12} - \i \gamma + k\pi) } \\[1ex]
	\times 32 q (q^2;q^2)^2 \prod_{\sigma = \pm}
        \frac{(q^4;q^4,q^4)^2}{(q^2;q^4,q^4)^2}
        \frac{(q^4  {\rm e}^{2 \i \sigma u_{12}};q^4,q^4)^2}
	     {(q^2  {\rm e}^{2 \i \sigma u_{12}};q^4,q^4)^2}
	\frac{(q^4  {\rm e}^{2 \i \sigma u_{12}};q^4)}
	     {(q^2  {\rm e}^{2 \i \sigma u_{12}};q^4)} \epc
\end{multline}
where $u_{12} = u_1 - u_2$.

One can easily show that this  is consistent with the result obtained here  
if the following  (Conjecture 1 in \cite{DGKS16b}) is valid,
\begin{equation}\label{conjecture_eq_1}
     {\cal F}^{zz}(\{u_1, u_2\}|k) =
          {\cal A}^{zz} (\{u_1 - \i \gamma \}, \{u_2\}|k)
	+ {\cal A} ^{zz} (\{u_2 - \i \gamma\}, \{u_1\}|k), 
\end{equation}
for $0 < {\rm Im}\, u_j <\gamma , \, j=1,2$ (recall that
${\cal A}^{zz}$ was defined in (\ref{azz})).
 
This was verified numerically in  \cite{DGKS16b}, as there existed
only an expression involving Fredholm determinants for ${\cal A}^{zz}$.
With the present result we have now an expression for ${\cal A}^{zz}$,
involving $\omega$ in (\ref{Omegaalpha1}) instead, whose analytic properties
are fully under control,
\begin{multline} \label{Azz_n1}
     {\cal A}^{zz}(\{u_1- \i \gamma \},\{u_2\}|k) =
     - 4 \sin^2 \tst{\2} \bigl(p(u_2) - p(u_1 - \i \gamma) + k \p \bigr)\\[1ex] \times
         \omega(u_1- \i \gamma, u_2)^2 \frac{\vartheta^2_2( -\frac{\pi k}{2}
	+ \frac{u_{21} + \i \gamma}{2}|q)}{\vartheta^2_2(0|q)}
	\frac{\Psi^2(0)}{\Psi(u_{12} - \i \gamma) \Psi(u_{21} + \i \gamma)} \epp
\end{multline}
Thus, now is the right occasion to prove the validity of (\ref{conjecture_eq_1}).
We utilize the anti-periodicity of the dressed momentum, $p(u - \i \gamma)
= - p(u)$, the explicit forms of $\Ps$ and $\omega$, and 
\begin{align}
     \frac{\Gamma_{q^4}^2(\frac{1}{2})}
            {\Gamma_{q^4}(\frac{1}{2} - \frac{\i x}{2\gamma})
	     \Gamma_{q^4}(\frac{1}{2} + \frac{\i x}{2\gamma})}
       & = \frac{\vartheta_4(x|q^2)}{\vartheta_4(0|q^2)} \epc \notag \\[1ex]
     \frac{1}{\Gamma_{q^4}(1- \frac{\i x}{2\gamma})\Gamma_{q^4}(1+ \frac{\i x}{2\gamma})}
       & = \frac{\vartheta_1(x|q^2)}{\sin(x)\, \vartheta'_1(0|q^2)} \epc \notag \\[1ex]
     \frac{G_{q^4}(1 - \frac{\i x}{2\gamma})G_{q^4}(1 + \frac{\i x}{2\gamma})}
            {G_{q^4}(\frac{1}{2} - \frac{\i x}{2\gamma})G_{q^4}(\frac{1}{2} + \frac{\i x}{2\gamma})}
       & = (1-q^4)^{\frac{3}{4}} (q^4;q^4)
	  \prod_{\sigma=\pm} \frac{(q^4 {\rm e}^{2\i \sigma x};q^4,q^4)}
	                          {(q^2 {\rm e}^{2\i \sigma x};q^4,q^4)} \epp
\end{align}
Then, after simple manipulations, one arrives at
\begin{multline}
     {\cal A}^{zz}(\{u_1 - \i \gamma \},\{u_2\}|k)
        = 32 q (q^2;q^2)^2   \sin^2 \tst{\2} \bigl(p(u_1) + p(u_2) + k \pi \bigr) \\[1ex] \times
	  \sin  (u_{21}) \tg \tst{\bigl(\frac{u_{21} + \i \gamma + k \pi}{2} \bigr)}
	  \vartheta_3^2 \tst{\bigl(\frac{u_{12}+\pi k}{2}| q \bigr)} \\[1ex]
	\times\prod_{\sigma = \pm}
        \frac{(q^4;q^4,q^4)^2}{(q^2;q^4,q^4)^2}
        \frac{(q^4  {\rm e}^{2 \i \sigma u_{12}};q^4,q^4)^2}
	     {(q^2  {\rm e}^{2 \i \sigma u_{12}};q^4,q^4)^2}
	\frac{(q^4  {\rm e}^{2 \i \sigma u_{12}};q^4)}
	     {(q^2  {\rm e}^{2 \i \sigma u_{12}};q^4)} \epp
\end{multline}
By interchanging  $u_1\leftrightarrow u_2$ in the above and by summing up,
one immediately verifies (\ref{conjecture_eq_1}). We thus obtain,
\begin{corollary}
The 1-ph excitation brings the same contribution to the form factor
expansion of the longitudinal correlation function as the 2-spinon excitation.
\end{corollary}
In order to present the result for the 2-ph excitations, we need a
few more objects. The analogue of (\ref{defo_ne1}) is defined by
\begin{equation}
     {\mathfrak o}_2(y_1,y_2,\{x\})
        = \frac{\bigl( q^2\frac{P_2}{P_1}, q^4\frac{P_1}
	       {P_2};q^4\bigr)_{\infty}}
          {\bigl( q^2\frac{H_1}{P_1},q^2\frac{H_2}{P_1},
	          q^4\frac{P_1}{H_1}, q^4\frac{P_1}{H_2};q^4\bigr)_{\infty}} \epp
		  \label{defo2}  
\end{equation}
This time we need two kinds of basic hypergeometric series, 
\begin{align}
     \Phi_1(y_1,y_2,\{x\},\alpha) &
        = \phantom{}_4 \Phi_3
	  \Biggl(
	     \begin{array}{@{}r@{}}
	        q^{-2},  \frac{P_1} {H_1}, \frac{P_1} {H_2}, q^2\frac{P_1}{P_2} \\
                q^2\frac{P_1} {H_1}, q^2 \frac{P_1} {H_2}, \frac{P_1}{P_2}
	     \end{array};
                q^4, q^4 \re^{-2\alpha \gamma}\Biggr) \epc \label{Phi1_ne2} \\[1ex]
     \Phi_2(y_1,y_2,\{x\},\alpha) &
        = \phantom{}_4 \Phi_3
          \Biggl(
	     \begin{array}{@{}r@{}}
	        q^6, q^4 \frac{P_2} {H_1}, q^4\frac{P_2} {H_2}, q^2\frac{P_2}{P_1} \\
                q^6 \frac{P_2} {H_1}, q^6\frac{P_2} {H_2}, q^8\frac{P_2}{P_1}
	     \end{array};
	        q^4, q^4 \re^{-2\alpha \gamma}\Biggr) \epp \label{Phi2_ne2}
\end{align}
The first one is the $\nex = 2$ counterpart of $\Phi_1^{\nex = 1}$.
In addition, we set
\begin{equation}\label{def_r2}
     r_2(y_1, y_2,\{x\}) =
        \frac{q^2(1-q^2)^2 \frac{P_2}{P_1}}
	     {(1-\frac{P_2}{P_1}) (1-q^4 \frac{P_2}{P_1})} 
        \prod_{k=1,2} \frac{(1-\frac{P_2}{H_k})}{(1-q^2 \frac{P_2}{H_k})} \epp
\end{equation}
We then introduce linear combinations of these objects,
\begin{align}
     \Phi_{[1]} (y_1,y_2,\{x \},\alpha)
         & = \frac{{\mathfrak o}_2(y_1,y_2,\{x \})}
	          {\overline{\mathfrak o}_2(y_1,y_2,\{x \})}
		  \Phi_1(y_1,y_2,\{x \},\alpha) % \notag \nolinebreak \\
         - {\rm e}^{2\i \Sigma_0} \overline{\Phi}_1(y_1,y_2,\{x \},\alpha) \epc
	     %\label{def_Phi_1_}
	     \notag \\
     \Phi_{[2]} (y_1,y_2,\{x \},\alpha)
        & = {\rm e}^{-2\alpha \gamma} r_2(y_1,y_2,\{x \})
	    \frac{{\mathfrak o}_2(y_2,y_1,\{x \})}
	         {\overline{\mathfrak o}_2(y_2,y_1,\{x \})}
		 \Phi_2(y_1,y_2,\{x \},\alpha)   \notag \\[.5ex] & \mspace{90.mu}
          - {\rm e}^{2\i \Sigma_0} {\rm e}^{2\alpha \gamma} \bar{r}_2(y_1,y_2,\{x \})
            \overline{\Phi}_2(y_1,y_2,\{x \},\alpha) \epp
	    \label{def_Phi_2_}
\end{align}
Note the similarly between the above $\Phi_{[1]}$ and the content of the
bracket in equation (\ref{Omega_ne1}) for $\Omega$. Finally, we introduce
\begin{multline} \label{delta0}
     \delta_0 ([y],[x],\alpha) =
        \frac{\vartheta_1(y_2-x_2+\i\alpha \gamma|q^2)}{\vartheta_1(y_2-x_2|q^2)}
        \frac{\vartheta_1(y_1-x_1+\i\alpha \gamma|q^2)}{\vartheta_1(y_1-x_1|q^2)} \\[1ex]
      - \frac{\vartheta_1(y_2-x_1+\i\alpha \gamma|q^2)}{\vartheta_1(y_2-x_1|q^2)}
        \frac{\vartheta_1(y_1-x_2+\i\alpha \gamma|q^2)}{\vartheta_1(y_1-x_2|q^2)} \epp
\end{multline}
By $[y]$ we mean that $\delta_0$ is anti-symmetric in $y_1$ and $y_2$.
Similarly for $[x]$.

We are now in position to write down the result of $\det_2
\Omega (x_j, y_k|\{x\},\{y\})$.
\begin{lemma}
Let $\nex = 2$. Then, for arbitrary $\alpha$,  $\det_2 \Omega (x_j, y_k|\{x\},\{y\})$
is given by
\begin{multline} \label{detOmega2}
     \tst{\det_2 \Omega (x_j, y_k|\{x\},\{y\})}
        = \delta_0 ([y],[x],\alpha)\, \mu(\{y\},\{x\},\alpha)^2 \re^{-4\i \Sigma_0} \\[1ex]
	  \times \det
          \begin{vmatrix}
	     \Phi_{[1]} (y_1,y_2,\{x \},\alpha)&   \Phi_{[2]} (y_1,y_2,\{x \},\alpha) \\
	     \Phi_{[2]} (y_2,y_1,\{x \},\alpha)&  \Phi_{[1]} (y_2,y_1,\{x \},\alpha)
	  \end{vmatrix} \epp
\end{multline}
\end{lemma}

We remark that this is one of the possible equivalent forms and that it
has several advantages over others. First, the anti-symmetry in $\{x\}$
and $\{y\}$, that originated from the definition, is reflected only in
$\delta_0$ and the other parts are symmetric in $\{x\}$ and $\{y\}$,
respectively. Second, it is clear from (\ref{delta0}) that
$\delta_0 ([y],[x],0) = 0$ and thus the determinant is explicitly
shown to be zero if $\alpha=0$. Third, the first term in the expansion
w.r.t.\ $\alpha$ can be immediately found: except in $\delta_0$,
one only has to set $\alpha=0$,
\begin{multline} \label{detOmega22}
     \tst{\det_2 \Omega (x_j, y_k|\{x\},\{y\})} = \i \alpha \gamma
        \ln' \biggl[\frac{\vartheta_1(y_2-x_2|q^2) \vartheta_1(y_1-x_1|q^2)}
	                 {\vartheta_1(y_2-x_1|q^2) \vartheta_1(y_1-x_2|q^2)} \biggr]
        \mu(\{y\},\{x\},0)^2 \\[1ex] \times \frac{H_1 H_2}{P_1 P_2}
       \det
       \begin{vmatrix}
          \Phi_{[1]} (y_1,y_2,\{x \},0) & \Phi_{[2]} (y_1,y_2,\{x \},0) \\
	  \Phi_{[2]} (y_2,y_1,\{x \},0) & \Phi_{[1]} (y_2,y_1,\{x \},0)
       \end{vmatrix} + \CO (\alpha^2) \epc
\end{multline}
where $\ln'$ stands for the logarithmic derivative.

We find that $\det_2 \overline{\Omega} ( y_k, x_j|\{x\},\{y\})$ is also
proportional to $\delta_0$ and thus we can safely set $\alpha=0$
in the rest. In place of (\ref{def_Phi_2_}) we set
\begin{align}
     & \Phi^{\cal T}_{[1]} (x_1,x_2,\{y \}) =
       \frac{\overline{{\mathfrak o}}(x_1,x_2,\{y \})}
            {{\mathfrak o}(x_1,x_2,\{y \})}
	    \overline{\Phi}_1(x_1,x_2,\{y \},0) \notag \\[.5ex]
     & \mspace{234.mu} + (-1)^{k-1} \re^{\i \sum_k (y_k -x_k)}
       {\Phi}_1(x_1,x_2,\{y \},0) \epc \notag \\[1.5ex]   %\sqrt{\frac{P_1 P_2}{H_1 H_2} }  
     & \Phi^{\cal T}_{[2]} (x_1,x_2,\{y \}) =
       \overline{r}_2(x_1,x_2,\{y \})
       \frac{\overline{{\mathfrak o}}(x_2,x_1,\{y \})}
            {{\mathfrak o}(x_2,x_1,\{y \})}
	    \overline{\Phi}_2(x_1,x_2,\{y \},0) \notag \\[.5ex]
     & \mspace{144.mu} + (-1)^{k-1} {\rm e}^{\i \sum_k (y_k -x_k)}
       {r}_2(x_1,x_2,\{y \}) {\Phi}_2(x_1,x_2,\{y \},0) \epp  %\sqrt{\frac{P_1 P_2}{H_1 H_2} }
\end{align}
Then the small-$\alpha$ expansion of $\det_2 \overline{\Omega} (y_k, x_j|\{x\},\{y\})$
reads
\begin{multline} \label{detOmegabar}
     \tst{\det_2 \overline{\Omega} (y_j, x_i|\{x\},\{y\})} = \i \alpha \gamma
        \ln' \biggl[\frac{\vartheta_1(y_2-x_2|q^2) \vartheta_1(y_1-x_1|q^2)}
	                 {\vartheta_1(y_2-x_1|q^2) \vartheta_1(y_1-x_2|q^2)} \biggr]
			 \mu(\{y\},\{x\},0)^2 \\[1ex] \times \frac{H_1 H_2}{P_1 P_2}
     \det
     \begin{vmatrix}
        \Phi^{\cal T}_{[1]} (x_1,x_2,\{y \}) & \Phi^{\cal T}_{[2]} (x_1,x_2,\{y \}) \\
	\Phi^{\cal T}_{[2]} (x_2,x_1,\{y \}) &  \Phi^{\cal T}_{[1]} (x_2,x_1,\{y \})
     \end{vmatrix} + \CO (\alpha^2) \epp
\end{multline}

In a subsequent paper, we will apply the above formulae to the
numerical investigation of the dynamical correlation functions
and exemplify their efficiency. The details of their derivation
will also be explained there.

Finally, we comment on our previous ${\cal A}_{\rm old}^{zz}$
\cite{DGKS16b} involving the Fredholm determinants. It might be
interesting to compare it numerically against the present one (denoted by
${\cal A}^{zz}_{\rm new}$ for comparison), as the actual computation of
${\cal A}_{\rm old}^{zz}$ required a discretized approximation.
We consider $\nex = 2$, fix $x_1 = 1 - \frac{\i \gamma}{2}$,
$y_1= \frac{1}{2} + \frac{\i \gamma}{2}$ and evaluate the
relative difference, $\bigl|1-\frac{{\cal A}^{zz}_{\rm new}}
{{\cal A}^{zz}_{\rm old}}\bigr|$, with $x_2 = u_0 - \frac{\i \gamma}{2}$
and $y_2= v_0 + \frac{\i \gamma}{2}$ for various $u_0$ and $v_0$.
The Fredholm determinants are approximated by determinants of 
$50 \times 50$ matrices, see Fig.~\ref{fig:comparisonn2}.

\begin{figure}[!h]
\centering
\includegraphics[width=6cm]{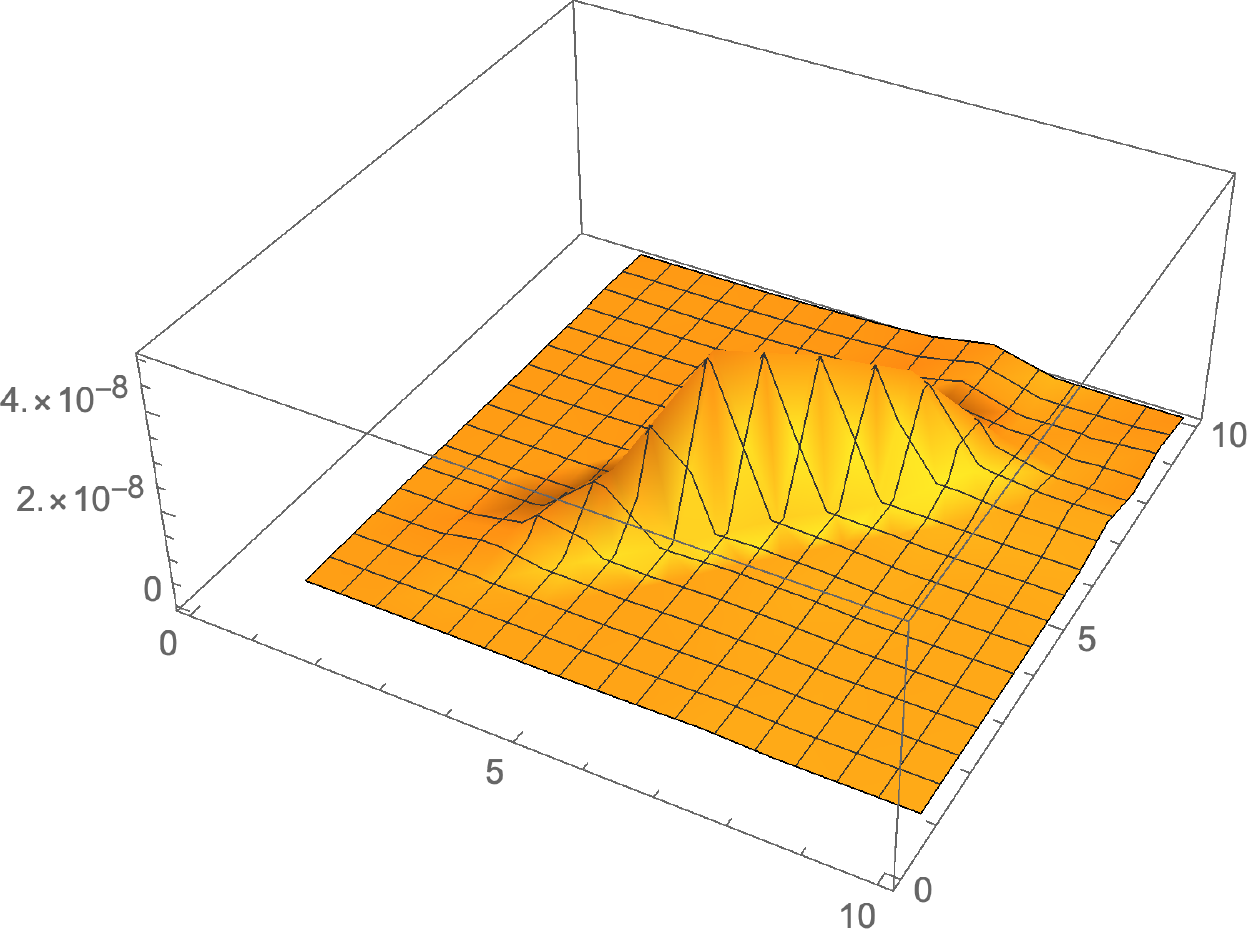} 
\caption{ 
$\bigl|1-\frac{{\cal A}^{zz}_{\rm new}}{{\cal A}^{zz}_{\rm old}}\bigr|$
for $q=0.5 , k=0, \nex = 2$.
}\label{fig:comparisonn2}
\end{figure}
The maximum value of the relative difference is $\sim 3\times 10^{-8}$.
We thus conclude that they agree with rather nice precision and that the
discretization scheme proposed in \cite{Bornemann10} is numerically efficient
in our case.

%\begin{thebibliography}{10}
%
%\bibitem{GasperRahman}
%G.~Gasper and M.~Rahman, `Basic Hypergeometric Series',
%(Cambridge University Press).
%
%\bibitem{Lashkevich02}
%M.~Lashkevich, `Free field construction for the eight-vertex model:
%representation for form factors', Nucl. Phys. B \textbf{621} (2002), 587.
%
%
%\bibitem{DugaveGohmannKozlowskiSuzuki16}
%M.~ Dugave, F.~ G{\"o}hmann, K.~K.~Kozlowski, and J.~Suzuki,
%`Thermal form factor approach to the ground-state correlation functions of the XXZ chain in the antiferromagnetic massive regime.'
%J.~Phys.~A 49 (2016) 39, 394001.
%
%\bibitem{Bornemann10}
%F.~Bornemann, \emph{On the numerical evaluation of {F}redholm determinants},
%Mathematics of Computation \textbf{79} (2010), 871.
%
%
%\end{thebibliography}
%
%\end{document}

\section{The isotropic limit}
The isotropic point $\D = 1$, $h = 0$ of the ground state phase
diagram of the XXZ chain is located at the boundary between
the antiferromagnetic massive and massless regimes (see
Figure~\ref{fig:phasediagram}). It can be accessed from the
antiferromagnetic massive regime, e.g.\ by first sending
$h \rightarrow 0$ and then $\g \rightarrow 0$. The first limit
is trivial at zero temperature, since in this case the
correlation functions are independent of the magnetic field.
The isotropic limit requires a rescaling of the integration
variables $u, v \rightarrow \g u, \g v$. Sending $\g
\rightarrow 0$ means sending $q \rightarrow 1$. In this
limit the functions of the $q$-gamma family become
functions of the ordinary gamma family, e.g.\
$\lim_{q \rightarrow 1} \G_q (u) = \G (u)$ and
$\lim_{q \rightarrow 1} G_q (u) = G (u)$. For those functions
involving Jacobi theta functions the limit can be easily
calculated after employing a special modular transformation,
\begin{align}
     \dh_1 (\la|q) & = - \i \sqrt{\frac{\p}{\g}} \re^{- \frac{\la^2}{\g}}
                       \dh_1 \bigl(\tst{\frac{\p \i \la}{\g}}\big|q'\bigr) \epc
     & \dh_2 (\la|q) & =  \sqrt{\frac{\p}{\g}} \re^{- \frac{\la^2}{\g}}
                       \dh_4 \bigl(\tst{\frac{\p \i \la}{\g}}\big|q'\bigr) \epc \notag \\
     \dh_3 (\la|q) & =  \sqrt{\frac{\p}{\g}} \re^{- \frac{\la^2}{\g}}
                       \dh_3 \bigl(\tst{\frac{\p \i \la}{\g}}\big|q'\bigr) \epc
     & \dh_4 (\la|q) & =  \sqrt{\frac{\p}{\g}} \re^{- \frac{\la^2}{\g}}
                       \dh_2 \bigl(\tst{\frac{\p \i \la}{\g}}\big|q'\bigr) \epc
\end{align}
where
\begin{equation}
     q' = \re^{- \frac{\p^2}{\g}} \epp
\end{equation}

We indicate the isotropic limit by putting a hat over the
respective function, $\hat f (u) = \lim_{\g \rightarrow 0} f(\g u)$.
Then the momentum and dressed energy go to
\begin{equation}
     \hat p (u) = \frac{\p}{2}
        - \i \ln \biggl(
        \frac{\ch \bigl( \frac \p 2(u + \frac \i 2) \bigr)}
             {\ch \bigl( \frac \p 2(u - \frac \i 2) \bigr)} \biggr) \epc \qd
     \hat \e(u|0) = - \frac{2\p J}{\ch(\p u)} \epp
\end{equation}
The function $\Ps$ becomes
\begin{equation}
     \hat \Ps (u) = \G \bigl(\tst{\2} - \tst{\frac{\i u}{2}}\bigr)
                    \G \bigl(1 - \tst{\frac{\i u}{2}}\bigr)
		    \frac{G^2 \bigl(1 - \tst{\frac{\i u}{2}}\bigr)}
                         {G^2 \bigl(\2 - \tst{\frac{\i u}{2}}\bigr)} \epp
\end{equation}

Severe simplifications of the series in the isotropic limit
arise from the theta-function factors. First of all, as is well
known, the staggered polarization vanishes in this limit, $\dh_1'/\dh_2
\rightarrow 0$. Moreover,
\begin{equation}
     \frac{\dh_2 (\Si_0)}{\dh_2}\biggr|_{\a = 0}
	\rightarrow
        \begin{cases}
	   1 & \text{for $k = 0$}\\
	   0 & \text{for $k = 1$,}
	\end{cases}
\end{equation}
implying that the whole `staggered part' of the series
(\ref{seriousseries}) vanishes for $\g \rightarrow 0$,
since the remaining factors for $k=1$ have a finite limit.
For this reason is enough to give an explicit description
of these remaining factors only for $k=0$. If we denote
\begin{subequations}
\begin{align}
     \hat \Om (u, v|{\cal U}, {\cal V}) & =
        \lim_{\g \rightarrow 0} \g \,
        \Om (\g u, \g v|\g {\cal U}, \g {\cal V})
	\bigr|_{k = 0, \, \a = 0} \epc \\
        \hat{\overline \Om} (u, v|{\cal U}, {\cal V}) & =
        \lim_{\g \rightarrow 0} \g \,
	\overline \Om (\g u, \g v|\g {\cal U}, \g {\cal V})
	\bigr|_{k = 0, \, \a = 0} \epc
\end{align}
\end{subequations}
then
\begin{subequations}
\begin{align}
     \hat \Om (u, v|{\cal U}, {\cal V}) & =
        \hat g(u,v) - \int_{- \infty}^\infty \frac{\rd t}{2 \p \i}
	              \frac{\hat \Ph^{(-)} (v)}{\hat \Ph^{(+)} (t)}
		      \hat g(u,t) \hat \re(t - v) \epc \\[1ex]
     \hat{\overline \Om} (v, u|{\cal U}, {\cal V}) & =
        \hat g(v,u) - \int_{- \infty}^\infty \frac{\rd t}{2 \p \i}
	              \frac{\hat \Ph^{(-)} (t)}{\hat \Ph^{(+)} (u)}
		      \hat g(v,t) \hat \re(u - t) \epc 
\end{align}
\end{subequations}
where
\begin{subequations}
\begin{align}
     \hat g(u,v) & = \frac{\p}{\sh\bigl(\p (v - u)\bigr)} \epc \qd
        \hat \re(u) = \frac{1}{u} - \frac{1}{u - \i} \epc \\[1ex]
     \Ph^{(\s)} (u) & = \prod_{w \in {\cal U} \ominus {\cal V}}
        \G \Bigl(\tst{\2 - \frac{\s \i (u - w)}{2}}\Bigr)
        \G \Bigl(\tst{1 + \frac{\s \i (u - w)}{2}}\Bigr) \epc \qd \s = \pm \epp
\end{align}
\end{subequations}
For $j = 1, \dots, \nex$ we further define
\begin{subequations}
\begin{align}
     \hat{\Om}^{\rm mod} (u_j, v_k|{\cal U}, {\cal V}) & =
        \begin{cases}
	   \int_{- \infty}^\infty \frac{\rd t}{\p} 
	   \frac{\hat \Ph^{(-)} (v_k)}{\hat \Ph^{(+)} (t)}
	   \frac{\hat g(u_j, t)}{v_k - t} & \text{for $k = 1$,} \\[1.5ex]
           \hat \Om (u_j, v_k|{\cal U}, {\cal V}) & \text{for $k = 2, \dots, \nex$,}
	\end{cases}
     \\[1ex]
     \hat{\overline{\Om}}\raisebox{10.pt}{}^{\rm mod} (v_j, u_k|{\cal U}, {\cal V}) & =
        \begin{cases}
	   \int_{- \infty}^\infty \frac{\rd t}{\p} 
	   \frac{\hat \Ph^{(-)} (t)}{\hat \Ph^{(+)} (u_k)}
	   \frac{\hat g(v_j, t)}{u_k - t} & \text{for $k = 1$,} \\[1.5ex]
           \hat{\overline \Om} (v_j, u_k|{\cal U}, {\cal V})
	   & \text{for $k = 2, \dots, \nex$.}
	\end{cases}
\end{align}
\end{subequations}
Inserting all of the above into equation (\ref{seriousseries}) we have
arrived at the following theorem.
\begin{theorem}
{\bf Two-point functions in the isotropic limit.}\\
The dynamical two-point correlation functions of the isotropic
Heisenberg chain in the zero-temperature limit have the form
factor series representation
\begin{multline} \label{xxxseries}
     \bigl\<\s_1^z \s_{m+1}^z (t)\bigr\> = % \\[1ex]
        \sum_{\nex \in {\mathbb N}}
	  \frac{(- 1)^{\nex - 1}}{((\nex - 1) !)^2}
	     \int_{\hat {\cal C}_h^\nex} \frac{\rd^\nex u}{(2\p)^\nex}
	     \int_{\hat {\cal C}_p^\nex} \frac{\rd^\nex v}{(2\p)^\nex} \:
	     \re^{- \i \sum_{\la \in {\cal U} \ominus {\cal V}}
	             (m \hat p(\la) + t \hat \e (\la|h))} \\[1ex] \times
	\sin^2 \biggl\{\sum_{\la \in {\cal U}\ominus {\cal V}} \frac{\hat p(\la)}{2}\biggr\}
	\biggl[\prod_{\la, \m \in {\cal U}\ominus {\cal V}}
	       \mspace{-18.mu} \hat \Ps (\la - \m)\biggr] \\[1ex] \times
        \det_{\nex} \{\hat \Om^{\rm mod} (u_j,v_k|{\cal U}, {\cal V})\}
	\det_{\nex} \{\hat{\overline{\Om}}\raisebox{10.pt}{}^{\rm mod}
	              (v_j,u_k|{\cal U}, {\cal V})\} \epc
\end{multline}
where $\hat {\cal C}_p = {\mathbb R} + \frac{\i (1 + 0_+)}{2}$,
$\hat {\cal C}_h = {\mathbb R} - \frac{\i (1 + 0_-)}{2}$.
\end{theorem}

\section{Conclusions}
We have obtained novel form factor series representations for the
longitudinal two-point correlation functions of the XXZ chain
in the antiferromagnetic massive regime and of the XXX chain
at vanishing magnetic field. These series take simple forms
at $T = 0$. They are series of multiple integrals of increasing
even multiplicity, 2, 4, 6, \dots. In previous works the
integrands in the general terms were of rather intricate forms.
They consisted either of sums of multiple contour integrals
with complicated contours \cite{JiMi95}, of sums over multiple
residues \cite{DGKS15a}, or, in the best case \cite{DGKS16b},
of functions involving Fredholm determinants in their definition.
In the present work the integrands in the $2 \nex$-fold integrals
are basically products of two order-$\nex$ determinants with
entries that can be represented as simple integrals over special
functions from the $q$-gamma family.

Expressions for form factors not involving multiple integrals
or Fredholm determinants were known for a long time for massive
integrable quantum field theories \cite{Smirnov86}, where they
had been obtained as solutions of functional equations
\cite{Smirnov92}.  They also appear in the context of the Fermionic
basis approach to the form factors of the Sine-Gordon model
\cite{JMS11b}. Thus, taken the close relationship between the
XXZ and Sine-Gordon models, the existence of a form factor
series for the correlation functions of the XXZ chain in
the antiferromagnetic massive regime with form factors
represented as finite determinants might have been anticipated.
We wish to stress, however, that we have used an eigenbasis
of the quantum transfer matrix rather than the Hamiltonian
eigenbasis. For this reason our form factor series comes with
a different interpretation of the particle content of the model,
which is in terms of particles and holes rather than in terms
of spinons. The latter are usually interpreted as a lattice
manifestation of the quantum solitons in the Sine-Gordon model.
This may mean that our representation is still structurally
different from everything we would expect to obtain by analogy
with the Sine-Gordon model. 

We recently learned about a thesis \cite{Kulkarni20} in which
a representation of the form factor amplitudes involving only
finite determinants was obtained in an eigenbasis of the
Hamiltonian of the XXX chain. It will be interesting to see,
if that representation can be brought to a more explicit form
and how it then compares with our result.

Due to the relative simplicity of our novel expressions, answers
to a number of longstanding questions appear now within reach
and new questions can be asked.
\begin{enumerate}
\item
We believe that the novel expressions are more efficient for
a numerical computation of the correlation functions.
\item
We hope that we will be able to prove the convergence of the
series employing methods that were recently developed in the
context of integrable massive quantum field theories
\cite{Kozlowski20app}.
\item
We hope that we can now attack the problem of calculating the
long-time, large-distance asymptotics of two-point correlations
in the XXX chain, which requires to estimate all terms in the
series as they all contribute.
\item
We think that the new series representation may make it possible
to prove our older conjecture \cite{DGKS16b} about the connection
of the form factor amplitudes in the spinon and particle-hole
bases.
\item
As the elements of the matrices $\Om$ and $\overline \Om$
can be entirely expressed in terms of basic hypergeometric
functions the question arises, if there are even more simple
and explicit expressions for the corresponding determinants
in the general case. The exploration of this question may 
take us deep into the theory of basic hypergeometric series
and may potentially yield new identities between basic
hypergeometric series or new derivations of some of the
known identities.
\end{enumerate}

In future work we would also like to gain a better understanding
of the generic finite temperature case and of the high-$T$
asymptotics. We would further like to explore the possibility
if similar series representations exist for the two-point functions
in the massless regime. This will require a better understanding
of the spectrum and general Bethe root patterns of the excited
states of the quantum transfer matrix of the XXZ chain.\\

{\bf Acknowledgments.}
We would like to thank Jesko Sirker for providing his DMRG data for
comparison and Ole Warnaar for helpful discussions about basic
hypergeometric series identities. CB and FG acknowledge financial
support by the DFG in the framework of the research unit FOR 2316.
The work of KKK is supported by the CNRS and by the ‘Projet international
de coop\'eration scientifique No. PICS07877’: \textit{Fonctions
de corr\'elations dynamiques dans la cha{\^\nodoti}ne XXZ \`a
temp\'erature finie}, Allemagne, 2018-2020. JS is grateful for
support by a JSPS Grant-in-Aid for Scientific Research (C)
No.\ 18K03452 and by a JSPS Grant-in-Aid for Scientific
Research (B) No.\ 18H01141.\\

{\bf Note added.} While the preparation of this manuscript had been
delayed for unforeseeable reasons, we found and proved a generalization
to arbitrary $\nex$ of the results of Section~\ref{sec:basehype} on the
explicit evaluation of the remaining finite determinants in terms
of basic hypergeometric series. In order to not further increase the
length of the present manuscript, and also for its independent
significance, we shall publish this additional result separately.\\

%{\bf Availability of data.} The data that support the findings of
%this study are available from the corresponding author upon request.

\clearpage

\renewcommand{\thesection}{\Alph{section}}
\renewcommand{\theequation}{\thesection.\arabic{equation}}

\begin{appendices}

\section{The generating function for the longitudinal
two-point function}
\setcounter{equation}{0}
\label{app:genfun}
Here we provide a derivation of equation (\ref{azzagen}) of the
main text. Consider two eigenstates $|l, h\>$ and $|n, h\>$ of
the dynamical quantum transfer matrix $t(\la|h)$ which are
non-degenerate for $l \ne n$. Then
\begin{multline}
    \6_{h'/2T} \<l, h| t(\la|h')|n, h'\>\bigl|_{h' = h} =
       \<l, h|\Si^z (\la|h)|n, h\> + \La_l (\la|h) \<l, h| \6_{h/2T}|n, h\> \\[1ex]
       = \de_{l, n} \<l, h|l, h\> \6_{h/2T} \La_n (\la|h)
         + \La_n (\la|h) \<l, h| \6_{h/2T}|n, h\> \epc
\end{multline}
implying that
\begin{multline} \label{szid1}
     \<l, h|\Si^z (\la|h)|n, h\> = \\[1ex]
        \de_{l, n} \<l, h|l, h\> \6_{h/2T} \La_l (\la|h)
	+ \bigl(\La_n (\la|h) - \La_l (\la|h)\bigr) \<l, h| \6_{h/2T}|n, h\> \epp
\end{multline}
Similarly,
\begin{multline} \label{szid2}
     \<n, h|\Si^z (\la|h)|l, h\> = \\[1ex]
        \de_{l, n} \<l, h|l, h\> \6_{h/2T} \La_l (\la|h)
	+ \bigl(\La_n (\la|h) - \La_l (\la|h)\bigr)
	  \bigl(\6_{h/2T} \<n, h|\bigr)|l, h\> \epp
\end{multline}
Setting $l = 0$, multiplying (\ref{szid1}) and (\ref{szid2})
and supplying the correct normalization we obtain
\begin{multline} \label{azzagen1}
     A_n^{zz} = \de_{n,0} \biggl(\frac{\6_{h/2T} \La (- \i \g/2|h)}
                                      {\La (- \i \g/2|h)} \biggr)^2 \\[1ex] +
                \frac{\<h|\bigl(\6_{h/2T} |n, h\>\bigr)\bigl(\6_{h/2T} \<n, h|\bigr)|h\>}
	             {\<h|h\>\<n, h|n, h\>}
                \biggl(
		   \frac{\La_n (- \i \g/2|h)}{\La (- \i \g/2|h)} - 2 +
		   \frac{\La (- \i \g/2|h)}{\La_n (- \i \g/2|h)} \biggr),
\end{multline}

Now let
\begin{equation}
     \varrho_n (h') = \frac{\<h|n, h'\>\<n, h'|h\>}{\<h|h\>\<n, h'|n, h'\>} \epc \qd
     \varsigma_n (h') = \frac{\La_n (- \i \g/2|h')}{\La (- \i \g/2|h)} - 2
                + \frac{\La (- \i \g/2|h)}{\La_n (- \i \g/2|h')} \epp
\end{equation}
It follows that
\begin{equation}
     \varrho_0 (h) = 1 \epc \qd
     \varsigma_0 (h) = \varsigma_0' (h) = 0 \epc \qd
     4 T^2 \varsigma_0'' (h) = 2 \biggl(\frac{\6_{h/2T} \La (- \i \g/2|h)}
                                     {\La(- \i \g/2|h)}\biggr)^2 \epc
\end{equation}
whereas for $n \ne 0$
\begin{equation}
     \varrho_n (h) = \varrho_n' (h) = 0 \epc \qd
     4 T^2 \varrho_n'' (h) = \frac{2 \<h|\bigl(\6_{h/2T} |n, h\>\bigr)
                             \bigl(\6_{h/2T} \<n, h|\bigr)|h\>}
			    {\<h|h\>\<n, h|n, h\>} \epc
\end{equation}
and $\varsigma_n (h) \ne 0$. Thus,
\begin{equation}
     2 T^2 \6_{h'}^2 \varrho_n (h') \varsigma_n (h') \bigr|_{h' = h} = A_n^{zz} (h)
\end{equation}
as a consequence of (\ref{azzagen1}), which is equivalent to (\ref{azzagen}).

\section{\boldmath Determinant of $g_0$}
\setcounter{equation}{0}
\label{app:ellipticcauchy}
In this section we shall derive a product representation of the
determinant
\begin{equation}
     \D_M (\xv, \yv) = \det_M \{g_0 (x_j, y_k)\} \epc
\end{equation}
where $\xv = (x_1, \dots, x_M)$, $\yv = (y_1, \dots, y_M)$. For later
convenience we also introduce the notation
\begin{equation}
     \xv_l = (x_1, \dots, x_{l - 1}, x_{l + 1}, \dots, x_M)
\end{equation}
and similarly for $\yv_m$.

Due to the multi-linearity of the determinant, the function $\D_M (\xv, \yv)$
inherits the pole structure and quasi-periodicity from $g_0$. In order
to understand the latter recall that
\begin{subequations}
\label{thetaperiods}
\begin{align} \label{theta1periods}
     \dh_1 (x + \p) & = - \dh_1 (x) \epc &&
     \dh_1 (x + \i \g) = - q^{-1} \re^{- 2 \i x} \dh_1 (x) \epc \\[1ex]
     \dh_2 (x + \p) & = - \dh_2 (x) \epc &&
     \dh_2 (x + \i \g) = q^{-1} \re^{- 2 \i x} \dh_2 (x) \epp
     \label{theta2periods}
\end{align}
\end{subequations}
It follows that
\begin{equation}
     g_0 (x + \p, y) = g_0 (x,y) \epc \qd
     g_0 (x + \i \g, y) = - g_0 (x,y) \epp
\end{equation}
Thus,
\begin{equation}
     \D_M (\xv + \p \ev_l, \yv) = \D_m (\xv, \yv) \epc \qd
     \D_M (\xv + \i \g \ev_l, \yv) = - \D_m (\xv, \yv) \epc
\end{equation}
where $\ev_l$ is the $l$th canonical unit vector.

We are looking for a generalization of the Cauchy-det formula. For
this reason we define
\begin{equation}
     \chi_M (\xv, \yv) =
     \frac{\prod_{1 \le j < k \le M} \dh_1 (x_j - x_k) \dh_1 (y_k - y_j)}
          {\prod_{j,k = 1}^M \dh_1 (y_j - x_k)} \epp
\end{equation}
Then, by (\ref{theta1periods}),
\begin{subequations}
\label{chiperiods}
\begin{align}
     \chi_M (\xv + \p \ev_l, \yv) & = - \chi_M (\xv, \yv) \epc \\[1ex]
     \chi_M (\xv + \i \g \ev_l, \yv) & = - q \re^{2 \i \sum_{j=1}^M (x_j - y_j)}
                                            \chi_M (\xv, \yv) \epp
\end{align}
\end{subequations}

We note that $\D_M$ is an elliptic function (with periods $\p$, $2 \i \g$)
in every $x_l$. As a function of $x_l$ it therefore has as many poles
as zeros per cell. The same is definitively not true for $\chi_M$.
$\chi_M$ cannot be elliptic as it has less zeros than poles per cell. Also
\begin{equation}
     \D_1 (x,y) = g_0 (x,y) = \frac{\dh_1' \dh_2 (y - x)}{\dh_2 \dh_1 (y - x)} \epc \qd
     \chi_1 (x,y) = \frac{1}{\dh_1 (y - x)} \epp
\end{equation}

This suggests to replace $\chi_M (\xv, \yv)$ by
\begin{equation}
     \PH_M (\xv, \yv) = \dh_2 \bigl( \tst{\sum_{j=1}^M (x_j - y_j)} \bigr)
                        \chi_M (\xv, \yv) \epc
\end{equation}
which for all $l = 1, \dots, M$ has the periodicity properties
\begin{equation}
     \PH_M (\xv + \p \ev_l, \yv) = \PH_M (\xv, \yv) \epc \qd
     \PH_M (\xv + \i \g \ev_l, \yv) = - \PH_M (\xv, \yv)
\end{equation}
following from (\ref{theta2periods}) and (\ref{chiperiods}).

We infer that the function
\begin{equation}
     {\cal F}_M (\xv, \yv) = \frac{\D_M (\xv, \yv)}{\PH_M (\xv, \yv)}
\end{equation}
is double periodic with periods $\p$, $\i \g$ and meromorphic in every
$x_l$, $l = 1, \dots, M$, hence an elliptic function of $x_l$
with periods $\p$, $\i \g$. ${\cal F}_M$ as a function of $x_l$ has
at most a single simple pole congruent to $\sum_{j=1}^M y_j
- \sum_{j=1, j \ne l}^M x_j - \p/2$ per cell. But such an elliptic
function must be a constant, since the sum of the residua of an
elliptic function at its poles in any cell is zero. It follows that
${\cal F}_M (\xv, \yv)$ is independent of $\xv$. Repeating all the
arguments for $\yv$ instead of $\xv$, we see that ${\cal F}_M
(\xv, \yv)$ is independent of $\yv$ as well, so it is a mere constant.

This constant can be calculated, e.g., by comparing the residua of
$\D_M$ and $\PH_M$ at $x_l = y_k$. This way we see that ${\cal F}_M
= (\dh_1')^M/\dh_2$. Thus,
\begin{equation} \label{detgzero}
     \D_M (\xv, \yv) =
	\frac{\dh_2 (\Si)}{\dh_2}
        \frac{\prod_{1 \le j < k \le M}
	      \bigl(\dh_1 (x_j - x_k)/\dh_1'\bigr)\bigl(\dh_1 (y_k - y_j)/\dh_1'\bigr)}
	     {\prod_{j,k = 1}^M \bigl(\dh_1 (y_j - x_k)/\dh_1'\bigr)} \epc
\end{equation}
where we have introduced the shorthand notation $\Si = \sum_{j=1}^M (x_j - y_j)$.

\end{appendices}

\bibliographystyle{amsplain}
\bibliography{hub}

\end{document}